\newcount\Comments
\def \fullversion {1}

\ifx \fullversion \undefined
\documentclass[envcountsame]{llncs}
\else
\documentclass[11pt]{article}
\usepackage{fullpage,paper}
\fi

\usepackage{dsfont,graphicx,color,hyperref,url,amsmath,amssymb}
\newcommand{\cut}[1]{}
\newcommand{\PDF}{\ensuremath{\mathsf{PDF}}}
\newcommand{\CDF}{\ensuremath{\mathsf{CDF}}}
\title{Privacy Games}
\author{Yiling Chen\thanks{Supported in part by NSF grant CCF-1301976.} \qquad Or Sheffet\thanks{Supported in part by NSF grant CNS-1237235.}  \qquad Salil Vadhan\thanks{Supported by NSF grant CNS-1237235, a gift from Google, Inc., and a Simons Investigator grant.}
\\
Center for Research on Computation and Society\\School of Engineering and Applied Sciences\\Harvard University\\ \texttt{\{yiling,osheffet,salil\}@seas.harvard.edu} 
}
\date{}

\newcommand{\erase}[1]{}
\newcommand{\kibitz}[2]{\ifnum\Comments=1\textcolor{#1}{#2}\fi}

\newcommand{\os}[1]{\kibitz{red}{\bf\noindent [OS: #1]}}

\newcommand{\type}{\ensuremath{t}}

\renewcommand{\Pr}{{\bf Pr}} 

\ifx \fullversion\undefined
\renewcommand{\paragraph}[1]{\vspace{1mm}\noindent\textit{#1}}
\fi

\begin{document}
\maketitle

\ifx \fullversion \undefined \vspace{-0.75cm} \fi
\begin{abstract}
The problem of analyzing the effect of privacy concerns on the behavior of selfish utility-maximizing agents has received much attention lately. 
Privacy concerns are often modeled by altering the utility functions of agents to consider also their privacy loss~\cite{Xiao13,GhoshR11,NissimOS12,ChenCKMV13}. 
Such privacy aware agents prefer to take a randomized strategy even in very simple games in which non-privacy aware agents play pure strategies. In some cases, the behavior of privacy aware agents follows the framework of Randomized Response, a well-known mechanism that preserves differential privacy.

Our work is aimed at better understanding the behavior of agents in settings where their privacy concerns are explicitly given. 
We consider a toy setting where agent $A$, in an attempt to discover the secret type of agent $B$, offers $B$ a gift that one type of $B$ agent likes and the other type dislikes. As opposed to previous works, $B$'s incentive to keep her type a secret isn't the result of ``hardwiring'' $B$'s utility function to consider privacy, but rather takes the form of a payment between $B$ and $A$. We investigate three different types of payment functions and analyze $B$'s behavior in each of the resulting games. As we show, under some payments, $B$'s behavior is very different than the behavior of agents with hardwired privacy concerns and might even be deterministic. Under a different payment we show that $B$'s BNE strategy does fall into the framework of Randomized Response.

\end{abstract}

\ifx \fullversion \undefined \vspace{-1.1cm}  \fi

\ifx \fullversion\undefined \else \newpage \setcounter{tocdepth}{2}\tableofcontents \newpage \fi
\section{Introduction}
\label{sec:intro}
\ifx \fullversion \undefined \vspace{-0.4cm} \fi

In recent years, as the subject of privacy becomes an increasing concern, many works have discussed the potential privacy concerns of economic utility-maximizing agents.
Obviously, utility-maximizing agents are worried about the effect of revealing personal information in the current game on future transactions, and wish to minimize potential future losses. In addition, some agents may simply care about what some outside observer, who takes no part in the current game, believes about them. Such agents would like to optimize the effect of their behavior in the current game on the beliefs of that outside observer. Yet specifying the exact way in which information might affect the agents' future payment or an outside observer's beliefs is a complicated and intricate task.

Differential privacy (DP), a mathematical model for privacy, developed for statistical data analysis~\cite{DworkMNS06,DworkKMMN06}, avoids the need for such intricate modeling by providing a worst-case bound on an agents' exposure to privacy-loss. Specifically, by using a $\epsilon$-differentially private mechanism, agents can guarantee that the belief of \emph{any} observer about them changes by no more than a multiplicative factor of $e^\epsilon\approx 1+\epsilon$ once this observer sees the outcome of the mechanism~\cite{Dwork06}
. Furthermore, as pointed out in~\cite{GhoshR11,NissimOS12}, using a $\epsilon$-differentially private mechanism the agents guarantee that, in expectation, \emph{any} future loss increases by no more than a factor of $e^\epsilon-1\approx \epsilon$. A recent line of work~\cite{Xiao13,GhoshR11,NissimOS12,ChenCKMV13} has used ideas from differential privacy to model and analyze the behavior of privacy-awareness in game-theoretic settings. The aforementioned features of DP allow these works to bypass the need to model future transactions. Instead, they model privacy aware agents as selfish agents with utility functions that are ``hardwired'' to trade off between two components: a (positive) reward from the outcome of the mechanism vs a (negative) loss from their non-private exposure. This loss can be upper-bounded using DP, and hence in some cases can be shown to be dominated by the reward (of carefully designed mechanisms), showing that privacy concerns don't affect an agent's behavior.


However, in other cases, the behavior of privacy-aware agents may differ drastically from the behavior of classical, non-privacy aware agents.
For example, consider a toy-game in which $B$ tells $A$ which of the two free gifts that $A$ offers (or \emph{coupons} as we call it, for reasons to be explained later) $B$ would like to receive. We characterize $B$ using one of two types, $0$ or $1$; where 
type $0$ prefers the first gift and type $1$ prefers the second one. (This is a rephrasing of the ``Rye or Wholewheat'' game discussed in~\cite{NissimOS12}.) Therefore it is simple to see that a non-privacy-aware agent always (deterministically) asks for the gift that matches her type. In contrast, if we model the privacy loss of a privacy-aware agent using DP as in the work of Ghosh and Roth~\cite{GhoshR11} (and the value of the coupon is large enough), a privacy-aware agent takes a randomized strategy. (See Section~\ref{subsubsec:privacy_aware_agents}.)  
Specifically, the agent plays \emph{Randomized Response}, a standard differentially private mechanism that outputs a random choice slightly biased towards the agent's favorable action. 

However, it was argued~\cite{NissimOS12,ChenCKMV13} that it is not realistic to use the worst-case model of DP to quantify the agent's privacy loss and predict her behavior. Differential privacy should only serve as an \emph{upper bound} on the privacy loss, whereas the agent's expected privacy loss can (and should in fact) be much smaller --- depending on the agent's predictions regarding future events,  adversary's prior belief about her, the types and strategies of other agents, and the random choices of the mechanism and of other agents. As discussed above, these can be hard to model, so it is tempting to use a worst-case model like differential privacy.

But what happens if we can formulate the agent's future transactions? What if we know that the agent is concerned with the belief of a specific adversary, and we can quantify the effects of changes to that belief? Is the behavior of a classical selfish agent in that case well-modeled by such a ``DP-hardwired'' privacy-aware agent? Will she even randomize her strategy? In other words, we ask:
\begin{center}
\begin{minipage}[c]{0.9\textwidth}
\centering
\textit{What is the behavior of a selfish utility-maximizing agent in  a setting with clear privacy costs?}
\end{minipage}
\end{center}
More specifically, we ask whether we can take the above-mentioned toy-game and alter it by introducing payments between $A$ and $B$ such that the behavior of a privacy-aware agent in the toy-game matches the behavior of classical (non-privacy aware) agent in the altered game. In particular, in case $B$ takes a randomized strategy --- does her behavior preserve $\epsilon$-differential privacy, and for what value of $\epsilon$? The study of these questions may also provide insights relevant for traditional, non-game-theoretic uses of differential privacy --- helping us understand how tightly differential privacy addresses the concerns of data subjects, and thus providing guidance in the setting of the privacy parameter $\epsilon$ or the use of alternative, non-worst-case variants of differential privacy (such as~\cite{BassilyGKS13}).

\paragraph{Our model.} In this work we consider multiple games that model an interaction between an agent which has a secret type and an adversary whose goal is to discover this type. Though the games vary in the resulting behavior of the agents, they all follow a common outline which is similar to the toy game mentioned above. Agent $A$ offers $B$ a free coupon, that comes in one of two types $\{0,1\}$. Agent $B$ has a secret type $\type\in \{0,1\}$ chosen from a known prior $(D_0,D_1)$, such that a type-$\type$ agent  has positive utility $\rho_t$ for type-$\type$ coupon and zero utility for a type-$(1-\type)$ coupon. And so the game starts with $B$ sending $A$ a signal $\hat\type$ indicating the requested type of coupon. (Formally, $B$'s utility for the coupon is $\rho_{\type} \mathds{1}_{[\hat\type=\type]}$ for some parameters $\rho_0, \rho_1$.) Following this interaction, agent $C$, who viewed the signal $\hat\type$ that $B$ sent, challenges $B$ into a game --- with $C$ taking action $\tilde\type$ and incurring a payment from $B$ of  $P(\tilde\type,\type)$. To avoid the need to introduce a third party into the game, we identify $C$ with $A$.\footnote{Hence the reason for the name ``The Coupon Game''. We think of $A$ as $G$ -- an ``evil'' car-insurance company that offers its client a coupon either for an eyewear store or for a car race; thereby increasing the client's insurance premium based on either the client's bad eyesight or the client's fondness for speedy and reckless driving.} 
Figure~\ref{fig:game_outline} gives a schematic representation of the game's outline. 

We make a few observations of the above interaction. 
We aim to model a scenario where $B$ has the most incentive to hide her true type whereas $A$ has the most incentive to discover $B$'s type. Therefore, all of the payments we consider have the property that if $B$'s type is $\type^*$ then $t^* = \arg\max_{\tilde \type}P(\tilde\type,\type^*)$. Furthermore, the game is modeled so that the payments are transferred from $B$ to $A$, which makes $A$'s and $B$'s goals as opposite as possible. (In fact, past the stage where $B$ sends a signal $\hat\type$, we have that $A$ and $B$ plays a zero-sum game.)
We also note that $A$ and $B$ play a Bayesian game (in extensive form) as $A$ doesn't know the private type of $B$, only its prior distribution. We characterize Bayesian Nash Equilibria (BNE) in this paper and will show that in each game, the BNE is unique except when parameters of the game satisfy certain equality constraints. It is not difficult to show that the strategies at every BNE of our games are part of a Perfect Bayesian Equilibrium (PBE), i.e. a subgame-perfect refinement of the BNE. However, we focus on BNE in this paper as the equilibrium refinement doesn't bring any additional insight to our problem.
\erase{ 
Third, following classical game theory, in each game we characterize the Nash Equilbirium (BNE). As we show, in each game the BNE is unique, unless that parameters of the game are set in  a way that satisfy certain equality constraints. Therefore, other solution concepts like subgame-perfect equilibrium and its Bayesian equivalent (should they exist in the game), must also be the same unique BNE Furthermore, as the prior on $B$'s type is known to all agents, we have that this BNE must all be perfect Bayes equilibrium. \os{Yiling - I think this is true and fairly trivial, am I right?} }

\begin{figure}[t]
\centering \includegraphics[scale=0.35]{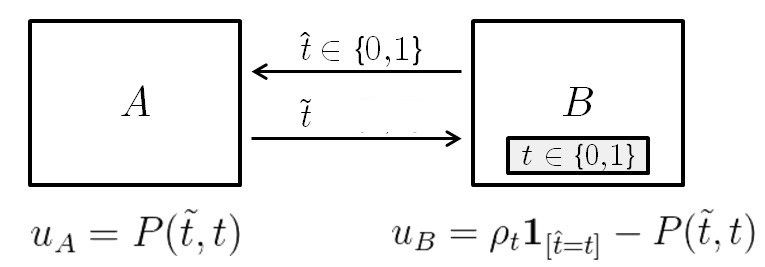}
\caption{A schematic view of the privacy game we model.\label{fig:game_outline}}
\ifx \fullversion \undefined \vspace{-0.5cm} \fi
\end{figure}

\paragraph{Our results and paper organization.} First, in Section~\ref{sec:preliminaries}, following preliminaries we discuss the DP-hardwired privacy-aware agent as defined by Ghosh and Roth~\cite{GhoshR11} and analyze her behavior in our toy game. Our analysis shows that given sufficiently large coupon valuations $\rho_t$, both types of $B$ agent indeed play Randomized Response. We also discuss conditions under which other models of DP-hardwired privacy-aware agents play a randomized strategy.

Following preliminaries, we consider three different games. These games follow the general coupon-game outline, yet they vary in their payment function. The discussion for each of the games follows a similar outline. We introduce the game, then analyze the two agents' BNE strategies and see if the strategy of the $B$ agent is indeed randomized or pure (and in case it is randomized --- whether or not it follows Randomized Response for some value of $\epsilon$). We also compare the coupon game to a ``benchmark game'' where $B$ takes no action and $A$ guesses $B$'s type without any signal from $B$. Investigating whether it is even worth while for $A$ to offer such a coupon, we compare $A$'s profit between the two games.\footnote{The benchmark game is not to be confused with the toy-game we discussed earlier in this introduction. In the toy game, $A$ takes no action and $B$ decides on a signal. In the benchmark game, $B$ takes no action and $A$ decides which action to take based on the specific payment function we consider in each game.}  The payment functions we consider are the following.
\begin{enumerate}
\item In Section~\ref{sec:scoring_rule} we consider the case where the payment function is given by a \emph{proper scoring rule}. Proper scoring rules allows us to quantify the $B$'s cost to any change in $A$'s belief about her type. We show that in the case of symmetric scoring rules (scoring rules that are invariant to relabeling of event outcomes) both types of $B$ agent follow a randomized strategy that causes $A$'s posterior belief on the types to resemble Randomized Response. That is, initially $A$'s belief on $B$ being of type-$0$ (resp. type-$1$) is $D_0$ (resp. $D_1$); but $B$ plays in a way such that after viewing the $\hat\type=0$ signal, $A$'s belief that $B$ is of type-$0$ (resp. type-$1$) is $\tfrac {1 + \epsilon} 2$ (resp. $\tfrac {1-\epsilon}2$) for some value of $\epsilon$ (and vice-versa in the case of the $\hat\type=1$ signal with the same $\epsilon$).
\item In Section~\ref{sec:matching_pennies} we consider the case where the payments between $A$ and $B$ are the result of $A$ guessing correctly $B$'s type. $A$ views the signal $\hat\type$ and then guesses a type $\tilde\type\in\bits$ and receives a payment of $\mathds{1}_{[\tilde\type=\type]}$ from $B$. This payment models the following viewpoint of $B$'s future losses: there is a constant gap (of one ``unit of utility'') between interacting with an agent that knows $B$'s type to an agent that does not know her type. We show that in this case, if the coupon valuations are fixed as $\rho_0$ and $\rho_1$, then at least one type of $B$ agent plays deterministically. However, if $B$'s valuation for the coupon is sampled from a continuous distribution, then $A$'s strategy effectively dictates a threshold with the following property: any $B$ agent whose valuation for the coupon is below the threshold lies and signals $\hat\type = 1-\type$, and any agent whose valuation is above the threshold signals truthfully $\hat\type = \type$.
Hence, an $A$ agent who does not know $B$'s valuation thinks of $B$ as following a randomized strategy.
\item In Section~\ref{sec:opt-out-possible} we consider a variation of the previous game where $A$ also has the option to opt out and not challenge $B$ into a payment game --- to report~$\bot$ and in return get no payment (i.e., $P(\bot,\type)=0$). We show that in such a game, under a very specific setting of parameters, the only BNE is such where both types of $B$ agent take a randomized strategy. Under alternative settings of the game's parameters, the strategy of $B$ is such that at least one of the two types plays deterministically.
\end{enumerate}
\ifx \cameraready \undefined
Conclusions and future directions appear in Section~\ref{sec:conclusions}, where we provide a discussion of our results. 
\else
Future directions are deferred to the full version of the paper, due to space limitation.
\fi
We find it surprising to see how minor changes to the privacy payments lead to diametrically different behaviors. In particular, we see the existence of a threshold phenomena. Under certain parameter settings in the game we consider in item 3 above,
we have that if the value of the coupon is above a certain threshold then at least one of the two types of $B$ agent plays deterministically; and if the value of the coupon is below this threshold, $B$ randomizes her behavior s.t. $\hat\type=\type$ w.p. close to $\tfrac 1 2$.
\ifx \fullversion \undefined \vspace{-0.6cm}  \fi
\subsection{Related Work}
\label{subsec:related_work}
\ifx \fullversion \undefined \vspace{-0.2cm}  \fi
The study of the intersection between mechanism design and differential privacy began with the seminal work of McSherry and Talwar~\cite{McsherryT07}, who showed that an $\epsilon$-differentially private mechanism is also $\epsilon$-truthful. The first attempt at defining a privacy-aware agent was of Ghosh and Roth~\cite{GhoshR11} who quantified the privacy loss using a linear approximation $v_i\cdot \epsilon$ where $v_i$ is an individual parameter and $\epsilon$ is the level of differential privacy that a mechanism preserves. Other applications of differentially privacy mechanisms in game theoretic settings were studied by Nissim et al~\cite{NissimST12}. The work of Xiao~\cite{Xiao13} initiated the study of mechanisms that are truthful even when you incorporate the privacy loss into the agents' utility functions. Xiao's original privacy loss measure was the mutual information between the mechanism's output and the agent's type. Nissim et al~\cite{NissimOS12} (who effectively proposed a preliminary version of our coupon game called ``Rye or Wholewheat'') generalized the models of privacy loss to only assume that it is \emph{upper bounded} by $v_i \cdot \epsilon$. Chen et al~\cite{ChenCKMV13} proposed a refinement where the privacy loss is measured with respect to the given input and output.
 Fleischer and Lyu~\cite{FleischerL12} considered the original model of agents as in Ghosh and Roth~\cite{GhoshR11} but under the assumption that $v_i$, the value of the privacy parameter of each agent, is sampled from a known distribution.

Several papers in economics look at the potential loss of agents from having their personal data revealed. In fact, one folklore objection to the Vickrey auction is that in a repeated setting, by providing the sellers with the bidders' true valuations for the item, the bidders subject themselves to future loss should the seller prefer to run a reserved-price mechanism in the future.
In the context of repeated interaction between an agent and a company, there have been works~\cite{ConitzerTW12,BergemannBM13} studying the effect of price differentiation based on an agent allowing the company to remember whether she purchased the same item in the past. Interestingly, strategic agents realize this effect and so they might ``haggle'' --- reject a price below their valuation for the item in round $1$ so that they'd be able to get even lower price in round $2$. In that sense, the fact that the agents publish their past interaction with the company  actually helps the agents. Other work~\cite{CalzolariP06} discusses a setting where a buyer sequentially interacts with two different sellers, and characterizes the conditions under which the first seller prefers not to give the buyer's information to the second seller. Concurrently with our work, Gradwohl and Smorodinsky~\cite{GradwohlS14}, whose motivation is to analyze the effect of privacy concerns, introduce a framework of games in which an agent's utility is affected by both her actions and how her actions are perceived by a third party.

The privacy games that we propose and analyze in this paper fall into the class of signaling games~\cite{MWG95}, where a sender ($B$ in our game) with a private type sends a message (i.e. a signal) to a receiver ($A$ in our game) who then takes an action. The payoffs of both players depend on the sender's message, the receiver's action, and the sender's type. Signaling games have been widely used in modeling behavior in economics and biology. The focus is typically on understanding when signaling is informative, i.e. when the message of the sender allows the receiver to infer the sender's private type with certainty, especially in settings when signaling is costly (e.g. Spence's job market signaling game~\cite{spence73}). In our setting, however, informative signaling violates privacy. We are interested in characterizing when the sender plays in a way such that the receiver cannot infer her type deterministically.

\ifx \fullversion \undefined \vspace{-0.5cm}  \fi
\ifx \fullversion \undefined \else \newpage \fi
\section{Preliminaries}
\label{sec:preliminaries}
\ifx \fullversion \undefined \vspace{-0.4cm}  \fi
\subsection{Equilibrium Concept}
\ifx \fullversion \undefined \vspace{-0.3cm} \fi

We model the games between $A$ and $B$ as Bayesian extensive-form games. However, instead of using the standard Perfect Bayesian Equilibrium (PBE), which is a refinement of Bayesian Nash Equilibrium (BNE)  for extensive-form games, as our solution concept, we analyze BNE for our games. It can be shown that all of the BNEs considered in our paper can be ``extended'' to PBEs (by appropriately defining the beliefs of agent A about agent B at all points in the game). We thus avoid defining the more subtle concept of PBE as the refinement doesn't provide additional insights for our problem. Below we define BNE.   

A \emph{Bayesian} game between two agents $A$ and $B$ is specified by their type spaces $(\Gamma_A, \Gamma_B)$, a prior distribution $\Pi$ over the type spaces (according to which nature draws the private types of the agents), sets of available actions $(C_A, C_B)$, and utility functions, $u_i : \Gamma_A \times \Gamma_B \times C_A \times C_B \to \mathbb{R}$, $i \in \{A, B\}$.  
A \emph{mixed} or \emph{randomized} strategy of agent $i$ maps a type of agent $i$ to a distribution over her available actions, i.e. $\sigma_i: \Gamma_i \to \Delta(C_i)$, where $\Delta(C_i)$ is the probability simplex over $C_i$.  
\ifx \cameraready \undefined
When $\sigma_i$ deterministically maps a type to an action, it is called a \emph{pure strategy}. The Bayesian Nash Equilibrium (BNE) of the two-player game is defined as follows. 
\fi
\begin{definition}
\label{def:NE}
A strategy profile $(\sigma_A, \sigma_B)$ is a \emph{Bayesian Nash Equilibrium} if 
\[
\E [ u_i (T_i, T_{-i}, \sigma_i(T_i), \sigma_{-i}(T_{-i})) | T_i=t_i] \geq \E [u_i (T_i, T_{-i}, \sigma'_i(T_i), \sigma_{-i}(T_{-i})) | T_i=t_i] 
\]
for all $i \in \{A, B\}$, all types $t_i \in \Gamma_i$ occurring with positive probability, and all strategies $\sigma'_i$, where $\sigma_{-i}$ and $T_{-i}$ denote the strategy and type of the other agent respectively and the expectation is taken over the randomness of agent type $T_{-i}$ and the randomness of the strategies, $\sigma_i$, $\sigma_{-i}$ and $\sigma'_i$.
\end{definition}
\ifx \cameraready \undefined
In other words, a strategy profile $(\sigma_A, \sigma_B)$ is a BNE if both agents maximize their expected utility by playing $\sigma_i$ in responding to the other player's strategy $\sigma_{-i}$, i.e. they both play \emph{best response}. 


As mentioned in Section \ref{subsec:related_work}, our games between $A$ and $B$ belong to the class of signaling games. For signaling games, the terms {\em separating equilibrium} and {\em pooling equilibrium} are often used to characterize when signaling is fully informative. At a separating equilibrium, a player's strategy allows the other player to deterministically infer her private type, while at a pooling equilibrium multiple types of a player may take the same action, preventing the other player to infer her type with certainty. 
\fi

\ifx \fullversion \undefined \vspace{-0.8cm} \fi
\subsection{Differential Privacy}
\label{subsec:DP}
\ifx \fullversion \undefined \vspace{-0.4cm} \fi
In order to define differential privacy, we first need to define the notion of neighboring inputs. Inputs are elements in $\mathcal{X}^n$ for some set $\mathcal{X}$, and two inputs $\mathcal{I},\mathcal{I}'\in \mathcal{X}^n$ are called neighbors if the two are identical on the details of all individuals (all coordinates) except for at most one.

\begin{definition}[\cite{DworkMNS06}]
\label{def:privacy}
An algorithm $\textsf{ALG}$ which maps inputs into some range $\mathcal{R}$ satisfies \emph{$\epsilon$-differential privacy} if for all pairs of neighboring inputs $\mathcal{I},\mathcal{I}'$ and for all subsets $\mathcal{S}\subset\mathcal{R}$ it holds that
$\Pr[\mathsf{ALG}(\mathcal{I}) \in \mathcal{S}] \leq e^\epsilon \Pr[\mathsf{ALG}(\mathcal{I'}) \in \mathcal{S}]$.
\end{definition}
One of the simplest algorithms that achieve $\epsilon$-differential privacy is called \emph{Randomized Response}~\cite{KasiviswanathanLNRS08,dwork2010differential}, which dates back to the 60s~\cite{Warner65}. This algorithm is best illustrated over a binary input, where each individual is represented by a single binary bit (therefore a neighboring instance is a neighbor in which one individual is represented by a different bit), Randomized Response works by perturbing the input. For each individual $i$ represented by the bit $b_i$, the algorithm randomly and independently picks a bit $\hat b_i$ s.t. $\Pr[\hat b_i = b_i] = \tfrac {1+\epsilon} 2$ for some $\epsilon \in [0,1)$.
It follows from the definition of the algorithm that it satisfies $\ln(\tfrac{1+\epsilon}{1-\epsilon}) \approx 2\epsilon$-differential privacy. Randomized Response is sometimes presented as a distributed algorithm, where each individual randomly picks $\hat b_i$ locally, and reports $\hat b_i$ publicly. 
Therefore, it is possible to view this work as an investigation of the type of games in which selfish utility-maximizing agents truthfully follow Randomized Response, rather than sending some arbitrary bit as $\hat b_i$.

In this work, we define certain games and analyze the behavior of the two types of $B$ agent in the BNE of these games. And so, denoting $B$'s strategy as $\sigma_B$, we consider the implicit algorithm $\sigma_B(t)$ that tells a type-$\type$ agent what probability mass to put on the $0$-signal and on the $1$-signal. Knowing $B$'s strategy $\sigma_B$, we say that $B$ satisfies $\ln(X_{\rm game})$-differential privacy where\footnote{We use the convention $\tfrac 0 0 =1$.}
\[  X_{\rm game}\stackrel{\rm def} = X_{\rm game}(\sigma_B)=\max_{\type,\hat\type\in \{0,1\}} \left(   \frac{\Pr[\sigma_B(t)=\hat \type]} {\Pr[\sigma_B(1-t)= \hat \type ]}  \right) \]
We are interested in finding settings where $X_{\rm game}(\sigma_B^*)$ is finite, where $\sigma_B^*$ denotes $B$'s BNE strategy. We say $B$ \emph{plays a Randomized Response strategy} in a game whenever her BNE strategy $\sigma_B^*$ satisfies $\Pr[\sigma_B^*(0)=0] = \Pr[\sigma_B^*(1)=1]=p$ for some $p\in [1/2,1)$.

\subsubsection{Privacy-Aware Agents.}
\label{subsubsec:privacy_aware_agents}
The notion of privacy-aware agents has been developed through a series of works~\cite{Xiao13,GhoshR11,NissimOS12,ChenCKMV13}. 
The utility function of our privacy-aware agent $B$ is of the form $u_B = u_B^{out}-u_B^{priv}$. The first term, $u_B^{out}$ is the utility of agent $B$ from the mechanism.
The second term, $u_B^{priv}$, represents the agent's privacy loss. The exact definition of $u_B^{priv}$ (and even the variables $u_B^{priv}$ depends on) varies between the different works mentioned above, but all works bound the privacy-loss of an agent that interacts with a mechanism that satisfies $\epsilon$-differential privacy by  $u_B^{priv} \leq v\cdot \epsilon$ for some $v>0$. Here we argue about the behavior of a privacy-aware agent with the maximal privacy loss function, which is the type of agent considered by Ghosh and Roth~\cite{GhoshR11} (i.e., the agent's privacy loss when interacting with a mechanism that satisfies $\epsilon$-differential privacy is exactly $v\cdot\epsilon$ for some $v>0$).

Recall our toy game: $B$ sends a signal $\hat\type$ and gets a coupon of type $\hat\type$. Therefore, the outcome of this simple game is $\hat\type$, precisely the action that $B$ takes. $B$'s type is picked randomly to be $0$ w.p. $D_0$ and $1$ w.p. $D_1$, and a $B$ agent of type $\type$ has valuation of $\rho_t$ for a coupon of type $\type$. Therefore, in this game $u_B^{out} = \rho_t\mathds{1}_{[\hat \type = \type]}$. The mechanism we consider is $\sigma_B^*$, $B$'s utlity-maximizing strategy, which we think of as the implicit algorithm that tells a type-$\type$ agents what probability mass to put on sending the $\hat\type=0$ signal and what mass to put on the $\hat\type=1$ signal. As noted above, this strategy satisfies $\ln(X_{\rm game})$-differential privacy, and so $u_B^{priv}(\sigma_B^*) = v\cdot\ln(X_{\rm game})$ for some parameter $v>0$. 
Assuming $D_0\rho_0\neq D_1\rho_1$, our proof shows that this privacy-aware agent chooses essentially between two alternatives in our toy game: either both types take the same deterministic strategy and send the same signal ($\Pr[\sigma_B^*(0)=b]=\Pr[\sigma_B^*(1)=b]=1$ for some $b\in\bits$); or the agent randomizes her behavior and plays using Randomized Response: $\Pr[\sigma_B^*(0)=0] = \Pr[\sigma_B^*(1)=1] \in [\tfrac 1 2,1)$. We show that for sufficiently large values of the coupon the latter alternative is better than the first. 

\newcommand{\behaviorPrivacyAware}{
Let $B$ be a privacy-aware agent, whose privacy loss is given by $v\ln(X_{\rm game})$ for some $v>0$. Assume that there exists an $\alpha>0$ s.t. for sufficiently large values of $\rho_0, \rho_1$ it holds that $\min\{\rho_0,\rho_1\} \geq  \alpha\cdot(\rho_0+\rho_1)$. Then, the unique strategy $\sigma_B^*$ that maximizes $B$'s utility is randomized and satisfies: $\Pr[\sigma_B^*(0)=0] = \Pr[\sigma_B^*(1)=1]=p^*$ for some $p^*\in[\tfrac 1 2,1)$.
}
\begin{theorem}
\label{thm:behavior_privacy_aware}
\behaviorPrivacyAware
\end{theorem}
\newcommand{\PAA}{
\begin{proof}
Recall, the type of $B$ is chosen randomly to be $0$ w.p. $D_0$ and $1$ w.p. $D_1$. Given a strategy $\sigma_B$ for $B$, we denote $p=\Pr[\sigma_B(0)=0]$ and $q=\Pr[\sigma_B(1)=1]$ (so $\Pr[\sigma_B(0)=1] = 1-p$ and $\Pr[\sigma_B(1) = 0] = 1-q$). Therefore, \[X_{\rm game}(\sigma_B) = X_{\rm game}(p,q) = \max \left\{  \frac {p}{1-q}, \frac {1-q}{p}, \frac {q}{1-p}, \frac {1-p}{q} \right\}\]
Note that $X_{\rm game}(p,q) \geq 1$, with equality iff $p=1-q$ (which means $\sigma_B(t)$ is independent of $t$ and $B$ reveals no information about her type).
And so $B$ aims to maximizes the following utility function: $u_B = D_0\rho_0 p + D_1\rho_1 q - v\ln(X_{\rm game})$.
When the strategy that optimizes $B$'s utility, denoted $(p^*,q^*)$, satisfies $p^*=q^* = 1/2 +\epsilon$ for some $\epsilon \in [0,\tfrac 1 2)$ then we say that $B$ plays using Randomized Response. 

First, observe that if $p+q<1$ then $X_{\rm game} > 1$ and the utility of $B$ is $D_0\rho_0 p + D_1\rho_1 q - v\ln(X_{\rm game}) \leq D_0\rho_0p+D_1\rho_1 q$, so $B$ can always improve the utility by replacing setting either $(p,q)=(1,0)$ or $(p,q) = (0,1)$. The same argument holds for any $(p,q)$ where $p+q=1$ and both are not integral. (If $D_0\rho_0=D_1\rho_1$ then the agent in indifferent between any $(p,q)$ satisfying $p=1-q$.) Secondly, observe that the maximum cannot be obtained for $(p=1,q>0)$ or $(p>0,q=1)$, because in that case $X_{\rm game}$ shoots to infinity, so the privacy loss is infinite. Therefore, if there exist a strategy $(p,q)$ s.t. $p>1-q$ and $p,q\in(0,1)$ whose utility is strictly greater than $\max\{ D_0\rho_0, D_1\rho_1\}$, then it is a utility maximizing strategy. (Otherwise, one of the two strategies $(0,1)$ or $(1,0)$ maximizes $B$'s utility.)

Suppose that the maximum is obtained on some $(p^*,q^*)$ with $p^*+q^*>1$ and $q^*>p^*$. This means that $X_{\rm game} = \tfrac p {1-q}>1$. For any $(p,q)$ in a small enough neighborhood of $(p^*,q^*)$ we can differentiate $u_B$ and it holds that \[ D_0\rho_0 - \frac v {X_{\rm game}(p^*,q^*)} \left(\tfrac {\partial}{\partial p} X_{\rm game}(p^*,q^*)\right) = 0~,~~D_1\rho_1 - \frac v {X_{\rm game}(p^*,q^*)} \left(\tfrac {\partial}{\partial q} X_{\rm game}(p^*,q^*)\right) = 0\]
with $\tfrac {\partial}{\partial p} X_{\rm game} = \tfrac 1 {1-q}$ and $\tfrac {\partial}{\partial q} X_{\rm game} = \tfrac {p}{(1-q)^2}$, we have $\tfrac {D_1\rho_1}{D_0\rho_0} = \tfrac {p^*}{1-q^*}$, and so $D_0\rho_0 p^* + D_1 \rho_1 q^* = D_1\rho_1$. Denote $X^*_{\rm game} \stackrel{\rm def}= X_{\rm game}(p^*,q^*) = \frac {p^*} {1-q^*}$, and deduce that in this case the maximal utility is $D_0\rho_0 p^* + D_1 \rho_1 q^*-v\ln(X^*_{\rm game})=D_1\rho_1 - v\ln\left(\tfrac{D_1\rho_1}{D_0\rho_0}\right) < D_1\rho_1 \leq \max\{D_0\rho_0, D_1\rho_1\}$. Hence $B$ is still better off playing either $(0,1)$ or $(1,0)$. The case with $p^*>q^*$ (or equivalently $q^*(1-q^*)>p^*(1-p^*)$) is symmetric, and so $B$ prefers playing $(0,1)$ or $(1,0)$.

It remains to check the case of $p^*=q^*$, with $p^*+q^*=2p^*>1$. In this case we have $X_{\rm game} = \tfrac {p^*}{1-p^*}$, and the utility function of $B$ is the univariate function $(D_0\rho_0+D_1\rho_1)p+ v\ln\left(\tfrac p {1-p}\right)$. Setting the derivative $u_B'(p^*) = 0$ we have $D_0\rho_0 + D_1\rho_1 = \tfrac {v}{p^*(1-p^*)}$, or $p^* = \tfrac 1 2 \left(1 + \sqrt{1-\tfrac {4v}{D_0\rho_0 + D_1\rho_1}}\right)$. Denoting $Y = D_0\rho_0+D_1\rho_1$,  we now use the assumption that $\rho_0,\rho_1 = \Omega(\rho_0+\rho_1)$ and observe that $\max\{D_0\rho_0,D_1\rho_1\} = (1-c)Y$ for some constant $c>0$. Therefore, $B$ prefers playing this randomized strategy if 
\[u_B(p^*) = \tfrac 12 Y \left(1 + \sqrt{1-\tfrac{4v}Y}\right) - v \ln\left( \frac {1+\sqrt{1-\frac {4v}Y}} {1-\sqrt{1-\frac {4v}Y}} \right) > (1-c)Y\] Since $\lim_{Y\to\infty} \tfrac {u_B(p^*)}{Y} = 1$ then for a large enough value of $Y$, the above inequality holds.
\end{proof}

As an immediate corollary of the proof, consider any alternative definition of a privacy aware agent in which the privacy valuation $u_B^{priv}$ (i) depends only on the strategy $\sigma_B$, is (ii) non-negative, (iii) upper bounded by $v\ln(X_{\rm game})$ for some $v>0$ and (iv) $u_B^{priv}=\infty$ whenever $X_{\rm game} = \infty$. We argue that the utility maximizing strategy of such an agent is also randomized. (Observe that we no longer guarantee that $B$'s optimal strategy $\sigma_B^*$ satisfies $\Pr[\sigma_B^*(0)=0]=\Pr[\sigma_B^*(1)=1]$.)

To see that, observe that whenever $p=1-q$ we have that $X_{\rm game}=0$ so the privacy loss of an agent is $0$. Therefore, playing either $(p,q)=(1,0)$ or $(0,1)$, the agent can guarantee a utility of $\max\{D_0\rho_0,D_1\rho_1\}$. In contrast, should the agent play any $(p,q)$ with $p<1-q$, then her utility is upper bounded by $D_0p+D_1q \leq \max\{D_0\rho_0,D_1\rho_1\}$, because the privacy loss is non-negative. Therefore, the agent prefers playing $(p,q)=(1,0)$ or $(0,1)$ to any $(p,q)$ with $p<1-q$. Secondly, since we assume infinite privacy loss whenever $X_{\rm game}$, then $B$ utility maximizing strategy cannot satisfy that $p=1$ and $q > 0$ (or vice-versa). Lastly, the proof of Theorem~\ref{thm:behavior_privacy_aware} gives a strategy  $(p,q)$ with $p>1-q$ where the lower bound on $B$'s utility is greater than $\max\{D_0\rho_0,D_1\rho_1\}$. It follows that $B$ strictly prefers playing some strategy $(p,q)$ with $p,q\in (0,1)$ over playing $(p,q)=(1,0)$ or $(p,q)=(0,1)$.

\ifx \fullversion\undefined
\paragraph{Two types of $B$ agent as different players.} 
\else
\subsubsection{The two types of $B$ agent as different players.} 
\fi
The above analysis assumed $B$ is an agent playing this coupon game, decides on a strategy before the realization of her type, and sticks to that strategy even after her type is revealed to her. It is possible though to think of the two types of $B$ agents as two different agents ex-post -- after each agent is revealed her own type. As we show, the analysis in this case is slightly different. Observe that in this case we discuss a straight-forward Nash-equilibrium, as both agents know their respective types. In the following, we continue using our notation from  earlier, where $\Pr[\sigma(i) = \hat\type]$ denotes the probability a $B$ agent of type $\type=i$ sends the signal $\hat\type$ according to strategy $\sigma$.

\begin{theorem}
\label{thm:two_players_types_of_B_agents}
Consider the $2$-player game where player $i\in\{0,1\}$ is a type $\type=i$ $B$ agent. Assume $\rho_0=\rho_1=\rho$ and that $\rho$ is sufficiently large. Then there exists some $z^* \in (0,\tfrac v \rho)$ s.t. any NE of the game falls into one of three categories
\begin{itemize}
\item $\Pr[\sigma(0)=0]=\Pr[\sigma(1)=0]=1-z$ for some $z\in [0,z^*]$. (Both agents take the same strategy and send the signal $\hat\type=0$ with high probability $1-z$.)
\item $\Pr[\sigma(0)=0]=\Pr[\sigma(1)=0]=z$ for some $z\in [1-z^*,1]$. (Both agents take the same strategy and send the signal $\hat\type=1$ with high probability $1-z$.)
\item $\Pr[\sigma(0)=0]=\Pr[\sigma(1)=1]=z$ for some $z\in [1-\tfrac v \rho, z^*]$. (Both agents play randomized response and report truthfully $\hat\type=\type$ with the same probability $z$.)
\end{itemize}
\end{theorem}
\begin{proof}
We continue using the same notation from Theorem~\ref{thm:behavior_privacy_aware}: $p=\Pr[\sigma(0)=0]$ and $q=\Pr[\sigma(1)=1]$, and so $X_{\rm game}=X_{\rm game}(p,q)$ as denoted in the proof of Theorem~\ref{thm:behavior_privacy_aware}. In particular, when $p+q\geq 1$ it holds that $X_{\rm game}(p,q) = \tfrac p {1-q}$ when $p\leq q$, and $X_{\rm game}(p,q) = \tfrac q {1-p}$ when $p \geq q$.

First of all, observe that the utilities of both agents are symmetric: $u_{B,0}(p,q) = \rho p - v\ln(X_{\rm game}(p,q))$ and $u_{B,1}(p,q) = \rho q - v\ln(X_{\rm game}(p,q))$. Secondly, observe that if one agent plays determinisitcally $\Pr[\sigma(\type)=\hat\type]=1$ then unless the other type deterministically sends the same signal, then $X_{\rm game}(p,q)=\infty$ causing both agents to have utility of $-\infty$. It is therefore clear that the strategies $(p,q)=(1,0)$ and $(p,q)=(0,1)$ are both NEs.

To find the remaining NEs of the game, we fix a certain strategy for the $\type=1$ agent, denoted $z = \Pr[\sigma(1)=1]$, and see what is the strategy $x=\Pr[\sigma(0)=0]$ that type $\type=0$ agent prefers deviating to. Since both agents are symmetric, then our analysis also translates to an analysis of type $\type=1$. Before continuing with our analysis, we point out to the following two functions.
\begin{itemize}
\item Fix $z$ and denote $f(x)=\rho x - v\ln(\tfrac {z} {1-x})$. Since $f'(x) = \rho - \tfrac v {1-x}$ is decreasing on the interval $[0,1)$, we have that $f$ is maximized at $x = 1-\tfrac v \rho$. In particular, $f$ is strictly increasing on the interval $[0,1-\tfrac v \rho]$ and strictly decreasing on the interval $[1-\tfrac v \rho, 1)$.
\item Fix $z$ and denote $g(x) = \rho x -v\ln(\tfrac x {1-z})$. Since $g'(x) = \rho - \tfrac v x$ and it is an increasing function on the interval $[0,1]$ then $g(x)$ is strictly decreasing on the $(0, \tfrac v \rho)$ interval and strictly increasing on the $[\tfrac v \rho,1]$ interval. 
\end{itemize}

We return on our NE analysis.
First, for any $z$, it is evident that $\type=0$ agent has incentive to deviate if $x < 1-z$. (In response to $z$, the $\type=0$ agent increases her utility by deviating to playing $\Pr[\sigma(0)=0]=1-z$ since $X_{\rm game}(1-z,z) = 0$.)

Assume $z < 1/2$ for now. Therefore $x \in [1-z,1]$, otherwise the $\type=0$ agent has incentive to deviate. Since $x\geq 1-z\geq1/2 > z$ then $X_{\rm game} = \tfrac {z} {1-x}$ so type $\type=0$ agent's utility is $f(x)$. Since $f$ is strictly decreasing on the interval  $[1-\tfrac v \rho,1)$ we have that if $1-z \geq 1-\tfrac v \rho$ then the type $\type=0$ agent has no incentive to deviate when $x=1-z$. I.e., type $\type=0$ agent does not deviate from any strategy $(p,q) =(1-z,z)$ with $z \leq \tfrac v \rho$. In addition, type $\type=0$ doesn't deviate from $(p,q) = (1-\frac v \rho, z)$ for $\tfrac v\rho <z <1/2$.

Assume now the case $1/2 \leq z$. Again, we only need to consider $x\in [1-z,1]$, so either $x \in [1-z,z)$ or $x\in [z,1]$. In the former case, the utility of the type $\type=0$ agent is $g(x)$, and in the latter her utility is $f(x)$. Therefore:
\begin{itemize}
\item when $z < 1-\tfrac v \rho$ she considers only two possible strategies: $x=1-\tfrac v \rho$ (which maximizes $f(x)$ on the interval $[z,1]$), or $x=z$ (which maximizes $g(x)$ on the interval $[1-z,z]$). As $g(z) = f(z) < f(1-\tfrac v \rho)$ we deduce that in this case, type $\type=0$ agent does not deviate only from the strategy $(1-\tfrac v \rho, z)$.
\item when $z \geq 1-\tfrac v \rho$ she considers only two possible strategies: $x=z$ (which maximizes $f(x)$ on the interval $[z,1]$) , or $x=1-z$ (which might maximize $g(x)$ on the interval $[1-z,z]$). As $g(1-z) = v$ and $f(z) = \rho z - v\ln(\tfrac {z}{1-z})$ we have that $f(z)-g(1-z) \leq 0$ for any $z>z^*$ for some $z^*$ (the solution of $f(z)-g(z)=0$). Observe that $1-\tfrac v \rho < z^*$. We deduce that for any $z\in [1-\tfrac v \rho, z^*]$ the type $\type=0$ agent doesn't deviate from the strategy $(p,q)=(z,z)$; and for $z\in [z^*,1]$ type $\type=0$ agent does not deviate from $(p,q) = (1-z,z)$.
\end{itemize}

Recall that the type $\type=1$ agent is symmetric to type $\type=0$ agent, with the same utility function. This implies that any $(1-\tfrac v \rho,q)$ cannot be a NE since type $\type=1$ agent prefers to deviate, unless $q=1-\tfrac v \rho$. Therefore, we essentially characterized the NEs of the game, as specified in the theorem statement.
\end{proof}

}
\newcommand{\privacyAwareAgentInsteadAPX}{
\ifx \cameraready \undefined
Due to space limitations, the proof is deferred to Appendix~\ref{apx_sec:privacy_aware_agent}.
\else
The proof is deferred to the full version of the paper.
\fi
 The proof of Theorem~\ref{thm:behavior_privacy_aware} also applies to some alternative models of a privacy-aware agent. 
In addition to Theorem~\ref{thm:behavior_privacy_aware}, we also analyze, for completeness, an alternative scenario where type $0$ and type $1$ are two competing agents. Observe that this is no longer a Bayesian game with a single player but rather a standard complete-information game with two players. We show that this game also has NEs where both types play randomized strategies that follow Randomized Response (i.e.,$ \Pr[\sigma_B^*(0) = 0] = \Pr[\sigma_B^*(1)=1]) > \tfrac 1 2$).
}
\ifx \fullversion \undefined
\privacyAwareAgentInsteadAPX
\else
\PAA
\fi

\ifx \fullversion \undefined \vspace{-0.5cm} \fi
\ifx \fullversion \undefined \else \newpage \fi
\section{The Coupon Game with Scoring Rules Payments}
\label{sec:scoring_rule}
\ifx \fullversion \undefined \vspace{-0.4cm} \fi
In this section, we model the payments between $A$ and $B$ using a proper scoring rule (see below). This model is a good ``first-attempt'' model for the following two reasons. (i) Proper scoring rules assign profit to $A$ based on the accuracy of her belief, so $A$ has incentives to improve her prior belief on $B$'s type. (ii) As we show, in this model it is possible to quantify the $B$'s trade-off between an $\epsilon$-change in the belief and the cost that $B$ pays $A$. In that aspect, this model gives a clear quantifiable trade-off that explains what each additional unit of $\epsilon$-differential privacy buys $B$. Interestingly, proper scoring rules were recently applied in the context of differential privacy~\cite{GhoshLRS14} (yet in a very different capacity). 

\newcommand{\properScoringRulesWithAPX}{
Proper scoring rules (see surveys~\cite{Winkler96,Gneiting:07}) were devised as a method to elicit experts to report their true prediction about some random variable. 
For a $\bits$-valued random variable $X$, an expert is asked to report a prediction $x\in[0,1]$ about the probability that $X=1$.  
We pay her $f_1(x)$ if indeed $X=1$ and $f_0(x)$ otherwise. A \emph{proper scoring rule} is a pair of functions $(f_0, f_1)$ such that $\arg\max_x \E_{\type\leftarrow X}[f_\type(x)] = \Pr[X=1]$. Hence a risk-neutral agent's best strategy is to report $x=\Pr[X=1]$. Most frequently used proper scoring rules are \emph{symmetric} (or label-invariant) rules, where $\forall x, f_1(x) = f_0(1-x)$ (also referred to as neutral scoring rules in~\cite{ChenDPV14}). With symmetric proper scoring rules, the payment to an expert reporting $x$ as the probability of a random variable $X$ to be $1$, is identical to the payment of an expert reporting $(1-x)$ as the probability of the random variable $(1-X)$ to be $1$. Additional background regarding proper scoring rules is deferred to 
\ifx \cameraready \undefined
Appendix~\ref{apx_sec:proper_scoring_rules}.
\else
the full version of this paper.
\fi
}
\properScoringRulesWithAPX

\newcommand{\backgroundProperScoringRules}{
\subsection{Background: Proper Scoring Rules}
\label{subsec:proper_scoring_rules}

Proper scoring rules (see surveys~\cite{Winkler96,Gneiting:07}) were devised as a method to elicit experts to report their true prediction as to the probability of an event happening. That is, given a Bernoulli random variable $X$, we ask an expert to report her estimation of $\mu=\Pr[X=1]$. Given that the expert reports $x$ we pay her $f_1(x)$ if indeed $X=1$ and pay her $f_0(x)$ otherwise. A \emph{proper scoring rule} is a pair of functions $(f_0, f_1)$ such that $\arg\max_x \E_{\type\leftarrow X}[f_\type(x)] = \mu$ where the maximum is obtained for a unique report. That is, it is in the expert's best interest to report the true prior. It was shown~\cite{Savage:71,Gneiting:07} that a pair of twice-differentiable functions $(f_0,f_1)$ give a proper scoring rule iff there exists a convex function $g$ (i.e. $g'' > 0$ on the $[0,1]$ interval) s.t. $f_0(x) =  g(x) - xg'(x)$, $f_1(x) = g(x) + (1-x) g'(x)$. Using the derivatives of both functions ($f_0'(x) = -xg''(x)$ and $f_1'(x) = (1-x)g''(x)$), we deduce that $f_0$ is a strictly decreasing function and $f_1$ is a strictly increasing function on the $[0,1]$ interval.
And so, given that $X=1$ w.p. $\mu$, we have that the expected payment for an expert predicting $x$ is
\begin{equation}
F_\mu(x) = (1-\mu)f_0(x) + \mu f_1(x) = g(x) - (x-\mu)g'(x) \label{eq:F_mu}
\end{equation}
which is maximized at $x=\mu$, where $F_\mu(\mu) = g(\mu)$. 

Most commonly discussed proper scoring rules are \emph{symmetric} (or label-invariant) proper scoring rules, that are oblivious to that outcomes of $X$ (also referred to as neutral scoring rules in~\cite{ChenDPV14}). That is, symmetric scoring rules have the property that for any two Bernoulli random variables $X$ and $X'$ s.t. $\Pr[X=1]=\Pr[X'=0]$ the expected payment for an expert predicting $x$ for $X$ is identical to the payment for an expert predicting $1-x$ for $X'$. Such symmetric scoring rules are derived from a convex function $g$ that is symmetric around $\tfrac 1 2$. I.e.: $g(x)=g(1-x)$, and so $g'(x) = -g'(1-x)$ and $g''(x)=g''(1-x)$.

Concrete examples of proper scoring rules, such as the quadratic scoring rule, the spherical scoring rule and the logarithmic scoring rules, are discussed in Section~\ref{subsec:specific_scoring_rules}.
}
\ifx \fullversion \undefined \vspace{-0.5cm} \fi
\subsection{The Game with Scoring Rule Payments}
\label{subsec:scoring_rule_game}
\ifx \fullversion \undefined \vspace{-0.3cm} \fi
We now describe the game, and analyze its BNE. In this game $A$ interacts with a random $B$ from a population that has $D_0$ fraction of type $0$ agents and $D_1$ fraction of type $1$ agents. Wlog we assume throughout Sections~\ref{sec:scoring_rule}, \ref{sec:matching_pennies} and \ref{sec:opt-out-possible} that $D_0 \geq D_1$. $A$ aims to discover $B$'s secret type. She has utility that is directly linked to her posterior belief on $B$'s type
and $A$ reports her belief that $B$ is of type $1$. $A$'s payments are given by a proper scoring rule, composed of two functions $(f_0,f_1)$, so that after reporting a belief of $x$, a $B$ agent of type $\type$ pays $f_\type(x)$ to $A$.

\paragraph{A benchmark game.} First consider the following straight forward (and more boring) game where $B$ does nothing, $A$ merely reports $x$ -- her belief that $B$ is of type $1$. In this game $A$ gets paid according to a proper scoring rule --- i.e., $A$ gets a payment of $F_{D_1}(x) \stackrel{\rm def}{=} D_0 f_0(x) + D_1 f_1(x)$ in expectation. Since $(f_0, f_1)$ is a proper scoring rule, $A$ maximizes her expected payment by reporting $x=D_1$. So, in this game $A$ gets paid $g(D_1) \stackrel{\rm def} {=} f_{D_1}(D_1)$ in expectation, whereas $B$'s expected cost is $g(D_1)$. (Alternatively, a $B$ agent of type $0$ pays $f_0(D_1)$ and a $B$ agent of type $1$ pays $f_1(D_1)$.)

\paragraph{The full game.} We now turn our attention to a more involved game. Here $A$, aiming to have a more accurate posterior belief on $B$'s type, offers $B$ a coupon. Agents of type $\type$ prefer a coupon of type $\type$. And so, $B$ chooses what type to report $A$, who then gives $B$ the coupon and afterwards makes a prediction about $B$'s probability of being of type $1$. The formal stages of the game are as follows.
\begin{enumerate}
\addtocounter{enumi}{-1}
\item $B$'s type, $\type$, is drawn randomly with $\Pr[\type=0]=D_0$ and $\Pr[\type=1] =D_1$.
\item $B$ reports to $A$ a type $\hat \type=\sigma_B(t)$ and receives utility of $\rho_\type$ if indeed $\hat \type = \type$. We assume throughout this section that $\rho_0=\rho_1=\rho$. 
\item $A$ reports a prediction $x$, representing $\Pr[\type=1 ~|~ \sigma_B(\type)=\hat\type]$, and receives a payment from $B$ of $f_\type(x)$.
\end{enumerate}

\newcommand{\BNEScoringRules}{
Consider the coupon game with payments in the form of a symmetric proper scoring rule and with the following added assumption about the value of the coupon: $f_1(D_0)-f_1(D_1) < \rho < f_1(1)-f_1(0)=f_0(0)-f_0(1)$.
The unique BNE strategy of $B$ in this game, denoted $\sigma_B^*$, satisfies that $\Pr[\type = 0~|~\sigma_B^*(\type)=0] = \Pr[\type = 1~|~\sigma_B^*(\type)=1]$. 
}
\begin{theorem}
\label{thm:BNE_scoring_rules}
\BNEScoringRules
\end{theorem}
Note that a Randomized Response strategy $\sigma_B$ for $B$ would instead have $\Pr[\sigma_B(0)=0]=\Pr[\sigma_B(1)=1]$. This condition is different from the condition in Theorem~\ref{thm:BNE_scoring_rules} when $\Pr[t=0]\neq\Pr[t=1]$ (i.e., $D_0\neq D_1$). 
\newcommand{\proofTheoremProperScoringRule}{
\begin{proof}
We first analyze both agents' utilities and strategies. The utility of $A$ is solely based on the payments of the proper scoring rule: $E_{\type\leftarrow \{D_0,D_1\}} [ f_\type(x) ]$. $A$ has to decide on two potential reports: $x_0$ and $x_1$, where for $b\in\bits$, $x_b$ represents $A$'s belief about $\Pr[\type=1 ~|~ \hat \type = b]$. Therefore, a strategy $\sigma_A$ of $A$ maps a signal $\hat\type$ into a report. The utility of $B$ has two components~--- $B$ gains a certain amount of utility $\rho_\type$ from reporting $A$ the true type, but then has to pay $A$ her scoring rule payments. Therefore a strategy $\sigma_B$ maps each of $B$'s types to a signal. Given a strategy $\sigma_B$ we use the following notation:
\begin{eqnarray*}
p= \Pr[\sigma_B(0)=0], && q = \Pr[\sigma_B(1)=1]
\end{eqnarray*}
This way, $B$'s utility function takes the form
\begin{eqnarray*}
&u_B & = D_0 u_{B,0} + D_1 u_{B,1} \cr
\textrm{where} & u_{B,0} & = p\left( \rho - f_0(x_{0})\right) + (1-p)\left(-f_0(x_{1})\right) \cr
& u_{B,1} & = q\left( \rho - f_1(x_{1})\right) + (1-q)\left(-f_1(x_{0})\right) \cr
\end{eqnarray*}

When $A$ sees the signal $\hat\type$ the probability over $B$'s type is given by Bayes Rule:
\begin{eqnarray}
y_0 = y_0(p,q) \stackrel{\rm def}= \Pr[ \type = 1 ~|~ \hat\type = 0] = \frac {D_1(1-q)}{D_0p + D_1(1-q)} = \frac 1 {1 + \frac {D_0p}{D_1(1-q)}}
\label{eq:y_0} \\
y_1 = y_1(p,q) \stackrel{\rm def}= \Pr[ \type = 1 ~|~ \hat\type = 1] = \frac {D_1q}{D_0(1-p) + D_1q} = \frac 1 {1 + \frac {D_0(1-p)}{D_1q}} 
\label{eq:y_1}\end{eqnarray} 
and since $A$'s payments come from a proper scoring rule it follows that $A$ reports $x_0 = \sigma_A(0)=y_0$  and $x_1=\sigma_A(1)=y_1$. In other words, given that $B$'s BNE strategy is $(p^*,q^*)$, then $A$ plays best-response of $x_0^* = y_0(p^*,q^*), x_1^* = y_1(p^*,q^*)$.

We now turn to analyze $B$'s utility. Denote the strategy that $A$ plays as $x_0$ and $x_1$. Then agent $B$ decides on $p$ and $q$ that maximize the utility function
\[ u_B = D_0 \cdot \left( p(\rho - f_0(x_0)) - (1-p)f_0(x_1) \right) + D_1\cdot\left( q(\rho - f_1(x_1))-(1-q)f_1(x_0) \right) \]
It is simple to characterize $B$'s best response to $A$'s strategy of $(x_0,x_1)$. 
\begin{flalign}
& \textrm{If }\rho > f_0(x_0)-f_0(x_1) \textrm{ then } p =1\cr
& \textrm{If }\rho < f_0(x_0)-f_0(x_1)  \textrm{ then } p =0 \cr
& \textrm{If }\rho = f_0(x_0)-f_0(x_1)  \textrm{ then $B$ may play any } p\in[0,1] &\cr 
& \textrm{If }\rho > f_1(x_1)-f_1(x_0)  \textrm{ then } q =1  \cr
& \textrm{If }\rho < f_1(x_1)-f_1(x_0) \textrm{ then } q =0 \cr
&\textrm{If }\rho = f_1(x_1)-f_1(x_0)  \textrm{ then $B$ may play any } q \in [0,1] \label{eq:Bs_strategies}
\end{flalign}

We now wish to characterize the game's BNEs. First, we claim that in a BNE, with $B$ playing $\sigma_B^* = (p^*,q^*)$, it cannot be that $p^*<1-q^*$. This follows from the fact that $y_0(p,q) > y_1(p,q) \Leftrightarrow p < 1-q$. It means that $A$'s best response to such $(p^*,q^*)$ is to answer some $(x_0,x_1)$ s.t. $x_0 > x_1$. But since $f_0$ is a decreasing function, $f_1$ is an increasing function and $\rho > 0$, then $B$'s best response to such $(x_0,x_1)$ is to deviate to $(1,1)$. Similarly, should $(p^*,q^*)$ be such that $p^*=1-q^*$ \emph{and both} $p^*,q^*\in(0,1)$, then $A$'s best response $(x_0,x_1)$ is $(\tfrac 1 2,\tfrac 1 2)$, which implies again that $B$ prefers to deviate to $(1,1)$. It follows that, with the exception of $(1,0)$ and $(0,1)$, any BNE strategy of $B$ satisfies $p^*>1-q^*$, and so any BNE strategy of $A$ satisfies $x_0 < x_1$.

Before continuing with the proof we would like to make two observations, which we will repeatedly use. Let $X$ be a uniform Bernoulli random variable. We examine the expected payment to an expert reporting a belief of $z$ as to the probability of the event $X=1$, which we denote as $F_{1/2}(z)= \tfrac 1 2 (f_0(z)+f_1(z))$. The function $F_{1/2}$ is a concave function with a unique maximum at $z=\tfrac 1 2$, and it is strictly increasing on the $[0,\tfrac 1 2]$ interval and strictly decreasing on $[\tfrac 1 2,1]$ interval. Therefore, for any $a$ there exists at most two distinct preimages $z_1 \leq \tfrac 1 2 \leq z_2$ satisfying $F_{1/2}(z_1)=F_{1/2}(z_2) = a$. Recall that we assume $(f_0,f_1)$ is a symmetric proper scoring rule (so $f_1(z)=f_0(1-z)$ for any $z\in[0,1]$). So our first observation is: for any $z_1,z_2$ satisfying $ F_{1/2}(z_1)=F_{1/2}(z_2)$ and $z_2 > z_1$, we have that $z_2 = 1-z_1$ with $z_1\in[0,1/2)$ and $z_2 \in (1/2,1]$. Using again the fact that $(f_0,f_1)$ is a symmetry proper scoring rule and the fact that $F_{1/2}$ is maximized at $z=\tfrac 1 2$, we make our second observation: for any $z,z'$ satisfying $F_{1/2}(z)\geq F_{1/2}(z')$ it must hold that $|z-\tfrac 1 2| \leq |z'-\tfrac 1 2|$, which implies that $z \in [z',1-z']$ if $z' \leq 1/2$.

We now return to the proof of the theorem using case analysis as to the potential BNE strategies of $B$. We will rely also on our assumption that $D_0 \geq D_1$.
\begin{itemize}
\item $(p^*,q^*)=(1,1)$, i.e. $B$ always plays $\hat\type=\type$. This means that $A$ sets $x_0=0$ and $x_1=1$. (I.e., $A$ always predicts $\type = b$ given the signal $\hat\type = b$.) 
\\$(\ast)$ We deduce that if $\rho \geq f_0(0)-f_0(1)$ and $\rho \geq f_1(1)-f_1(0)$, then the game has a BNE of \[(x_0^*,x_1^*) = (0,1), ~~~ (p^*,q^*)=(1,1)\] We comment that since $(f_0,f_1)$ is a symmetric proper scoring rule, then we have that $f_0(0)-f_0(1) = f_1(1)-f_1(0)$.

\item $(p^*,q^*) = (1,0)$, i.e. $B$ only sends the $\hat\type=0$ signal. So when $A$ sees the $\hat\type=0$ signal she sets $x_0=D_1$ just as in the benchmark yet. But $A$ is indifferent as to the choice of $x_1$ since the $\hat\type=1$ signal is never sent. In order for this to be a BNE it must holds that $f_1(x_1)-f_1(D_1)\geq \rho \geq f_0(D_1)-f_0(x_1)$ so that both types of $B$ agent would keep sending the $\hat\type=0$ signal. So $x_1$ satisfies that $F_{1/2}(x_1)=\tfrac 1 2\left(f_0(x_1)+f_1(x_1)\right) \geq F_{1/2}(D_1) = \tfrac 1 2 \left(  f_0(D_1) + f_1(D_1) \right)$. Based on our second observation, we have that $x_1\in[D_1, D_0]$.\\
$(\ast)$ We deduce that if the parameters of the game are set such that there exists $v\in [D_1,D_0]$ satisfying both $f_0(v)\geq f_0(D_1)- \rho$ and $f_1(v) \geq f_1(D_1) + \rho$ then the game has a BNE of \[(x_0^*,x_1^*) = (D_1,v), ~~~ (p^*,q^*)=(1,0)\] 
As $f_1$ is an increasing function, it must hold that $\rho \leq f_1(D_0)-f_1(D_1)$. In other words, when $\rho > f_1(D_0)-f_1(D_1)$ then this cannot be a BNE. 

\item $(p^*,q^*)=(0,1)$. This means that $B$ only sends the $\hat\type=1$ signal. So now $A$ sets $x_1=D_1$ but $A$ is indifferent regarding  the value of $x_0$. In order for $B$ not to deviate from$(0,1)$ then $x_0$ should satisfy both $\rho \leq f_0(x_0)-f_0(D_1)$ and $\rho \geq f_1(D_1)-f_1(x_0)$. This implies that $F_{1/2}(x_0) \geq F_{1/2}(D_1)$ and our second observation gives that $x_0 \in [D_1,D_0]$. But observe that $f_0(x_0) \geq \rho-f_0(D_1) > f_0(D_1)$. This contradicts the fact that $f_0$ is a strictly decreasing function.

\item $p^*=1$ while $q^* \in (0,1)$. This means $A$ sets $x_1 = 1$ (because only type $1$ agents can send $\hat\type=1$), while setting $x_0=y_0(p^*,q^*)>0$. To keep $B$ from deviating then $x_0$ should satisfy that $\rho \geq f_0(x_0) - f_0(1)$ and $\rho = f_1(1)-f_1(x_0)$. Therefore $F_{1/2}(1) \geq F_{1/2}(x_0)$, so our observation yields the contradiction $1\in [x_0,1-x_0]$.

\item $q^*=1$ while $p^*\in (0,1)$. This case is symmetric to the previous case, and we get a similar contradiction using $F_{1/2}(0) \geq F_{1/2}(x_1)$.

\item $p^*,q^* \in (0,1)$ with $p^* > 1-q^*$. We know that $A$'s best response is setting $x_0^*=y_0(p^*,q^*)$ and $x_1^*=y_1(p^*,q^*)$ and we have already shown that $y_0 < y_1$. In order for $B$ to play best response against $(y_0,y_1)$ we must have that $\rho = f_0(y_0) - f_0(y_1) = f_1(y_1)-f_1(y_0)$ so $F_{1/2}(y_0)=F_{1/2}(y_1)$. Based on our first observation from before we have that $y_1 = 1-y_0$. In other words, $B$ picks $p^*$ and $q^*$ s.t. the signals $\hat\type =0$ and $\hat\type=1$ are symmetric:
\begin{eqnarray*}
& \Pr[\type = 1 ~|~ \hat\type = 1] & = y_1 = 1- y_0\cr
&& = 1 - \Pr [\type=1 ~|~ \hat\type = 0] = \Pr[\type = 0 ~|~ \hat\type=0]
\end{eqnarray*}
so regardless of the value of $b$, the expression $\Pr[\type = \hat\type ~|~ \hat \type = b]$ is the same.\\
Observe that we have $\rho = f_0(y_0)-f_0(y_1) = f_0(y_0) - f_0(1-y_0) =-g'(y_0)$ or $\rho=g'(y_1)$. (Recall, $(f_0,f_1)$ are derived using a convex function $g$ as detailed in Section~\ref{subsec:proper_scoring_rules}.) In other words, $B$ sets $(p^*,q^*)$ by  first finding $y_1\in(\tfrac 1 2,1]$ s.t. $g'(y_1)=\rho$, then finding $(p^*,q^*)$ that satisfy Equation~\eqref{eq:p_q_and_ratios} and yield $y_1$. Formally, $B$ finds $(p^*,q^*)$ that satisfy 
\begin{equation}
\rho = g'( \frac{D_1q^*} {D_0(1-p^*)+D_1q^*}) = -g'( \frac{D_1(1-q^*)} {D_0p^*+D_1(1-q^*)})
\label{eq:derivative_by_p}
\end{equation}
Recall that $g$ is convex and $g''>0$ on the $[0,1]$ interval. This implies that as $\rho$ increases, the point $y_1(p^*,q^*)$ gets further away from $\tfrac 1 2$ and closer to $1$.
\end{itemize}
\end{proof}

Recall that in order for $B$ to play according to Randomized Response, $B$ should set $p^*=q^*$. Yet, in this game, a rational agent $B$ plays s.t. $A$'s posterior on $B$'s type is symmetric. Indeed, $1-y_0 = y_1$~implies
\begin{equation}
\frac {D_0p^*}{D_0p^*+D_1(1-q^*)} = \frac{D_1q^*} {D_0(1-p^*)+D_1q^*} ~~~\Rightarrow~~~ D_0^2p^*(1-p^*) = D_1^2q^*(1-q^*) \label{eq:p_q_and_ratios}
\end{equation}
and so, unless $D_0=D_1$, we have that $p^*\neq q^*$.

Lastly, we comment about $A$'s payment. Using the notation of Equation~\eqref{eq:F_mu}, when $\hat\type=0$ then $A$ gets an expected payment of $F_{y_0}(y_0) = g(y_0)$, and when $\hat\type=1$ then $A$ gets $F_{y_1}(y_1) = g(y_1)$. But as $y_0=1-y_1$ and the scoring rule is symmetric, we have that $A$ gets the same payment regardless of the signal, so $A$'s payment is $g(y_1)$. Recall that $y_1$ is the point where $\rho=g'(y_1)$.

So, is this game worth while for $A$? Imagine that $A$ could choose between either this coupon game, or the ``benchmark game'' in which $A$ guesses $B$'s type without viewing any signal from $A$. Recall, in the benchmark game, $A$ gets an expected profit of $g(D_1) = g(D_0)$. Recall that $g$ is a convex function that is minimized at $x=\tfrac 1 2$. Therefore, $g(y_1) > g(D_0)$ if $\tfrac 1 2 \leq D_0 < y_1$ which also implies $g'(D_0) < g'(y_1) = \rho$. In other words, $A$ gains more money than in the benchmark game only if $A$ offers a coupon of high-value. 

\paragraph{The case with $\rho_0 \neq \rho_1$.} We briefly discuss the case where $\rho_0$ and $\rho_1$ are not equal. First of all, observe that now there could be situations in which the BNE is of the form $(1,q^*)$ with a non-integral $q^*$, or the symmetric $(p^*,1)$. This is because the previous contradiction no longer holds. More interestingly the BNE we get: $(y_0,y_1)$ and  $(p^*,q^*)$ still satifies Equations~\eqref{eq:y_0} and~\eqref{eq:y_1}, and also
\begin{eqnarray*}
\rho_0 = f_0(y_0)-f_0(y_1) &,& \rho_1 = f_1(y_1)-f_1(y_0)
\end{eqnarray*}
which, using $\rho_0\rho_1$ can be manipulated into
\[ \frac {\rho_1}{\rho_0 + \rho_1} f_0(y_0) + \frac{\rho_0}{\rho_0 + \rho_1}f_1(y_0) =  \frac{\rho_1}{\rho_0 + \rho_1} f_0(y_1) + \frac {\rho_0}{\rho_0 + \rho_1} f_1(y_1)  \]
In otherwords,  setting $\mu= \tfrac {\rho_0}{\rho_0+\rho_1}$, we have $F_\mu(y_0) = F_\mu(y_1)$. Alternatively, it is possible to subtract the two equalities and deduce:
\[ \tfrac 1 2(\rho_0-\rho_1)= \tfrac 1 2 (f_0(y_0) + f_1(y_0)) - \tfrac1 2 (f_0(y_0) + f_1(y_1))  = F_{1/2}(y_0)-F_{1/2}(y_1)\]
These two conditions (along with $y_0 < y_1$) dictate the value of $y_0,y_1$, and thus the values of $(p^*,q^*)$. Sadly, it is no longer the case that $y_1=1-y_0$.
}
\newcommand{\insteadOfProofForScoringRulesAPX}{
\ifx \cameraready \undefined
The proof of Theorem~\ref{thm:BNE_scoring_rules} is deferred to Appendix~\ref{apx_sec:proper_scoring_rules}, where we also compare $A$'s profit in the benchmark game to her profit from her BNE strategy in the full game.
\else
The proof of Theorem~\ref{thm:BNE_scoring_rules} is in the full version of this paper, where we also compare $A$'s profit in the benchmark game to her profit from her BNE strategy in the full game.
\fi
}
\ifx \fullversion \undefined
\insteadOfProofForScoringRulesAPX
\else
\proofTheoremProperScoringRule
\fi

In Appendix~\ref{subsec:specific_scoring_rules} we discusse the implications of using specific scoring rules.


\newcommand{\specificScoringRules}{
\subsection{Strategies Under Specific Scoring Rules}
\label{subsec:specific_scoring_rules}
We now plug-in different types of proper and symmetric scoring rules, and find what $p^*$ and $q^*$ are in each case. We analyze the game for a value of $\rho$ s.t. the BNE is obtained where neither $p^*$ nor $q^*$ are integral. We also characterize what is the $\epsilon$ in $A$'s posterior probability --- the value of $\max_b\left\{\ln\left(\frac {\Pr[\hat \type=\type ~|~ \type=b] } {\Pr[\hat\type=1-\type ~|~ \type = b]} \right)\right\}$.

There exists 3 canonical rules often used in literature: Quadratic, Spherical and Logarithmic.

\paragraph{Quadratic Scoring Rule.} The quadratic scoring rule is defined by the functions $(f_0(x),f_1(x)) = (2-2x^2,4x-2x^2)$. The quadratic scoring rule is generated by the convex function $g(x) = x^2+(1-x)^2+1 = 2-2x+2x^2$. (So, $g'(x) = -2+4x$ and $g''(x)=-2$.) Therefore $(f_0'(x),f_1'(x)) = ( -4x, 4(1-x) )$.

Observe that since $g' \in [-2,2]$, Equation~\eqref{eq:derivative_by_p} gives that $\rho \in [-2,2]$ as well. Hence, Equation~\eqref{eq:derivative_by_p} takes the form
\begin{eqnarray*}
& \rho = 2 - \frac 4 {1+ \frac {D_0p}{D_1(1-q)}} &\Rightarrow~~  \frac {D_0p}{D_1(1-q)} =\frac {2+\rho} {2-\rho} \Rightarrow~~ p = \frac {D_1}{D_0} \frac {2+\rho} {2-\rho} (1-q)\cr
& \rho = -2 + \frac 4 {1+ \frac {D_0(1-p)}{D_1q}} &\Rightarrow~~  \frac {D_0(1-p)}{D_1q} = \frac {2-\rho}{2+\rho} \Rightarrow~~ q = \frac {D_0}{D_1} \frac {2+\rho} {2-\rho} (1-p)\cr
\end{eqnarray*}
So we have  \[ p = \frac {D_1}{D_0} \frac {2+\rho} {2-\rho} - \left(\frac {2+\rho} {2-\rho}\right)^2(1-p) ~\Rightarrow ~~ p = \left(\frac {2+\rho} {2-\rho}\right) \left(\frac {D_0}{D_1}-\frac {2+\rho} {2-\rho}\right) / \left(1-\left(\frac {2+\rho} {2-\rho}\right)^2 \right)\]
which boils down to
\[ p  = \frac {2+\rho} 4 \left(\frac{\frac {2+\rho} {2-\rho}-\frac {D_0}{D_1}} {\frac {2+\rho} {2-\rho}-1}\right) = \frac {2+\rho} 4 \left( \frac { 2(1-\frac {D_0}{D_1})+\rho(1+\frac {D_0}{D_1}) } {2\rho}  \right) = \frac {2+\rho} 4 \left(\frac 1 2(1+\frac {D_0}{D_1}) - \frac {D_0-D_1}{\rho D_1} \right)  \]

And similarly,
\[ q = \frac {D_0}{D_1} \frac {2+\rho} {2-\rho} - \left(\frac {2+\rho} {2-\rho}\right)^2(1-q) ~\Rightarrow ~~ q = \left(\frac {2+\rho} {2-\rho}\right) \left(\frac {D_1}{D_0}-\frac {2+\rho} {2-\rho}\right) / \left(1-\left(\frac {2+\rho} {2-\rho}\right)^2 \right)\]
which gives
\[ q  = \frac {2+\rho} 4 \left(\frac{\frac {2+\rho} {2-\rho}-\frac {D_1}{D_0}} {\frac {2+\rho} {2-\rho}-1}\right) = \frac {2+\rho} 4 \left( \frac { 2(1-\frac {D_1}{D_0})+\rho(1+\frac {D_1}{D_0}) } {2\rho}  \right) = \frac {2+\rho} 4 \left(\frac 1 2(1+\frac {D_1}{D_0}) + \frac {D_0-D_1}{\rho D_1} \right)  \]

More importantly, under these $p$ and $q$ values, $y_0 = \frac {2-\rho} 4$ and $y_1 = \frac {2+\rho} 4$. So from $A$'s perspective, there is a Randomized Response move here with $e^\epsilon = y_1/y_0$, hence 
\[\epsilon = \ln(\frac {2+\rho}{2-\rho} )\]
The expected utility of $A$ is $u_A = g(y_0) = y_0^2+y_1^2+1 = \frac {8+2\rho^2}{16}+1 = 2 - \tfrac 1 2 + \frac {\rho^2}8$. This is in comparison to $g(D_1) = 1+D_0^2+D_1^2 = 2-2D_1+2D_1^2 = 2-2D_1(1-D_1) = 2-2D_1D_0$. It follows that $A$ prefers the $2$nd game (with the coupon) to the first only if $\tfrac 1 2 - \tfrac {\rho^2}8 < 2 D_0D_1$ or $\rho^2 > 4-16D_0D_1$. Clearly, with $D_0=D_1=\tfrac 1 2$ we have that $A$ prefers the coupon game over the benchmark-game.

\paragraph{Spherical Scoring Rule.} The spherical scoring rule is defined by the functions
\[ (f_0(x),f_1(x) ) = ( \frac {1-x} {\sqrt{x^2+(1-x)^2}}, \frac x {\sqrt{x^2+(1-x)^2}})\] which are generated using $g(x) = \sqrt{x^2+(1-x)^2}$. (So, $g'(x) = \frac {2x-1} {\sqrt{x^2+(1-x)^2}}$ and $g''(x) = (x^2+(1-x)^2)^{-\tfrac 3 2}$.)
Therefore $(f_0'(x),f_1'(x)) = (  -x(1-2x+2x^2)^{-\tfrac 3 2} , (1-x)(1-2x+2x^2)^{-\tfrac 3 2} )$.

Using the definition of $g'(x)$, Equation~\eqref{eq:derivative_by_p} now yields
\begin{eqnarray*}
& \rho \sqrt{y_0^2 + 1-2y_0+y_0^2} = -2y_0+1 &\Rightarrow (4-2\rho^2)y_0^2-(4-2\rho^2)y_0 + (1-\rho^2)=0 \cr
& \rho \sqrt{y_1^2 + 1-2y_1+y_1^2} = 2y_1-1 & \Rightarrow (4-2\rho^2)y_1^2-(4-2\rho^2)y_1 + (1-\rho^2)=0
\end{eqnarray*}
So $y_0$ and $y_1$ are the two different roots of the equation $x^2-x+\frac{1-\rho^2} {4-2\rho^2} = 0$, namely $\tfrac 1 2 \pm \tfrac1 2 \sqrt{\frac {\rho^2}{2-\rho^2} }$
Plugging in the values of $y_0$ and $y_1$ we have
\begin{eqnarray*}
&&\frac{ D_1q - D_0(1-p)}{D_1q + D_0(1-p)} = \sqrt{\frac {\rho^2}{2-\rho^2}} \cr
&& \frac{ D_0p - D_1(1-q)}{D_0p + D_1(1-q)} = \sqrt{\frac {\rho^2}{2-\rho^2}}
\end{eqnarray*}
and this is because we assume $D_0 p > D_1(1-q)$ and $D_1q > D_0(1-p)$. (That is, when we see the signal $\hat\type = 0$ it is more likely to come from a $\type=0$-type agent than a $\type=1$-agent, and similarly with the $\hat\type=1$ signal.)

After arithmetic manipulations, we have
\begin{eqnarray*}
& (1-\rho^2)(D_0^2p^2 + D_1^2(1-q)^2) = 2 D_0D_1p(1-q) & \Rightarrow (1-\rho^2)D_0p = D_1(1-q)\left( 1 \pm  \rho\sqrt{2-\rho^2}\right) \cr
& (1-\rho^2)(D_0^2(1-p)^2+D_1^2q^2) = 2D_0D_1(1-p)q &\Rightarrow 
(1-\rho^2)D_1q = D_0(1-p)\left(1 \pm  \rho \sqrt{2-\rho^2}\right)
\end{eqnarray*}
using the fact that $\rho \leq 1$ and that $D_0p>D_1(1-q)$ and $D_1q>D_0(1-p)$, then
\begin{eqnarray*}
&& D_0p = D_1(1-q)\frac{ 1 +  \rho\sqrt{2-\rho^2}}{ 1-\rho^2} \stackrel{\rm def}= {Z_\rho} D_1(1-q)\cr
&& D_1q = D_0(1-p)\frac{1 +  \rho \sqrt{2-\rho^2}}{1-\rho^2} \stackrel{\rm def}= Z_\rho D_0(1-p)
\end{eqnarray*}
(because $1-\rho\sqrt{2-\rho^2} \leq 1-\rho \leq 1-\rho^2$.)
We have that
\begin{eqnarray*}
&& D_1 = D_1q+D_1(1-q) = Z_\rho D_0(1-p) + \frac 1 {Z_\rho} D_0p \cr
&& D_0 = D_0p+D_0(1-p) = {Z_\rho} D_1(1-q) + \frac 1{Z_\rho} D_1q 
\end{eqnarray*}
We deduce
\[ p = \frac {Z_\rho^2 - Z_\rho \frac{D_1}{D_0}}{Z_\rho^2-1}, \qquad q = \frac {Z_\rho^2 - Z_\rho \frac{D_0}{D_1}}{Z_\rho^2-1}\]

More importantly, from $A$'s perspective, the signal is like a Randomized Response with parameter $e^\epsilon = y_1/y_0$ so
\[\epsilon = \ln ( \left(1 + \sqrt{\frac {\rho^2}{2-\rho^2}}\right) / \left(1 - \sqrt{\frac {\rho^2}{2-\rho^2}}\right))\]
The utility of $A$ from the game is now $g(y_0)$ which boils down to $\frac 1 {2-\rho^2}$. This is in contrast to $D_0^2 + D_1^2$, so $A$ prefers the game with the coupon over the baseline when $\rho^2 > 2 - \frac 1 {D_0^2+D_1^2} = \frac {(D_0-D_1)^2} {D_0^2+D_1^2}$.	Complimentary to that, $B$'s expected payment is 
\[\rho(D_0p+D_1q) - g(y_0) = \rho \left( \frac{D_0Z_\rho^2-D_1 Z\rho +D_1Z_\rho^2-D_0Z_\rho} {Z_\rho^2-1}\right) -\frac 1 {2-\rho^2} = \frac {\rho Z_\rho}{Z_\rho + 1} - \frac 2 {2-\rho^2}\]

\paragraph{Logarithmic Scoring Rule.} The logarithmic scoring rule is defined by the functions $ (f_0(x),f_1(x)) = (\ln(1-x), \ln(x))$ which are generated by $g(x) = - H(x) = x\ln(x)+(1-x)\ln(1-x)$. (So, $g'(x) = \ln(x) -\ln(1-x)$ and $g''(x) = \tfrac 1 x +\tfrac 1 {1-x}$.)
Therefore $(f_0'(x),f_1'(x)) = ( -\frac 1 {1-x}	, \frac 1 x  ) $.
Observe that the logarithmic scoring rule has \emph{negative} costs, and furthermore, we may charge infinite cost from an expert reporting $x=0$ or $x=1$. 

Using $g'(x)$, Equation~\eqref{eq:derivative_by_p} takes the form
\begin{eqnarray*}
& \rho = \ln(\frac{1-y_0}{y_0}) = \ln(\frac{y_1}{1-y_1}) &\Rightarrow y_0 = \frac 1 {1+e^\rho},~~y_1 = \frac 1 {1+e^{-\rho}}
\end{eqnarray*}
This implies that
\[ \frac {D_0 p}{D_1(1-q)} = \frac {D_1q}{D_0(1-p)}=e^{\rho} ~~\Rightarrow~~ p = \frac {e^{2\rho}-e^\rho \tfrac{D_1}{D_0}}{e^{2\rho}-1},~~~ q = \frac {e^{2\rho}-e^\rho \tfrac{D_0}{D_1}}{e^{2\rho}-1}\]

The Randomized Response behavior that $A$ observes is for $e^\epsilon = y_1/y_0$ which means that simply $\epsilon = \rho$. The utility for $A$ is now $g(y_0) = -\frac {\ln(1+e^{\rho})}{1+e^\rho} - \frac {\ln(1+e^{-\rho})}{1+e^{-\rho}}$. And the utility for $B$ is $u_B = \rho(D_0p+D_1q)-g(y_0) = \rho \frac {e^{2\rho}-e^\rho}{e^{2\rho}-1} - g(y_0) = \frac {\rho e^\rho}{e^\rho+1} - g(y_0)$.
}
\cut{
\subsection{Skewing the Scoring-Rule Payments to Compensate for the Prior}
\label{subsec:skewed_scoring_rule}

As observed in Section~\ref{subsec:scoring_rule_game}, the way $B$ plays is to use the signal $\hat\type$ obscures the prior $\{D_0,D_1\}$ over its original type. Here, we show that this phenomena is a result of the payments to $A$ being symmetric. In this section we suggest using a shifted scoring rule.

Observe, given a proper scoring rule $(f_0,f_1)$ and any two positive constants $w_0, w_1$, we can define $\tilde f_0(x) = w_0 f_0(x)$ and $\tilde f(x_1) = w_1 f_1(x)1$. These are no longer proper scoring rule. In particular, given a r.v. $X$ s.t. $\Pr[X=1]=x$ it is most beneficial for an expert to report $x'=\frac {w_1 x}{w_0(1-x)+w_1 x} = \frac 1 {1+\frac {w_0(1-x)}{w_1x}}$. (Alternatively, to report $x'$ s.t. $\frac {x'}{1-x'} = \frac {w_1}{w_0}\cdot \frac x {1-x}$, or $\frac {w_0}{w_1} \frac {x'}{1-x'} = \frac x {1-x}$.) 

Plugging this into the game, we have that if agent $B$ follows strategy $(p,q)$, then it is $A$'s best response that given the signal $\hat\type = 0$ she predicts $x_0$ s.t. $\frac {w_0}{w_1} \frac{x_0}{1-x_0} = \frac {D_1(1-q)}{D_0p}$; given the signal $\hat\type=1$ she predicts $x_1$ s.t. $\frac{w_0}{w_1} \frac {x_1}{1-x_1} = \frac {D_1q}{D_0(1-p)}$. It is now simple to see that by setting $w_0 = D_1$ and $w_1 = D_0$ we ``neutralize'' the effect of the prior. This sets $A$'s best response for $(p,q)$ as $(x_0,x_1) = \left(\frac {1-q} {p+1-q},\frac q {1-p+q}\right)$.

$B$'s best response strategies are just the same as given in Eq~\eqref{eq:Bs_strategies}, only w.r.t $(\tilde f_0,\tilde f_1)$. 
In particular, similar conclusions hold for the different equilibria.
\begin{itemize}
\item If $\rho \geq \max\{D_0(f_1(1)-f_1(0)),D_1(f_0(0)-f_0(1))\}$, then the game has a BNE of \[(x_0^*,x_1^*) = (0,1), ~~~ (p^*,q^*)=(1,1)\]
\item If $\exists v>D_1, D_1(f_0(D_1)-f_0(v)) < \rho < D_0(f_1(v)-f_1(D_1))$ then the game has a BNE of \[(x_0^*,x_1^*) = (D_1,v), ~~~ (p^*,q^*)=(1,0)\]
\item If $\exists v>D_1, D_0(f_1(D_1)-f_1(v)) < \rho < D_1(f_0(v)-f_0(D_1))$ then the game has a BNE of \[(x_0^*,x_1^*) = (v,D_1), ~~~ (p^*,q^*)=(0,1)\]
\item Similarly to before, we can rule out any other case where at least one of $\{p^*,q^*\}$ is integral. Most cases follow the exact same line of reasoning. The case where $p^*=1$ while $q^* \in (0,1)$ or the symmetric case where $q^*=1$ and $p^*\in(0,1)$ can be ruled out by observing that for a r.v. $X$ s.t. $\Pr[X=1]=D_0$ predicting (using the proper scoring rule) any fractional $x'$ yields better utility than predicting $0$ or $1$. 
\end{itemize}
We are now left with the case of $p^*,q^*\in(0,1)$. In this case, Equation~\eqref{eq:tradeoff_of_rho} takes the form
\[ D_1f_0(x_0) + D_0 f_1(x_0) = D_1f_0(x_1) + D_0 f_1(x_1)\] I.e., given a r.v. $X$ s.t. $\Pr[X=1]=D_0$, then $x_0$ and $x_1$ yield the same expected utility: $\E[f_X(x_0)] = \E[f_X(x_1)]$. The function $F(x) = D_1 f_0(x) + D_0 f_1(x) = D_0 f_0(1-x) + D_1 f_1(1-x)$ has derivative of $F'(x) = (D_0-x)g''(x)$, and so it is strictly increasing on the $(0,D_0)$-interval and strictly decreasing on the $(D_0,1)$-interval. 
So one solution is for $x_0=x_1$, which we can immediately rule out as it leads to $\rho=0$. We deduce that $x_0 < D_0 < x_1$ (if $x_1<x_0$ then $\rho$ has to be negative). 

Unfortunately, there is no closed-form solution to this problem, unless we plug-in the functions $f_0, f_1$. But observe that there is no reason to have $x_0=1-x_1$ (which Randomized Response, with $p=q$ should give).
}

\ifx \fullversion \undefined \vspace{-0.4cm} \fi
\ifx \fullversion \undefined \else \newpage \fi
\section{The Coupon Game with the Identity Payments}
\label{sec:matching_pennies}
\ifx \fullversion \undefined \vspace{-0.35cm} \fi
In this section, we examine a different variation of our initial game. As always, we assume that $B$ has a type sampled randomly from $\{0,1\}$ w.p. $D_0$ and $D_1$ respectively, and wlog $D_0 \geq D_1$. Yet this time, the payments between $A$ and $B$ are given in the form of a $2\times 2$ matrix we denote as $M$. This payment matrix specifies the payment from $B$ to $A$ in case $A$ ``accuses'' $B$ of being of type $\tilde\type\in\{0,1\}$ and $B$ is of type $\type$. In general we assume that $A$ strictly gains from finding out $B$'s true type and potentially loses otherwise (or conversely, that a $B$ agent of type $\type$ strictly loses utility if $A$ accuses $B$ of being of type $\tilde\type=\type$ and potentially gains money if $A$ accuses $B$ of being of type $\tilde\type=1-\type$).
In this section specifically, we consider one simple matrix $M$ -- the identity matrix $I_{2\times 2}$. Thus, $A$ gets utility of $1$ from correctly guessing $B$'s type (the same utility regardless of $B$'s type being $0$ or $1$) and $0$ utility if she errs. 
\ifx \fullversion \undefined \vspace{-0.6cm} \fi
\subsection{The Game and Its Analysis}
\label{subsec:coupon_game}
\ifx \fullversion \undefined \vspace{-0.3cm} \fi
\paragraph{The benchmark game.} The benchmark for this work is therefore a very simple ``game'' where $B$ does nothing, $A$ guesses a type and $B$ pays $A$ according to $M$. It is clear that $A$ maximizes utility by guessing $\tilde\type=0$ (since $D_0 \geq D_1$) and so $A$ gains in expectation $D_0$; where an agent $B$ of type $\type =0$ pays $1$ to $A$, and an agent $B$ of type $\type=1$ pays $0$ to $A$.

\paragraph{The full game.} Aiming to get a better guess for the actual type of $B$, we now assume $A$ first offers $B$ a coupon. As before, $B$ gets a utility of $\rho_\type$ from a coupon of the right type and $0$ utility from a coupon of the wrong type. And so, the game takes the following form now.
\begin{enumerate}
\addtocounter{enumi}{-1}
\item $B$'s type, denoted $\type$, is chosen randomly, with $\Pr[\type=0]=D_0$ and $\Pr[\type=1]=D_1$.
\item $B$ reports a type $\hat\type=\sigma_B(t)$ to $A$. $A$ in return gives $B$ a coupon of type $\hat\type$.
\item $A$ accuses $B$ of being of type $\tilde\type=\sigma_A(\hat\type)$ and $B$ pays $1$ to $A$ if indeed $\tilde\type=\type$.
\end{enumerate}

And so, the utility of agent $A$ is $u_A = \mathds{1}_{[\tilde\type=\type]}$. The utility of agent $B$ is a summation of two factors -- reporting the true type to get the right coupon and the loss of paying $A$ for finding $B$'s true type. So $u_B = \rho_\type\mathds{1}_{[\hat\type=\type]} - \mathds{1}_{[\tilde\type=\type]}$.
\newcommand{\couponMatchingPennies}{
In the coupon game with payments given by the identity matrix with $\rho_0\neq\rho_1$, any BNE strategy of $B$ is pure for at least one of the two types of $B$ agent. Formally, for any BNE strategy of $B$, denoted $\sigma_B^*$, there exist $\type,\hat\type\in\bits$ s.t. $\Pr[\sigma_B^*(\type)=\hat\type]=1$.
}
\begin{theorem}
\label{thm:coupon_matching_pennies}
\couponMatchingPennies
\end{theorem}
In the case where $\rho_0=\rho_1$ then $B$ has infinitely many randomized BNE strategies, including a BNE strategy $\sigma_B^*$ s.t. $\tfrac 1 2 \leq \Pr[\sigma_B^*(0)=0] = \Pr[\sigma_B^*(1)=1] < 1$ (Randomized response).
\newcommand{\proofTheoremIdentityMatrix}{
\begin{proof}
First, we denote the strategies of agents $A$ and $B$. We denote
\begin{flalign*}
 \textrm{For $B$: }  & ~~p = \Pr[\sigma_B(0)=0] \textrm{, and }~  q = \Pr[\sigma_B(1)=1] \cr
 \textrm{For $A$: } 
 & ~~x =\Pr[\sigma_A(0)=0] \textrm{, and }~ y=\Pr[\sigma_A(1)=1]
\end{flalign*}
Using these $4$ parameters,we analyze the utility functions of the agents of the game. We start with the utility function of $A$:
\begin{align*}
 u_A = D_0 p x + D_0 (1-p)(1-y) + D_1 q y + D_1 (1-q)(1-x)
\end{align*}
This function characterizes $A$'s best response strategy as follows. $A$ determines $x  =\Pr[\sigma_A(0)=0]$ based on the relation between $D_0p$ ($= \Pr[\type=0 \wedge \hat\type=0]$) and $D_1(1-q)$ ($= \Pr[\type=1 \wedge \hat\type=0] $) --- if $D_0p$ is the larger term, then $x=1$; if $D_1(1-q)$ is the larger term, then $x=0$; and if both are equal then $A$ is free to set any $x\in [0,1]$. Similarly, the relationship between $D_1q = \Pr[\type=1 \wedge \hat\type=1]$ and $D_0(1-p) = \Pr[\type=0\wedge\hat\type=1]$ determines the value of $y = \Pr[\sigma_A(1)=1]$. 

We therefore denote the following two lines on the $[0,1]\times [0,1]$ square of possible choices for $p$ and $q$
\begin{eqnarray*}
&\ell_1 : & q = 1- \tfrac {D_0}{D_1} p \textrm{, (i.e., $D_0p=D_1(1-q)$)}\cr
&\ell_2 : & q = \tfrac {D_0}{D_1} (1-p)\textrm{, (i.e., $D_0(1-p)=D_1q$)}
\end{eqnarray*}
These are $A$'s ``lines of indifference'': when $B$ plays $(p,q)\in \ell_1$ then $A$ is indifferent to any value of $x$ in the range  $[0,1]$, and when $B$ plays $(p,q)\in\ell_2$ then $A$ is indifferent between any value of $y$.

Observe that $\ell_1$ and $\ell_2$ have the same slope, and so they are parallel, and that the point $(p,q)=(1,0)$ is above $\ell_1$ yet on $\ell_2$. It follows that $\ell_2$ is above $\ell_1$ (unless $D_0=D_1=\tfrac 1 2$ in which case the two lines coincide). The two lines are shown in Figure~\ref{fig:strategy_space_for_B}.

\cut{
\begin{flalign*}
\textrm{For $A$:} &\cr
& \textrm{ if } D_0 p > D_1 (1-q), ~\textrm{ i.e. } \Pr[t=0 \wedge \hat\type=0] > \Pr[\type=1 \wedge \hat\type=0] \textrm{, then } x = 1\cr
& \textrm{ if } D_0 p < D_1 (1-q), ~\textrm{ i.e. } \Pr[t=0 \wedge \hat\type=0] < \Pr[\type=1 \wedge \hat\type=0] \textrm{, then } x = 0\cr
& \textrm{ if } D_0 p = D_1 (1-q), ~\textrm{ i.e. } \Pr[t=0 \wedge \hat\type=0] = \Pr[\type=1 \wedge \hat\type=0] \textrm{, then } A \textrm{ can play any } x\in[0,1]\cr
& \textrm{ if } D_1 q > D_0 (1-p), ~\textrm{ i.e. } \Pr[t=1 \wedge \hat\type=1] > \Pr[\type=0 \wedge \hat\type=1] \textrm{, then } y = 1\cr
& \textrm{ if } D_1 q < D_0 (1-p), ~\textrm{ i.e. } \Pr[t=1 \wedge \hat\type=1] < \Pr[\type=0 \wedge \hat\type=1]  \textrm{, then } y = 0\cr
& \textrm{ if } D_1 q = D_0 (1-p), ~\textrm{ i.e. } \Pr[t=1 \wedge \hat\type=1] = \Pr[\type=0 \wedge \hat\type=1]  \textrm{, then } A \textrm{ can play any } y\in[0,1]\cr
\textrm{For $B$:} &\cr
& \textrm{ if } \rho_0 >  x+y-1 \textrm{ then } p = 1\cr
& \textrm{ if } \rho_0 <  x+y-1 \textrm{ then } p=0\cr
& \textrm{ if } \rho_0 =  x+y-1 \textrm{ then } B \textrm{ can play any } p\in[0,1]\cr
& \textrm{ if } \rho_1 >  x+y-1 \textrm{ then } q = 1\cr
& \textrm{ if } \rho_1 <  x+y-1 \textrm{ then } q=0\cr
& \textrm{ if } \rho_1 =  x+y-1 \textrm{ then } B \textrm{ can play any } q\in[0,1]
\end{flalign*}
Using $A$'s best response strategies, we look at the $[0,1]\times [0,1]$ square of possible choices for $p$ and $q$, and denote two lines on this square. 
\begin{eqnarray*}
&\ell_1 : & q = 1- \tfrac {D_0}{D_1} p \textrm{, (i.e., $D_0p=D_1(1-q)$)}\cr
&\ell_2 : & q = \tfrac {D_0}{D_1} (1-p)\textrm{, (i.e., $D_0(1-p)=D_1q$)}
\end{eqnarray*}
The two lines are parallel (if $D_0 > D_1$) or collapse into a single line when $D_0=D_1 = \tfrac 1 2$. The best response analysis for $A$ means that if $(p,q)$ is a point above the $\ell_1$-line then $A$ sets $x=1$, and if $(p,q)$ is below the $\ell_1$-line then $A$ sets $x=0$. Similarly, if $(p,q)$ is above the $\ell_2$-line then $y=1$, and if $(p,q)$ is below the $\ell_2$-line then $y=0$. The strategy space for $B$ is shown in Figure~\ref{fig:strategy_space_for_B}.
}
\begin{figure}[t1]
\centering
\includegraphics[scale=0.35]{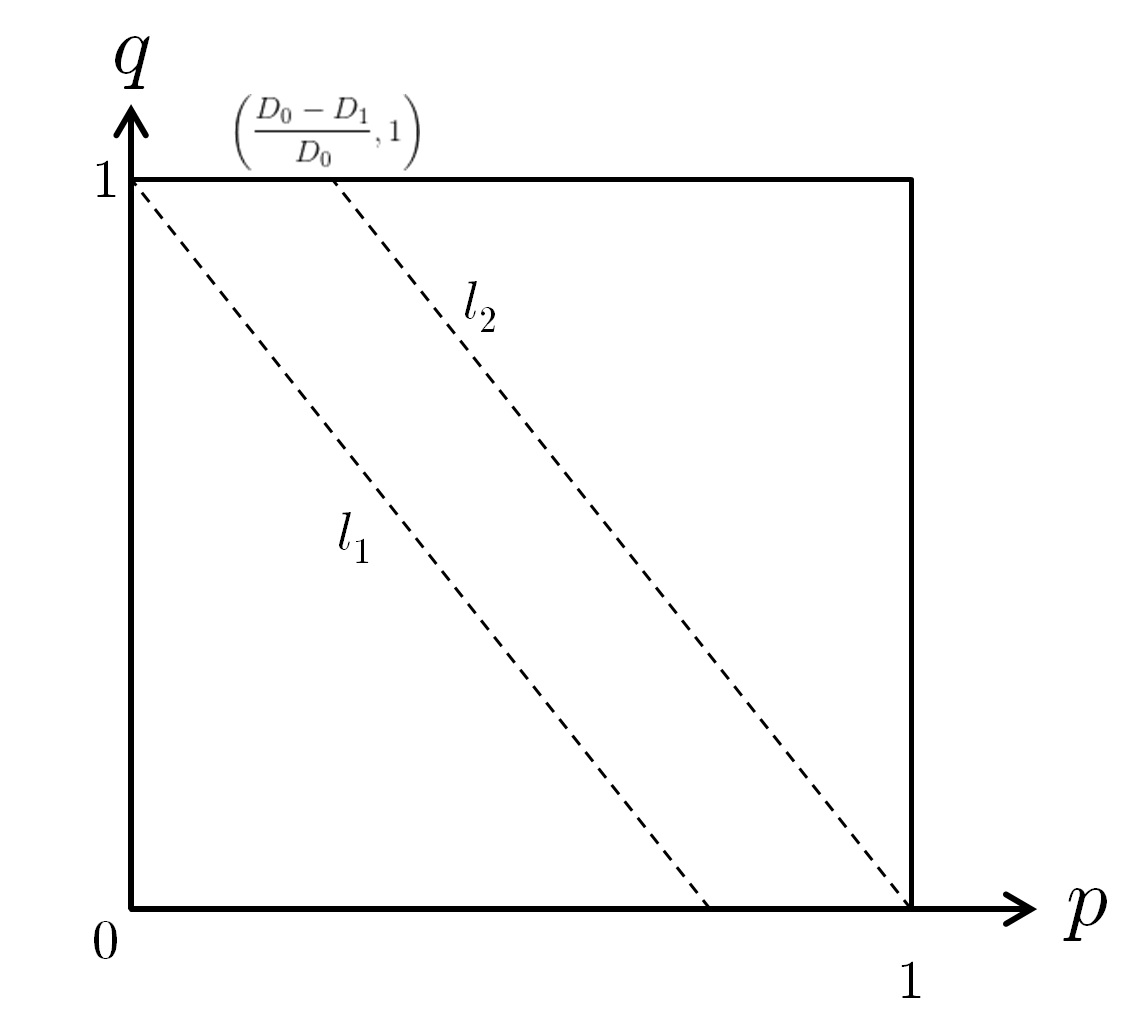}
\caption{\label{fig:strategy_space_for_B} The strategy space for agent $B$ and the corresponding best response of $A$. When $B$ plays a $(p,q)$-point above line $\ell_1$, then $A$ sets $x=1$; When $B$ plays a $(p,q)$-point above line $\ell_2$, then $A$ sets $y=1$.}
\end{figure}

We now turn our attention to the utility functions of $B$. The utility of $B$ of type $\type=0$ is
\[u_{B,0} =  p\cdot(\rho_0 - x) + (1-p)(-1+y)\]
and the utility of $B$ of type $\type=1$ is
\[u_{B,1} =  q\cdot(\rho_1 - y) + (1-q)(-1+x)\]
which means that $B$'s best response strategies are:
\begin{flalign*}
& \textrm{ if } \rho_0 >  x+y-1 \textrm{ then } p = 1\cr
& \textrm{ if } \rho_0 <  x+y-1 \textrm{ then } p=0\cr
& \textrm{ if } \rho_0 =  x+y-1 \textrm{ then } B \textrm{ can play any } p\in[0,1]\cr
& \textrm{ if } \rho_1 >  x+y-1 \textrm{ then } q = 1\cr
& \textrm{ if } \rho_1 <  x+y-1 \textrm{ then } q=0\cr
& \textrm{ if } \rho_1 =  x+y-1 \textrm{ then } B \textrm{ can play any } q\in[0,1]
\end{flalign*}

Using the best response strategies of both $A$ and $B$, we can analyze the game's potential BNEs. First we cover the simple case.

If $\max \{\rho_0,\rho_1\} > 1$: then at least one coupon have value strictly greater than $1$ and so one of the two types of $B$ agents strictly prefers deviating to playing deterministically. Wlog this type is the $t=0$ type, and so in any BNE of the game we have that $\Pr[\sigma^*_B(0)=0]=1$. 

The interesting case is when $\max\{\rho_0,\rho_1\} \leq 1$ and since we assume $\rho_0\neq\rho_1$ then for some type of $B$ agent the value of the coupon is strictly $<1$. (This intuitively makes sense~--- the coupon game becomes interesting only when $B$'s value for the coupon is below the max-payment from $B$ to $A$, hence $B$ has incentive to hide her true type.) 

We continue with a case analysis as to the potential BNE strategies of $B$.
\begin{itemize}
\item Strictly above the $\ell_2$ line, where $D_1q>D_0(1-p)$.\\
This means $B$ plays s.t. $D_1\Pr[\sigma^*_B(1)=1] > D_0 \Pr[\sigma^*_B(0)=1]$, and as a result
\[D_1\Pr[\sigma^*_B(1)=0] = D_1 - D_1q < D_0 - D_0(1-p) = D_0\Pr[\sigma^*_B(0)=0]\]
Therefore, given that $A$ observes any signal $\hat\type\in\{0,1\}$ it is more likely that a $B$ agent of type $\hat\type$ sent that this signal. So $A$ responds to such strategy by playing deterministically $\Pr[\sigma^*_A(\hat\type)=\hat\type]=1$ for any signal $\hat\type\in\bits$. As $A$ prefers to play $x=y=1$ and some type of $B$ agent has coupon valuation $<1$ then that type deviates (so either $p=0$ or $q=0$), and so the BNE strategy of $B$ cannot be above the $\ell_2$ line.

\item Strictly below the $\ell_2$ line, where $D_1q < D_0(1-p)$.\\
This means $B$ plays s.t. $D_1\Pr[\sigma^*_B(1)=1] < D_0 \Pr[\sigma^*_B(0)=1]$. So the $\hat\type=1$ signal is more likely to come from a $\type=0$ type agent, and so $A$'s best response is to set $(1-y)=\Pr[\sigma^*_A(1)=0]=1$. We thus have that $x+y-1 \leq 0$ whereas $\rho_0,\rho_1 > 0$. Hence $B$ deviates to playing $(p,q)=(1,1)$, and so the BNE of $B$ cannot be below the $\ell_2$ line.

\item On the $\ell_2$ line, where $D_1q = D_0(1-p)$.\\
This means $B$ plays s.t. $D_1\Pr[\sigma^*_B(1)=1] = D_0 \Pr[\sigma^*_B(0)=1]$, and as a result
\[D_1\Pr[\sigma^*_B(1)=0] = D_1 - D_1q < D_0 - D_0(1-p) = D_0\Pr[\sigma^*_B(0)=0]\] assuming $D_1<D_0$ (the special case where $D_0=D_1=\tfrac 1 2$ will be discussed later). And so when $A$ views the $\hat\type=0$ signal it is more likely that the type of $B$ agent is $t=0$ so $x=\Pr[\sigma^*_A(0)=0]=1$; whereas when $A$ views the $\hat\type=1$ signal both types of $B$ agents are as likely to send the signal, so $A$ is indifferent as to the value of $y = \Pr[\sigma^*_A(1)=1]$.\\
Since $\rho_0\neq \rho_1$, then by setting the single parameter $y$, $A$ can makes at most one of the two types of $B$ agent indifferent, while the other type plays a pure strategy. In other words, $B$'s BNE strategy can only be one of the two extreme point: $(p,q)=(1,0)$ or $(p,q) = (\tfrac {D_0-D_1}{D_1}, 1)$. The pure strategy of the non-indifferent type is determined by the relation between $\rho_0 \gtrless \rho_1$. So the possible BNEs are:
\begin{eqnarray*}
\textrm{If } \rho_0 > \rho_1, &&\sigma_A^* = (1,y) \textrm{ with } y\in [\rho_1,\rho_0],~ \sigma_B^*=(1,0)  \cr
\textrm{If } \rho_0 = \rho_1, && \sigma_A^* = (1,\rho_0), ~\sigma_B^*=(p,q) \textrm{ with } (p,q) \in \ell_2 \cr
\textrm{If } \rho_0 < \rho_1, && \sigma_A^* = (1, \rho_0), ~ \sigma_B^*= (\tfrac {D_0-D_1} {D_0}, 1) 
\end{eqnarray*}
In the special case where $D_0=D_1$ (the two lines are the same one) then $(p=1,q=0)$ and $(p=0,q=1)$ are both BNE regardless with $A$ BNE strategy being any $(x,y)$ satisfying with $\min\{\rho_0,\rho_1\} \leq x+y-1 \leq \max\{\rho_0,\rho_1\}$.
\end{itemize}

Observe that in the case with $\rho_0=\rho_1<1$, in a BNE, $B$ may play any strategy on the $\ell_2$-line and $A$ makes both types of $B$ agent indifferent to the value of $p,q$ by setting $y=\rho_0=\rho_1$. Since  the line $p=q$ (i.e., agent $B$ plays Randomized Response) does intersect the $\ell_2$ line at $(D_0,D_0)$, it is possible that $B$ plays Randomized Response (with $\epsilon = \ln(D_0/D_1)$). (And if $\rho_0=\rho_1=1$ then $B$ may play any $(p,q)$ on $\ell_2$ or above it whereas $A$ plays $x=y=1$.)
\end{proof}

At the BNE, $B$ plays in a way where the $\hat\type=0$ signal leads $A$ to play the same way she plays in the benchmark game (with no coupon) -- to always play $\tilde\type=0$, because $\Pr[\type = 0 ~|~ \hat\type=0] > \Pr[\type=1 ~|~ \hat\type=0]$. However, given the signal $\hat\type=1$, it holds that $\Pr[\type=0 ~|~ \hat\type =1] = \Pr[\type=1 ~|~ \hat\type=1]$ (since $B$ BNE strategy is on the $\ell_2$ line). In other words, after viewing the $\hat\type=1$ signal, $A$ has posterior belief on $B$'s type of $(\tfrac 1 2,\tfrac 1 2)$. We comment that if $B$ plays the strategy $(1,0)$, then this last statement is vacuous since the $\hat\type=1$ signal is never sent.

Observe, when $B$ plays any strategy on the $\ell_2$ line, the utility that $A$ gets from using the Nash-strategy of $(1, y)$ is $D_0p +D_0(1-p) = D_0$. In other words, moving from the benchmark game to this more complicated coupon game gives $A$ no additional revenue. In fact, the only agent that gains anything is $B$. In the benchmark game $B$'s utility is $-D_0$. In the coupon game, $B$'s utility is $D_0(\rho_0-1)$ when $\rho_0\geq \rho_1$, or $D_0(\rho_0-1) + D_1(\rho_1-\rho_0)$ when $\rho_1>\rho_0$.
}
\ifx \fullversion\undefined
\else
\proofTheoremIdentityMatrix
\fi
\ifx \fullversion \undefined \vspace{-0.4cm} \fi
\subsection{Continuous Coupon Valuations}
\label{subsec:diff_continuous_rho}
\ifx \fullversion \undefined \vspace{-0.2cm} \fi
We now consider the same game with the same payments, but under a different setting. Whereas before we assumed the valuations that the two types of $B$ agents have for the coupon are fixed (and known in advance), we now assume they are not fixed. In this section we assume the existence of a continuous prior over $\rho$, where each type $\type \in \{0,1\}$ has its own prior, so $\CDF_0(x) \stackrel {\rm def} = \Pr[\rho < x ~|~ t=0]$ with an analogous definition of $\CDF_1(x)$. We use $\CDF_B$ to denote the cumulative distribution function of the prior over $\rho$ (i.e., $\CDF_B(x) = \Pr[\rho < x] = D_0\CDF_0(x)+D_1\CDF_1(x)$). We assume the $\CDF$ is continuous and so $\Pr[\rho=y]=0$ for any $y$. Given any $z \geq 0$ we denote $\CDF_B^{-1}(z)$ the set $\{y :~ \CDF_B(y)=z\}$.
\newcommand{\proofDefferedContinuousValuationsAPX}{
\ifx \fullversion\undefined
Due to space constraints, this analysis is deferred to Appendix~\ref{apx_sec:matching_pennies}.
\else \ifx \cameraready \undefined
\else
the full version of the paper.
\fi
}
\newcommand{\matchingPenniesContinuousValuations}{
In every BNE $(\sigma_A^*,\sigma_B^*)$ of the coupon game with identity payments, where $D_0 \neq D_1$ and the valuations of the $B$ agents for the coupon are taken from a continuous distribution over $[0,\infty)$, the BNE-strategies are as follows.
\begin{itemize}
\item Agent $A$ always plays $\tilde\type=0$ after viewing the $\hat\type=0$ signal (i.e.,  $\Pr[\sigma_A^*(0) =0]=1$); and plays $\tilde\type=1$ after viewing the $\hat\type=1$ signal with probability $y^*$ (i.e., $\Pr[\sigma_A^*(1)=1] = y^*$), where $y^*$ is any value in $\CDF_B^{-1}(D_1)$ when $\Pr[\rho<1]\geq D_1$ and $y^*=1$ when $\Pr[\rho<1]<D_1$.
\item Agent $B$ reports truthfully (sends the signal $\hat\type=\type$) whenever her valuation for the coupon is greater than $y^*$, and lies (sends the signal $\hat\type=1-\type$) otherwise. That is, for every $t \in \{0, 1\}$ and $\rho \in [0,\infty)$, we have that if $\rho > y^*$ then $\Pr[\sigma_B^*(\type) = \type]=1$ and if $\rho < y^*$ then $\Pr[\sigma_B^*(\type)= \type]=0$.
\end{itemize}}
\begin{theorem}
\label{thm:matching-pennies-continuous-valuations}
\matchingPenniesContinuousValuations
\end{theorem}
\proofDefferedContinuousValuationsAPX
\newcommand{\continuousCouponValuations}{
\begin{proof}
We assume $B$'s parameters are sampled as follows. First, we pick a type $\type$ s.t. $\Pr[\type=1]=D_1$ and $\Pr[\type=0] = D_0$. Then, given $\type$ we sample $\rho\leftarrow \PDF_\type$, where $\Pr[\rho \leq 0]=0$ for both types. And while $A$ knows $D_0,D_1, \PDF_0$ and $\PDF_1$, $A$ does not know $B$'s realized type and valuation.

We apply the same notation from before, denoting a strategy $\sigma_B$ of $B$ using $p$ and $q$ (where $p = \Pr[\sigma_B(0)=0]$ and $q=\Pr[\sigma_B(1)=1]$), and denoting a strategy $\sigma_A$ of $A$ using $x$ and $y$ (where $x=\Pr[\sigma_A(0)=0]$ and $y=\Pr[\sigma_A(1)=1]$).

The utility function of $B$ remains the same:
\[ u_{B,0,\rho} = p(\rho - x)+(1-p)(-1+y)~,~\qquad u_{B,1,\rho} = q(\rho-y) + (1-q)(-1+x) \]
So $B$'s best response to any strategy of $(x,y)$ of $A$ is given by
\[ \sigma_B^{br}(\rho,\type) = \begin{cases} \type, & \textrm{ if } \rho > x+y-1 \cr 1-\type, & \textrm{ if } \rho < x+y-1 \end{cases}\]
We call such a strategy a \emph{threshold strategy} characterized by a parameter $T$, where any agent whose $\rho < T$ plays $\hat\type=1-\type$ and any agent with $\rho >T$ plays $\hat\type=\type$.\footnote{Since $\rho$ is sampled from a continuous distribution, then the probability of the event $\rho=T$ is $0$.} Clearly, in any BNE, $B$ follows a threshold strategy for some value of $T$.

Therefore, since $A$'s BNE strategy is best response to $B$'s BNE strategy, it suffices to consider $A$'s best response against a threshold strategy. Given that $B$ follows a threshold strategy with threshold $T$ we have that $A$'s utility function is
\begin{eqnarray*}
& u_A=  & x D_0 (1-\CDF_0(T))  + (1-x)D_1 \CDF_1(T) \cr
&& + yD_1 (1-\CDF_1(T)) + (1-y)D_0 \CDF_0(T)\cr
&& = \CDF_B(T) + x(D_0 - \CDF_B(T)) + y(D_1-\CDF_B(T)) 
\end{eqnarray*} 
where we use the notation $\CDF_B = D_0 \CDF_0 + D_1\CDF_1$. As $A$ maximizes her strategy, we have that $A$ sets $x>0$ only if $\CDF_B(T)\leq D_0$. Similarly, $y>0$ only if $\CDF_B(T)\leq D_1$. Since $D_1\leq D_0$ we only have three cases to consider.
\begin{itemize}
\item If $\CDF_B(T) <  D_1$: In this case $A$'s best response is to set $x=y=1$ and $B$'s best-response to $(1,1)$ is to set the threshold parameter $T = x+y-1 = 1$. (So every $B$ agent with $\rho<1$ determinisitically sends the signal $\hat\type = 1-\type$ and any $B$ agent with $\rho > 1$ sends the signal $\hat\type=\type$.) Clearly, if it holds that $\CDF_B(1) < D_1$, i.e., that the probability of a random $B$ agent to have $\rho\leq 1$ is less than $D_1$, then we have a BNE.
\item If $\CDF_B(T) > D_1$: In this case $A$'s best response sets $y=0$ and $x\in [0,1]$. As $B$ is playing best response to $A$'s strategy then it means that the threshold parameter is set to $T=x-1 \leq 0$. (And so all $B$ agents have coupon valuation $\rho>0$ we have that all $B$ agents determinisitically send the signal $\hat\type=\type$.) But for such $T$ we have that $\CDF_B(T) \leq \CDF_B(0) = 0 < D_1$ we get an immediate contradiction.
\item If $\CDF_B(T) = D_1$: In this case $A$ sets $x = 1$ and is indifferent to the choice of $y$. Observe that $B$'s best response to $A$'s strategy of $(1,y)$ is to set the threshold parameter to $T=y$. We have that this is indeed a BNE if $y \in \CDF_B^{-1}(D_1)$. Assuming uniqueness to the inverse of $\CDF_B$ then $\sigma_A^*=(1,y^*)$ with $y^*=\CDF_B^{-1}(D_1)$ is $A$'s BNE strategy, and $B$'s BNE strategy is a threshold strategy with the threshold parameter set to $T=y^*$. We comment that in the case where $D_0 = D_1$ and $A$ is indifferent to the choice of $x$ as well, the BNE strategy of $A$ is defined using any $x^*,y^* \in [0,1]$ that satisfy $x^*+y^* \in \CDF_B^{-1}(D_1)$.
\end{itemize}

\cut{
So for a fixed $x,y$, we have that $\Pr[\type = 0 \textrm{ and } \rho < x+y-1] = D_0 \int_0^{x+y-1} \PDF_0(t) dt = D_0 \CDF_0(x+y-1)$ and similarly $\Pr[\type=1 \textrm{ and } \rho<x+y-1] = D_1 \CDF_1(x+y-1)$.
Since $\rho$ is sampled from a continuous distribution, we have that $\Pr[\rho=x+y-1] = 0$. Therefore, for a given $(x,y)$, the $A$ agent can quantify the probability of getting a $\hat\type=0$ signal and a $\hat\type=1$ signal. Hence, the utility of $A$ is
\begin{eqnarray*}
& u_A=  & x D_0 (1-\CDF_0(x+y-1))  + (1-x)D_1 \CDF_1(x+y-1) \cr
&& + yD_1 (1-\CDF_1(x+y-1)) + (1-y)D_0 \CDF_0(x+y-1)\cr
&& = xD_0+yD_1 + D_0(-x+1-y)\CDF_0(x+y-1) + D_1(-y+1-x)\CDF_1(x+y-1) \cr
&& = xD_0+yD_1 - (x+y-1)\CDF_B(x+y-1) \cr
&& = 1- D_1x - D_0y + (x+y-1) \left(1-\CDF_B(x+y-1)\right)
\end{eqnarray*} 
with $\CDF_B = D_0 \CDF_0 + D_1\CDF_1$.

To analyze what it $A$'s utility-maximizing strategy, we need to solve the following problem.
\begin{eqnarray*}
& \textrm{maximize} & 1- D_1x - D_0y + z(1-\CDF_B(z)) \cr
& \textrm{s.t.} & x,y \in [0,1] \cr
&& z = x+y-1
\end{eqnarray*}
Given that $(x^*,y^*,z^*)$ maximizes the above quantity, we use the fact that $D_0 \geq D_1$ to deduce that: if $1+z^*<1$ then it must be that $x^*=1+z^*$ and $y^*=0$, and if $1+z^*\geq 1$ then it must hold that $x^*=1$ and $y = z^*$. Therefore, this maximization problem is equivalent to maximizing a univariate function
\begin{eqnarray*}
& \textrm{maximize} & \begin{cases} &1- D_1(1+z) + z(1-\CDF_B(z)) \cr &\textrm{s.t. } z\in[-1,0] \end{cases} \textrm{ or } \begin{cases}  &1- D_1 - D_0z + z(1-\CDF_B(z)) \cr & \textrm{s.t. } z\in [0,1]\end{cases}
\end{eqnarray*}
Since we assume that no agent has negative valuation for the coupon, we have that $\CDF_B(z) = 0$ if $z<0$. So the first optimization problem turns out to maximize $D_0 - D_1 z + z = D_0 +  D_0z$ on the range $z \in [-1,0]$. As this function increases with $z$ it is easy to see that the maximum is $D_0$. Therefore it suffices to consider the latter optimization problem
\begin{eqnarray*}
& \textrm{maximize} & D_0 - D_0z + z(1-\CDF_B(z)) = D_0 + z(D_1 - \CDF_B(z))\cr
& \textrm{s.t.} & z \in [0,1] 
\end{eqnarray*}
The maximum here is $\geq D_0$ as $z=0$ gives a value of $D_0$. This makes intuitive sense as $A$ can always ignore the signal $\hat\type$ and gain $D_0$ by always playing $\tilde\type=0$.

Given that $z^*$ maximized this function, $A$ sets $(x^*,y^*) = (1,z^*)$. (Unless $D_0=D_1=\tfrac 1 2$ in which case $A$ can set any $(x^*,y^*)$ s.t. $x^*+y^*=1+z^*$.) Any $B$ agent whose valuation for the coupon satisfies $\rho > 1+z^*$ signals truthfully $\hat\type=\type$ and any agent whose valuation is $\rho < 1+z^*$ signals $\hat\type=1-\type$.

The function $z(D_1-\CDF_B(z))$ is a differentiable function, whose derivative is $D_1 - \CDF_B(z) - z\PDF_B(z)$. The derivative for $z=0$ is positive and it is a decreasing function of $z$. So we have that if it exists, then the point $z^*\in[0,1]$ s.t. $\CDF_B(z^*)+z^*\PDF_B(z^*) = D_1$ maximizes $A$'s utility, otherwise we set $z^*=1$. In any case, $z^*>0$ so $A$ gains strictly more than $D_0$, which is her gain in the straw-man game. I.e., in the continuous valuations model, $A$ is better off playing the coupon game.
}
\end{proof}
Observe that in this game, from $A$ perspective, without knowing the realized value of $\rho$, it appears that $B$ is playing a randomized strategy. Furthermore, should the coupon valuation and the type be chosen independently (i.e. $\PDF_0=\PDF_1$) then $A$ views $B$'s strategy as Randomized Response --- since for a randomly chosen $\rho$ it holds that $\Pr[\sigma(\rho,0) = 0] = \Pr[\sigma_B(\rho,1) = 1]$. In that case the behavior of $B$ preserves $\epsilon$-differential privacy for
\[  \epsilon = \ln \left( \frac {1-\CDF_B(y^*)}{\CDF_B(y^*)}\right) = \ln\left(\frac {D_0}{D_1}\right)  \]

\cut{
One can plug in various functions as $\PDF$s. In particular, when $\rho$ is sampled uniformly from the $[0,1]$-interval (regardless of the value of $\type$) then $z^* = \tfrac 1 2 D_1$, so $\epsilon = \ln( \tfrac{2-D_1}{D_1})$. When $\rho$ is sampled from the exponential distribution, i.e. $\PDF(z) = e^{-z}$, then $z^*$ is the value satisfying $(1-z)e^{-z} = D_0$. }
}
\continuousCouponValuations
\cut{
\subsection{The Game with a General $2\times 2$ Payment Matrix}
\label{subsec:arbitrary_matrix}

\os{Other than the fact that this is a good preview for the next section, don't really see why I should put it here. The results are nice, but uninteresting.}

Finally, we consider an extension of the game, where we replace the identity-matrix payments with general payments. Indeed, there is no intuitive reason why, from $A$'s perspective, realizing that $B$ has type $\type=0$ is worth just as much as finding $B$ has type $\type=1$. After all, it could be that type $\type=0$ are the healthy people and $\type=1$ represents having some medican condition, so finding a person of $\type=1$ should be more worthwhile for $A$. In our new payment matrix, we still require that $A$ gains utility if she correctly guesses $B$'s type, and loses if she accuses $B$ of being of the wrong type. In other words, we consider payments of the form
\[ M = \left[ \begin{array}{c|c}
M_{0,0} & -M_{0,1}\cr
\hline
-M_{1,0} & M_{1,1}
\end{array}\right] \]
where $A$ is the row player and $B$ is the column player. We assume $M_{0,0}, M_{0,1}, M_{1,0}, M_{1,1}$ are all non-negative.

\cut{
For example, we can imagine that being of type-$1$ is rather embarrassing. Therefore, $M_{1,1}$ can be much larger than $M_{0,0}$, but similarly $M_{1,0}$ is also probably larger than $M_{0,1}$. (Falsely accusing $B$ of being of the embarrassing type is costlier than falsely accusing a $B$ of type $1$ of belonging to the non-embarrassing majority.) And so, as leading example, we pick some parameter $a>1$ and think of $M$ as the payment matrix $\left[ \begin{array}{c | c}
1 & -1\cr
\hline -a & a 
\end{array}\right]$
}
\paragraph{The benchmark game.} First, consider the game where $B$ has no moves, and $A$ just gets to accuse $B$ of belonging to type $\tilde\type$. $A$ knows that $B$ plays the $0$-column w.p. $D_0$ and the $1$-column w.p. $D_1$, and so she plays the row that maximizes
\[ \tilde\type=\arg\max\{ D_0 M_{0,0} - D_1 M_{0,1}, D_1 M_{1,1} - D_0 M_{1,0} \} \] We denote this row as $r^*$. An equivalent formulation of this is that $r^*=0$ if $D_0(M_{0,0}+M_{1,0}) > D_1(M_{0,1}+M_{1,1})$ and $r^*=1$ otherwise.

\paragraph{The coupon game.} The game itself remains the same. 
\begin{flalign*}
 \textrm{For $B$: }  & ~~p = \Pr[\hat\type= 0 ~|~ \type =0] \textrm{, and }~  q = \Pr[\hat\type=1 ~|~ \type = 0] \cr
 \textrm{For $A$: } 
 & ~~x =\Pr[\tilde\type = 0 ~|~ \hat\type = 0] \textrm{, and }~ y=\Pr[\tilde\type=1~|~\hat\type=1]
\end{flalign*}
The new utility functions are generalization of the utility functions from before.
\begin{flalign*}
u_A =& D_0\left[ M_{0,0}p x - M_{1,0} p(1-x) + M_{0,0} (1-p) (1-y) - M_{1,0} (1-p) y  \right] \cr 
&+ D_1\left[ M_{1,1}q y - M_{0,1} q(1-y) + M_{1,1} (1-q) (1-x) - M_{0,1} (1-q) x  \right] \cr
=& x \left( M_{0,0} D_0 p -M_{0,1}D_1(1-q)\right) + (1-x) \left(-M_{1,0}D_0p + M_{1,1} D_1 (1-q)\right) \cr
&+ y \left(  -M_{1,0}D_0(1-p) + M_{1,1}D_1q  \right) + (1-y) \left( M_{0,0}D_0(1-p)-M_{0,1}D_1q  \right) \cr
u_B =& D_0 u_{B,0} + D_1 u_{B,1} \textrm{, where we define }\cr
u_{B,0} =& p ( \rho_0 - M_{0,0}x + M_{1,0}(1-x) ) + (1-p) ( -M_{0,0}(1-y) + M_{1,0} y ) \cr
u_{B,1} =&  q ( \rho_1 - M_{1,1}y + M_{0,1}(1-y) ) + (1-q) ( -M_{1,1}(1-x) + M_{0,1} x )
\end{flalign*}

And so, $x$ is determined by the point $(p,q)$ being above or below the line \[\ell_0: (M_{0,0} + M_{1,0})D_0p = (M_{0,1}+M_{1,1})D_1(1-q)\] and $y$ is determined by the point $(p,q)$ being above or below the line \[\ell_1: (M_{0,0} + M_{1,0})D_0(1-p) = (M_{0,1}+M_{1,1})D_1q\] Again, the two lines are parallel (or identical), with $(0,1)\in \ell_0$ and $(1,0)\in \ell_1$. Yet, this time, it is not clear which line is above the other. Looking at $p=1$, we have that the point $(1, 1- \frac {D_0(M_{0,0}+M_{1,0})} {D_1(M_{0,1}+M_{1,1})})$ is on the $\ell_0$-line, and the point $(1,0)$ is on the $\ell_1$ line. So, when $r^* = 0$ then the $\ell_0$-line is below the $\ell_1$-line, and vice-versa. In other words, the $\ell_{r^*}$-line is below the $\ell_{1-r^*}$-line.

The two parallel line partition the $[0,1]\times[0,1]$ square into three regions: below the both lines (where $A$ sets $x+y=0$), between the two lines (where $A$ sets $x+y=1$, and above the two lines (where $A$ sets $x+y=2$). On the lower of the two lines $A$ may set $x+y$ to be any number in $[0,1]$, and the higher of the two lines $A$ may set $x+y$ to be any number in $[1,2]$.\footnote{Here we use the assumption that all of $\{M_{0,0}, M_{0,1}, M_{1,0}, M_{1,1}\}$ are positive.}	

Similar to before, $B$'s best response is determined by the relationship between $\rho_t$ and $x+y-1$. The relationship $\rho_0 \gtrless (M_{0,0}+M_{1,0})(x+y-1)$ determines whether $p=1$ or $p=0$; the  relationship $\rho_1 \gtrless (M_{0,1}+M_{1,1})(x+y-1)$ determines whether $q=1$ or $q=0$. Observe that when $x+y \leq 1$ then both agents prefer to take the coupon (we assume the valuation for the coupon is positive), and when $x+y = 2$ then only agents with large enough valuation for the coupon take it.

\paragraph{Analyzing the Game in the case of fixed valuations.} We first consider the case where the valuations of the two types of $B$ agent for the coupon are fixen and known in advance. In this case, the analysis mimic the analysis in Section~\ref{subsec:coupon_game}
First consider the simple case where both $\rho_0 \geq M_{0,0}+M_{1,0}$ and $\rho_1 \geq M_{0,1}+M_{1,1}$. In this case, it is clear that $p=q=1$ and $x=y=1$ is the only BNE (if $\rho_0 = M_{0,0}+M_{1,0}$ or $\rho_1 = M_{0,1}+M_{1,1}$ then $B$ may play any $p$ or any $q$, as long as $x=y=1$ is a best response to $(p,q)$. Meaning, as long as $(p,q)$ are above both lines.

Now suppose that at least on type of agent, $b\in \{0,1\}$ has valuation $0 < \rho_b < M_{0,b}+M_{1,b}$. In such a case, and just like before, we claim that the BNE strategy of $B$, $(p^*,q^*)$, must lie on the highest line. Indeed, everywhere else, $A$'s best response is to set $(x,y)$ s.t. either $x+y=2$ or $x+y \leq 1$; $B$'s best response to $x+y=2$ is to set either $p$ or $q$ to $0$, and for any $x+y\leq 1$ is to set both $p=1$ and $q=1$. 

Therefore, $A$'s strategy is to set $x+y$ s.t. either $B$ of type-$0$ is indifferent, or $B$ of type-$1$ is indifferent. In other words, either
\[ x+y = \frac {\rho_0}{M_{0,0}+M_{1,0}} + 1 ~\textrm{, or }~~x+y = \frac {\rho_1}{M_{0,1}+M_{1,1}} + 1\]
$A$ makes the type-$b$ agent indifferent by setting one of the $x$ or $y$ to be $1$ (the one corresponding to~$r^*$) and the other to $\frac {\rho_b}{M_{0,b}+M_{1,b}}$. That is, $A$'s strategy sets $B$ agent of type-$b$ to be indifferent, whereas type-$(1-b)$ plays $\begin{cases}
1&\textrm{, if } \frac {\rho_{1-b}}{\rho_b} >  \frac {M_{0,1-b}+M_{1,1-b}}{M_{0,b}+M_{1,b}} \cr
\textrm{Anything on }[0,1]&\textrm{, if } \frac {\rho_{1-b}}{\rho_b}=  \frac {M_{0,1-b}+M_{1,1-b}}{M_{0,b}+M_{1,b}} \cr
0&\textrm{, if } \frac {\rho_{1-b}}{\rho_b}<  \frac {M_{0,1-b}+M_{1,1-b}}{M_{0,b}+M_{1,b}}
\end{cases}$.\\
If out of the three conditions it happens that the equality holds, any $(p,q)$  on the highest line is a BNE strategy for $B$. In that case, it is possible that $B$ plays Randomized Response --- with $(p,q)$ being on the intersection of the line $p=q$ with $\ell_{1-r^*}$.

In the case where the equality does not hold, then the BNE is characterized as follows. If $\ell_1$ is the highest line, then  \[ \begin{cases} (x=1, y=\frac {\rho_1}{M_{0,1}+M_{1,1}}), (p=1,q=0), &\textrm{ if } \frac {\rho_{0}}{\rho_1} >  \frac {M_{0,0}+M_{1,0}}{M_{0,1}+M_{1,1}} \cr
(x=1, y=\frac {\rho_0}{M_{0,0}+M_{1,0}}), (p=1 - \frac {D_1}{D_0} \frac{M_{0,1}+M_{1,1}}{M_{0,0}+M_{1,0}},q=1), &\textrm{ o/w } \end{cases}\]
If $\ell_0$ is the highest line, then  \[ \begin{cases} (x=\frac {\rho_1}{M_{0,1}+M_{1,1}}, y=1), (p=1,q=1- \frac {D_0(M_{0,0}+M_{1,0})} {D_1(M_{0,1}+M_{1,1})}), &\textrm{ if } \frac {\rho_{0}}{\rho_1} >  \frac {M_{0,0}+M_{1,0}}{M_{0,1}+M_{1,1}} \cr
(x=\frac {\rho_0}{M_{0,0}+M_{1,0}},y=1), (p=0,q=1), &\textrm{ o/w } \end{cases}\]

\paragraph{Analyzing the game in the case of continuous valuations.} As before, we continue the analysis with coupon valuations sampled from a distribution. Same as in Section~\ref{subsec:diff_continuous_rho} we denote $\PDF_0$ and $\PDF_1$ as the PDFs of the distributions from which the valuations for $0$ and $1$ are sampled from. We assume of course that $\PDF_0(x)=\PDF_1(x)=0$ for any $x<0$.

As before $x = \Pr[\tilde\type=0 | ~\hat\type=0]$ and $y = \Pr[\tilde\type = 1 |~ \hat\type=1]$. Again, given a $B$ agent of type either $0$ or $1$ and coupon valuation of $\rho$, its utility function is
\begin{eqnarray*}
u_{B,0} =& p ( \rho - M_{0,0}x + M_{1,0}(1-x) ) + (1-p) ( -M_{0,0}(1-y) + M_{1,0} y ) \cr
u_{B,1} =&  q ( \rho - M_{1,1}y + M_{0,1}(1-y) ) + (1-q) ( -M_{1,1}(1-x) + M_{0,1} x )
\end{eqnarray*}
So $B$ chooses whether or not to give the signal $\hat\type=\type$ based on $\rho \stackrel {?} \gtrless (M_{0,\type}+M_{1,\type})(x+y-1)$. We thus denote for $\type\in\{0,1\}$ the threshold value $T_\type = 
(M_{0,\type}+M_{1,\type})(x+y-1)$, and we have
\begin{eqnarray*}
\Pr[\hat\type = 1 | \type = 0] & = \CDF_0(T_0)\cr
\Pr[\hat\type = 0 | \type = 1] & = \CDF_1(T_1)\end{eqnarray*}

We can now formulate $A$'s utility from this game. 
\begin{eqnarray*}
& u_A & = D_0 \left(\CDF_0(T_0) (1-y) + (1-\CDF_0(T_0))x\right) + D_1\left( \CDF_1(T_1)(1-x) + (1-\CDF_1(T_1))y  \right) \cr
&& = x D_0 + yD_1 - (x+y-1) \left(D_0\CDF_0(T_0) + D_1\CDF_1(T_1) \right) \cr
&& = 1 - y D_0 - x D_1 + (x+y-1) \left( D_0(1-\CDF_0(T_0)) + D_1(1-\CDF_1(T_1)) \right) 
\end{eqnarray*}

Therefore, as $A$ maximizes her own utility, this optimization fucntion is dependent on the variable $z$ the represents $z = x+y-1$.
\begin{eqnarray*}
&\textrm{maximize} & 1 - y D_0 - x D_1 + z \left( D_0(1-\CDF_0(T_0)) + D_1(1-\CDF_1(T_1)) \right) \cr
&\textrm{s.t.} & z = x+y-1\cr
&& x,y\in[0,1]
\end{eqnarray*}
The same analysis from before gives that we can consider just the case where $z\in [0,1]$ and $x=1, y=z$. This leaves us with 
\begin{eqnarray*}
&\textrm{maximize}_{z\in[0,1]} & u_A(z) = D_0 - z D_0 + z \left( D_0(1-\CDF_0(T_0)) + D_1(1-\CDF_1(T_1)) \right) \cr
&\Leftrightarrow \textrm{maximize}_{z\in[0,1]} & u_A(z) = D_0 + z D_1 - z \left( D_0\CDF_0(T_0) + D_1\CDF_1(T_1) \right) \cr
\end{eqnarray*}
We therefore define the function $F_B(x) = D_0 \CDF_0\big( (M_{0,0}+M_{1,0})\cdot x \big) + D_1 \CDF_1\big( (M_{0,1}+M_{1,1})\cdot x \big)$. At this point, we have the same analysis from before, except that the function $\CDF_B$ is replaced with $F_B$. In particular, $A$ finds some value $z$ s.t. from $A$'s perspective $B$ plays Randomized Response with the same $\epsilon$ as before.
\cut{
Differentiating:
\begin{flalign*}
u_A' & = D_1 -D_0\CDF_0(T_0)-D_1\CDF_1(T_1) - D_0(M_{0,0}+M_{1,0})z\PDF_0(T_0) - D_1(M_{0,1}+M_{1,1})z\PDF_1(T_1)\cr
& = D_1(1 - \CDF_1(T_1) - T_1\CDF_1(T_1)) -D_0(\CDF_0(T_0)+ T_0 \PDF_0(T_0))\cr
& = D_1(1 - \CDF_1(T_1) - T_1\CDF_1(T_1)) +D_0(1-\CDF_0(T_0)-T_0 \PDF_0(T_0)) -D_0
\end{flalign*}
Solving for $u = (0,0)$ we have that
\[ (1-\CDF_0(T_0) - T_0\PDF_0(T_0))(1-\CDF_1(T_1) - T_1\PDF_1(T_1))=1\]
}
}

\ifx \fullversion \undefined \vspace{-0.7cm} \fi
\ifx \fullversion \undefined \else \newpage \fi
\section{The Coupon Game with an Opt Out Strategy}
\label{sec:opt-out-possible}
\ifx \fullversion \undefined \vspace{-0.45cm} \fi
In this section, we consider a version of the game considered in Section~\ref{sec:matching_pennies}. The revised version of the game we consider here is very similar to the original game, except for $A$'s ability to ``opt out'' and not guess $B$'s type. 

In this section, we consider the most general form of matrix payments. We replace the identity-matrix payments with general payment matrix $M$ of the form $M = \left[ \begin{array}{c|c}
M_{0,0} & -M_{0,1}\cr
\hline
-M_{1,0} & M_{1,1}
\end{array}\right]$
with the $(i,j)$ entry in $M$ means $A$ guessed $\tilde\type=i$ and $B$'s true type is $\type=j$, and so $B$ pays $A$ the amount detailed in the $(i,j)$-entry. We assume $M_{0,0}, M_{0,1}, M_{1,0}, M_{1,1}$ are all non-negative.

Indeed, when previously considering the identity matrix payments, we assumed the for $A$, realizing that $B$ has type $\type=0$ is worth just as much as finding $B$ has type $\type=1$. But it might be the case that finding a person of $\type=1$ should be more worthwhile for $A$. For example, type $\type=1$ (the minority, since we always assume $D_0\geq D_1$) may represent having some embarrassing medical condition while type $\type=0$ representing not having it. Therefore, $M_{1,1}$ can be much larger than $M_{0,0}$, but similarly $M_{1,0}$ is probably larger than $M_{0,1}$. (Falsely accusing $B$ of being of the embarrassing type is costlier than falsely accusing a $B$ of type $1$ of belonging to the non-embarrassing majority.) Our new payment matrix still motivates $A$ to find out $B$'s true type --- $A$ gains utility by correctly guessing $B$'s type, and loses utility by accusing $B$ of being of the wrong type.

\paragraph{The ``strawman'' game.} First, consider a simple game where $B$ makes no move ($A$ offers no coupon) and $A$ tries to guess $B$'s type without getting any signal from $B$. Then $A$ has three possible pure strategies: (i) guess that $B$ is of type $0$; (ii) guess that $B$ is of type $1$; and (iii) guess nothing. In expectation, the outcome of option (i) is $D_0 M_{0,0}-D_1 M_{0,1}$ and the outcome of option (ii) is $D_1 M_{1,1} - D_0 M_{1,0}$. If the parameters of $M$ are set such that both options are negative then $A$'s preferred strategy is to opt out and gain $0$. We assume throughout this section that indeed the above holds. (Intuitively, this assumption reflects the fact that we don't make assumptions about people's type without first getting any information about them.)
So we have
\ifx \fullversion \undefined \vspace{-3mm} \fi
\begin{eqnarray}
&&\frac {M_{0,0}}{M_{0,1}} < \frac {D_1}{D_0} \textrm{ , }\qquad \textrm{ and }~~~ \frac {M_{1,1}}{M_{1,0}} < \frac {D_0}{D_1} \label{eq:strawman_game_assumptions}
\end{eqnarray}
A direct (and repeatedly used) corollary of Equation~\eqref{eq:strawman_game_assumptions} is that $\tfrac {M_{0,0}}{M_{0,1}} < \tfrac {M_{1,0}}{M_{1,1}}$.

\paragraph{The full game.} We now give the formal description of the game.
\begin{enumerate}
\addtocounter{enumi}{-1}
\item $B$'s type, denoted $\type$, is chosen randomly, with $\Pr[\type=0]=D_0$ and $\Pr[\type=1]=D_1$.
\item $B$ reports a type $\hat\type$ to $A$. $A$ in return gives $B$ a coupon of type $\hat\type$.
\item $A$ chooses whether to accuse $B$ of being of a certain type, or opting out. 
\begin{itemize}
\item If $A$ opts out (denoted as $\tilde\type=\bot$), then $B$ pays $A$ nothing.
\item If $A$ accuses $B$ of being of type $\tilde\type$ then: if $\tilde\type=\type$ then $B$ pays $M_{\type,\type}$ to $A$, and if $\tilde\type=1-\type$ then $B$ pays $-M_{1-\type,\type}$ to $A$ (or $A$ pays $M_{1-\type,\type}$ to $B$).
\end{itemize}
\end{enumerate}
Introducing the option to opt out indeed changes significantly the BNE strategies of $A$ and $B$.
\newcommand{\NEWhichIsRR}{
If we have that $D_0^2 M_{0,0}M_{1,0} = D_1^2M_{0,1}M_{1,1}$ and the parameters of the game satisfy the following condition:
\begin{align}
&0 < \rho_1M_{1,0} - \rho_0M_{1,1} < M_{0,1}M_{1,0}-M_{0,0}M_{1,1} \cr
&0 < \rho_0M_{0,1} - \rho_1M_{0,0} <  M_{0,1}M_{1,0}-M_{0,0}M_{1,1} \label{eq:condition_of_randomized_BNE}
\end{align}
then the unique BNE strategy of $B$, denote $\sigma_B^*$,  is such that $B$ plays Randomized Response: $\tfrac 1 2 \leq \Pr[\sigma_B^*(0)=0]=\Pr[\sigma_B^*(1)=1] < 1$.}
\begin{theorem}
\label{thm:NE_which_is_RR}
\NEWhichIsRR
\end{theorem} 
\noindent Proving Theorem~\ref{thm:NE_which_is_RR} is the goal of this section. 
\newcommand{\optoutGameDeferredToAPX}{ 
\ifx \cameraready \undefined
The proof itself is deferred to the next section. 
We give here, in Table~\ref{tab:Conditions_for_NE} a summary of the various BNEs of this game. The six cases detailed in Table~\ref{tab:Conditions_for_NE} cover all possible settings of the game and they are also mutually exclusive (unless some inequality holds as an equality). 
The notation in Table~\ref{tab:Conditions_for_NE} is consistent with our notation in the analysis of the game. A strategy $\sigma_B$ of agent $B$ is denoted as $(p,q)$ and a strategy $\sigma_A$ of agent $A$ is denoted as $(x_0,x_1,y_0,y_1)$. Formally, we denote $p=\Pr[\sigma_B(0)=0]$ and $q = \Pr[\sigma_B(1)=1]$; and $x_b = \Pr[\sigma_A(0)=b]$ and $y_b = \Pr[\sigma_A(1)=b]$ for $b\in\bits$. (So $A$'s opting out probabilities are $x_\bot = \Pr[\sigma_A(0)=\bot] = 1-x_0-x_1$ and $y_\bot = \Pr[\sigma_A(1)=\bot]=1-y_0-y_1$.)
\begin{table}[t]
\centering
\begin{tabular}{ | c | c | c | c |}
\hline
Case & Condition & $A$'s Strategy & $B$'s strategy \cr
No.& & (always: $x_1=y_0=0$) &\cr
\hline
$1$& $\rho_0 \geq M_{0,0}+M_{1,0}$ and~~~ $\rho_1\geq M_{0,1}+M_{1,1}$ & $(x_0,y_1)=(1,1)$ & $(1,1)$\cr\hline
$2$& $\rho_0 \leq M_{0,0}$ and~~~ $\tfrac{\rho_0}{\rho_1}\leq \tfrac{M_{0,0}}{M_{0,1}}$ & $(x_0,y_1) = (\tfrac {\rho_0}{M_{0,0}},0)$ & $P_1=(0,1)$ \cr\hline
$3$& $0 \leq \rho_0-M_{0,0}\leq M_{1,0}$ & $(x_0,y_1) = (1,\tfrac {\rho_0-M_{0,0}}{M_{1,0}})$ & $P_2$ \cr
& $\rho_1M_{1,0}-\rho_0M_{1,1} \geq M_{0,1}M_{1,0} -M_{0,0}{M_{1,1}}$ & &\cr\hline
$4$& $\rho_1 \leq M_{1,1}$ and~~~ $\tfrac{\rho_0}{\rho_1} \geq \tfrac{M_{1,0}}{M_{1,1}}$ & $(x_0,y_1)=(0,\tfrac{\rho_1}{M_{1,1}})$ & $P_3=(1,0)$ \cr\hline
$5$& $0 \leq \rho_1 - M_{1,1} \leq M_{0,1}$ & $(x_0,y_1)=(\tfrac {\rho_1-M_{1,1}}{M_{0,1}},1)$ & $P_4$ \cr
& $\rho_0M_{0,1}-\rho_1M_{0,0} \geq M_{0,1}M_{1,0}-M_{0,0}M_{1,1}$ & &\cr\hline
$6$& $0 \leq \rho_1M_{1,0} - \rho_0M_{1,1} \leq M_{0,1}M_{1,0}-M_{0,0}M_{1,1}$ & See below & $P_5$ \cr
& $0 \leq \rho_0M_{0,1} - \rho_1M_{0,0} \leq   M_{0,1}M_{1,0}-M_{0,0}M_{1,1}$ & & \cr\hline
\end{tabular}
\caption{\label{tab:Conditions_for_NE}\small The various conditions under which we characterize the BNEs of the Game. We use the notation $P_2 = (1 - \frac {D_1M_{1,1}}{D_0 M_{1,0}} , 1)$, $P_4 = (1, 1-\frac {D_0 M_{0,0}}{D_1M_{0,1}})$, and $P_5 = \left( \frac {D_0D_1M_{0,1}M_{1,0} - D_1^2 M_{0,1}M_{1,1}} {D_0D_1 M_{0,1}M_{1,0} - D_0 D_1 M_{0,0}M_{1,1}} , \frac {D_0D_1M_{0,1}M_{1,0} - D_0^2 M_{0,0}M_{1,0}} {D_0D_1 M_{0,1}M_{1,0} - D_0 D_1 M_{0,0}M_{1,1}} \right)$. The point $P_5$ lies at the intersection between two specific lines, and points $P_2$ and $P_4$ are the intersection points of each of those lines with the $(q=1)$-line and $(p=1)$-line resp. In case $6$, the strategy of $A$ is given by
$(x_0,y_1) = {\frac 1 {M_{1,0}M_{0,1} - M_{0,0}M_{1,1}}} (-M_{1,1}\rho_0 + M_{1,0}\rho_1, M_{0,1} \rho_0 - M_{0,0}\rho_1)$.}
\end{table}

The various conditions given in Table~\ref{tab:Conditions_for_NE} are \emph{feasibility} conditions. They guarantee that $A$ is able to find a strategy $(x_0,x_1,y_0,y_1)\in [0,1]^4$ that cause at least one of the two types of $B$ agent to be indifferent as to the signal she sends. Case $6$, which is the case relevant to Theorem~\ref{thm:NE_which_is_RR}, can be realized starting with any matrix $M$ satisfying $M_{0,0}M_{1,1} < M_{0,1}M_{1,0}$ (which is a necessary condition derived from Equation~\eqref{eq:strawman_game_assumptions}), which intuitively can be interpreted as having a wrong ``accusation''  being costlier than the gain from a correct ``accusation'' (on average and in absolute terms). Given such $M$, one can set $D_0$ and $D_1$ s.t. $\tfrac {D_0}{D_1} =\sqrt{\tfrac {M_{0,1}}{M_{0,0}}\cdot \tfrac {M_{1,1}} {M_{1,0}}}$ as to satisfy Equation~\eqref{eq:strawman_game_assumptions}. This can be interpreted as balancing the ``significance'' of type $0$ (i.e. $M_{0,0}M_{0,1}$) with the ``significance'' of type $1$ (i.e. $M_{1,0}M_{1,1}$), setting the more significant type as the less probable (i.e. if type $1$ is more significant than type $0$, than $D_1 < D_0$). We then pick $\rho_0,\rho_1$ that satisfy $\tfrac {M_{1,1}}{M_{1,0}}< \tfrac {\rho_1}{\rho_0}< \tfrac {M_{0,1}}{M_{0,0}} $ and scale both by the sufficiently small multiplicative factor so we satisfy the other inequality of case $6$. (In particular, setting $\tfrac {\rho_1}{\rho_0}=\tfrac {D_0}{D_1}$ is a feasible solution.)
Here, $\rho_0$ and $\rho_1$ are set such that the ratio $\tfrac {\rho_1}{\rho_0}$ balances the significance ratio w.r.t type $1$ accusation (i.e. $\tfrac {\rho_1}{\rho_0} > \tfrac {M_{1,1}}{M_{0,0}}$) and the ratio $\tfrac {\rho_0}{\rho_1}$ balances the significance ratio w.r.t to type $0$ accusation (i.e. $\tfrac {\rho_0}{\rho_1} > \tfrac {M_{0,0}}{M_{1,0}}$).
More concretely, for any matrix $M = \left( \begin{array}{c c}
1 & c \cr c & d
\end{array} \right)$ with parameters $c,d$ satisfying $d < c^2$, we can set $\tfrac {D_0}{D_1} = \sqrt{d}$ and any sufficiently small $\rho_0,\rho_1$ satisfying $\tfrac {\rho_1}{\rho_0} \in (\tfrac d c, c)$ and satisfy the requirements of Theorem~\ref{thm:NE_which_is_RR}.

\else
The proof itself is deferred to the full version of this paper, where we also give a complete summary of the various BNEs of this game. We detail $6$ different cases that cover all possible settings of the game. Each of these $6$ cases is defined by a different \emph{feasibility} condition. These conditions guarantee that $A$ is able to find a strategy that cause at least one of the two types of $B$ agent to be indifferent as to the signal she sends. 

The feasibility condition detailed in Equation~\eqref{eq:condition_of_randomized_BNE} can be realized starting with any matrix $M$ satisfying $M_{0,0}M_{1,1} < M_{0,1}M_{1,0}$ (which is a necessary condition derived from Equation~\eqref{eq:strawman_game_assumptions}), which intuitively can be interpreted as having a wrong ``accusation''  being costlier than the gain from a correct ``accusation'' (on average and in absolute terms). Given such $M$, one can set $D_0$ and $D_1$ s.t. $\tfrac {D_0}{D_1} =\sqrt{\tfrac {M_{0,1}}{M_{0,0}}\cdot \tfrac {M_{1,1}} {M_{1,0}}}$ as to satisfy Equation~\eqref{eq:strawman_game_assumptions}. This can be interpreted as balancing the ``significance'' of type $0$ (i.e. $M_{0,0}M_{0,1}$) with the ``significance'' of type $1$ (i.e. $M_{1,0}M_{1,1}$), setting the more significant type as the less probable (i.e. if type $1$ is more significant than type $0$, than $D_1 < D_0$). We then pick $\rho_0,\rho_1$ that satisfy $\tfrac {M_{1,1}}{M_{1,0}}< \tfrac {\rho_1}{\rho_0}< \tfrac {M_{0,1}}{M_{0,0}} $ and scale both by the sufficiently small multiplicative factor so we satisfy the other inequality in Equation~\eqref{eq:condition_of_randomized_BNE}. (In particular, setting $\tfrac {\rho_1}{\rho_0}=\tfrac {D_0}{D_1}$ is a feasible solution.)
Here, $\rho_0$ and $\rho_1$ are set such that the ratio $\tfrac {\rho_1}{\rho_0}$ balances the significance ratio w.r.t type $1$ accusation (i.e. $\tfrac {\rho_1}{\rho_0} > \tfrac {M_{1,1}}{M_{0,0}}$) and the ratio $\tfrac {\rho_0}{\rho_1}$ balances the significance ratio w.r.t to type $0$ accusation (i.e. $\tfrac {\rho_0}{\rho_1} > \tfrac {M_{0,0}}{M_{1,0}}$).
More concretely, for any matrix $M = \left( \begin{array}{c c}
1 & c \cr c & d
\end{array} \right)$ with parameters $c,d$ satisfying $d < c^2$, we can set $\tfrac {D_0}{D_1} = \sqrt{d}$ and any sufficiently small $\rho_0,\rho_1$ satisfying $\tfrac {\rho_1}{\rho_0} \in (\tfrac d c, c)$ and satisfy the requirements of Theorem~\ref{thm:NE_which_is_RR}.

\fi
} \optoutGameDeferredToAPX
\newcommand{\apxOptOut}{
Recall, we assume $\Pr[t=0]=D_0$ and $\Pr[t=1]=D_1$ where wlog $D_0\geq D_1$. As we did before, we denote $B$'s strategy $\sigma_B$ using $p = \Pr[\sigma_B(0)=0]$ and $q = \Pr[\sigma_B(1)=1]$. In contrast to the previous analysis, now $A$ has to decide between $3$ alternatives per $\hat\type$ signal, so $A$ has $6$ options. However, seeing as $A$'s choice to opt-out always give $A$ a utility of $0$, we just denote $4$ alternatives:
\begin{eqnarray*}
x_0 = \Pr[\sigma_A(0)=0] , && x_1=\Pr[\sigma_A(0)=1] \cr
y_0 = \Pr[\sigma_A(1)=0] , && y_1=\Pr[\sigma_A(1)=1] 
\end{eqnarray*}
and we constrain $x_0+x_1 \leq 1$ and $y_0 + y_1 \leq 1$.\footnote{Whereas in the previous section we constrained $x_0+x_1=1$ and $y_0+y_1=1$.} 

Now, given that $A$ views a signal $\hat\type$, she has $3$ alternatives:
\begin{itemize}
\item Accuse $B$ of being of type $0$ and get an expected revenue of 
\begin{align*}
& M_{0,0} \Pr[t=0 ~|~ \hat\type] - M_{0,1} \Pr[\type=1 ~|~ \hat\type]\cr
& = \frac 1 {\Pr[\hat\type]} \left( M_{0,0} \Pr[t=0]\Pr[\sigma_B(0) = \hat\type] - M_{0,1} \Pr[\type=1]\Pr[\sigma_B(1)= \hat\type]  \right)\cr
& = \frac 1 {\Pr[\hat\type]} \left( M_{0,0} D_0 \Pr[\sigma_B(0) = \hat\type] - M_{0,1} D_1 \Pr[\sigma_B(1) = \hat\type]\right)
\end{align*}
\item Accuse $B$ of being of type $1$ and get an expected revenue of
\begin{align*}
&M_{1,1} \Pr[t=1 ~|~ \hat\type] - M_{1,0} \Pr[\type=0 ~|~ \hat\type] \cr
&=\frac 1 {\Pr[\hat\type]} \left( M_{1,1} \Pr[t=1]\Pr[\sigma_B(1) = \hat\type] - M_{1,0} \Pr[\type=0]\Pr[\sigma_B(0)= \hat\type]  \right)\cr
& = \frac 1 {\Pr[\hat\type]} \left( M_{1,1} D_1 \Pr[\sigma_B(1) = \hat\type] - M_{1,0} D_0 \Pr[\sigma_B(0) = \hat\type]\right)
&\end{align*}
\item Opt out and get revenue of $0 = \frac 0 {\Pr[\hat\type]}$.
\end{itemize}
This means that $A$ prefers accusing $B$ of being of type $0$ to opting out when
\[ \Pr[\sigma_B(0) = \hat\type] > \frac{M_{0,1} D_1}{M_{0,0} D_0}\Pr[\sigma_B(1) = \hat\type] 
\] 
Similarly, $A$ prefers accusing $B$ of being of type $1$ to opting out when
\[ \Pr[\sigma_B(0) = \hat\type] < \frac{M_{1,1} D_1}{M_{1,0} D_0}\Pr[\sigma_B(1) = \hat\type] 
\]
From Equation~\eqref{eq:strawman_game_assumptions} we have that $\tfrac {M_{1,1}D_1}{M_{1,0}D_0} < 1 < \tfrac {M_{0,1}D_1}{M_{0,0}D_0}$. Therefore, \emph{given that $\Pr[\hat\type]>0$}, then $A$'s best response is determined by the ratio:
\begin{equation*}
\frac {\Pr[\sigma_B(0)=\hat\type]}{\Pr[\sigma_B(1)=\hat\type]} 
\begin{cases}
< \frac {M_{1,1}D_1}{M_{1,0}D_0}, & A \textrm{ plays } \Pr[\sigma_A(\hat\type) = 1] = 1 \cr
 = \frac {M_{1,1}D_1}{M_{1,0}D_0}, & A \textrm{ is indifferent between }\bot\textrm{ and playing } \tilde\type=1\cr
 \in (\frac {M_{1,1}D_1}{M_{1,0}D_0}, \frac {M_{0,1}D_1}{M_{0,0}D_0}) & A \textrm{ plays } \Pr[\sigma_A(\hat\type)=\bot]=1\cr
 = \frac {M_{0,1}D_1}{M_{0,0}D_0}, & A \textrm{ is indifferent between }\bot\textrm{ and playing } \tilde\type=0\cr
 > \frac {M_{0,1}D_1}{M_{0,0}D_0} & A \textrm{ plays } \Pr[\sigma_A(\hat\type)=0]=1
\end{cases}
\end{equation*}
Therefore $A$'s BNE strategy when viewing the signal $\hat\type$ (which is best response to $B$'s BNE strategy) is such that $A$ never plays both $\tilde\type = \hat\type$ and $\tilde\type=1-\hat\type$ with non-zero probability.

Switching to the $B$ agent, the utility functions of $B$ are similar to before:
\begin{align*}
\textrm{For type } \type=0:~~~ & U_{B,0} = p(\rho_0 - x_0 M_{0,0} + x_1 M_{1,0}) + (1-p) (-y_0 M_{0,0} + y_1 M_{1,0}) \cr
\textrm{For type } \type=1:~~~ & U_{B,1} = q(\rho_1 - y_1 M_{1,1} + y_0 M_{0,1}) + (1-q) (-x_1 M_{1,1} + x_0 M_{0,1}) 
\end{align*}
and so $p =1$ if $\rho_0 > M_{0,0}(x_0-y_0) - M_{1,0}(x_1-y_1)$ and $p=0$ if $\rho_0 < M_{0,0}(x_0-y_0) - M_{1,0}(x_1-y_1)$; similarly $q=1$ if $\rho_1 >  M_{1,1}(y_1-x_1)-M_{0,1}(y_0-x_0)$ and $q=0$ if $\rho_1 <  M_{1,1}(y_1-x_1)-M_{0,1}(y_0-x_0)$.

We can now make our first claim about the BNE of the game.
\begin{claim}
In any BNE strategy of $B$ we have that either 
\[ \frac p {1-q} = \frac {\Pr[\sigma_B^*(0)=0]}{\Pr[\sigma^*_B(1)=0]} \geq \frac {M_{0,1}D_1}{M_{0,0}D_0} ~~\textrm{ or }~~ \frac {1-p} {q} = \frac {\Pr[\sigma_B^*(0)=1]}{\Pr[\sigma^*_B(1)=1]} \leq \frac {M_{1,1}D_1}{M_{1,0}D_0}\] 
\end{claim}
\begin{proof}
Assume for the same of contradiction that both conditions do not hold. Then given the $\hat\type=0$ signal it holds that $x_0 = \Pr[\sigma^*_A(0)=0] = 0$, and given the $\hat\type=1$ signal it holds that $y_1 = \Pr[\sigma^*_A(1)=1]=0$. Thus, $B$'s best response to $A$'s strategy is to switch to $(p,q)=(1,1)$ (since $\rho_0,\rho_1>0$) and now both conditions do hold.
\end{proof}
\begin{claim}
In any BNE strategy of $B$ we have that both
\[ \frac p {1-q} = \frac {\Pr[\sigma_B^*(0)=0]}{\Pr[\sigma^*_B(1)=0]} > \frac {M_{1,1}D_1}{M_{1,0}D_0} ~~\textrm{ and }~~ \frac {1-p} {q} = \frac {\Pr[\sigma_B^*(0)=1]}{\Pr[\sigma^*_B(1)=1]} < \frac {M_{0,1}D_1}{M_{0,0}D_0}\] 
\end{claim}
\begin{proof}
Based on the previous claim, one of the two inequalities is immediate. Assume we have $\frac {p}{1-q} \geq \frac {M_{0,1}D_1}{M_{0,0}D_0} > 1 > \frac {M_{1,1}D_1}{M_{1,0}D_0}$, we now show that $\frac {1-p}q < \frac {M_{0,1}D_1}{M_{0,0}D_0}$ must also hold. If, for contradiction, we have that   $\frac {1-p}q \geq \frac {M_{0,1}D_1}{M_{0,0}D_0}$ then
\[ 1 = p + (1-p) \geq \frac {M_{0,1}D_1}{M_{0,0}D_0} \big(q + (1-q)\big) = \frac {M_{0,1}D_1}{M_{0,0}D_0}\] which contradicts Equation~\eqref{eq:strawman_game_assumptions}. The argument for the case $\frac {1-p}q \leq \frac {M_{1,1}D_1}{M_{1,0}D_0}$ is symmetric.
\end{proof}
Based on the last claim and on $A$'s best-response analysis, we have that in any BNE strategy of $A$ it holds that $x_1 = \Pr[\sigma^*_A(0)=1] = 0$ and $y_0 = \Pr[\sigma^*_A(1)=0]=0$. (I.e., given the signal $\hat\type$ then $A$ never plays $\tilde\type = 1-\hat\type$.) As a result, $B$'s best response analysis simplifies to: $p =1$ if $\rho_0 > M_{0,0}x_0 + M_{1,0}y_1$ and $p=0$ if $\rho_0 < M_{0,0}x_0 + M_{1,0}y_1$; similarly $q=1$ if $\rho_1 >  M_{1,1}y_1+M_{0,1}x_0$ and $q=0$ if $\rho_1 <  M_{1,1}y_1+M_{0,1}x_0$.

We are now able to prove the existence of a BNE as specified Theorem~\ref{thm:NE_which_is_RR}.
\begin{claim}
Assume that $0 \leq \rho_1M_{1,0} - \rho_0M_{1,1} \leq M_{0,1}M_{1,0}-M_{0,0}M_{1,1}$ and $0 \leq \rho_0M_{0,1} - \rho_1M_{0,0} \leq   M_{0,1}M_{1,0}-M_{0,0}M_{1,1}$. The strategies $\sigma_A^*$ and $\sigma_B^*$ denoted below are BNE strategies.
\begin{align*}
\textrm{For } A: & x_0^* &= \Pr[\sigma^*_A(0)=0] &= \frac {M_{1,0}\rho_1-M_{1,1}\rho_0} {M_{1,0}M_{0,1} - M_{0,0}M_{1,1}} \cr
& x_1^* &= \Pr[\sigma^*_A(0)=1] &= 0 \cr
& y_0^* &= \Pr[\sigma^*_A(1) = 0] &=0 \cr
& y_1^* &= \Pr[\sigma^*_A(1)= 1] &= \frac {M_{0,1} \rho_0 - M_{0,0}\rho_1} {M_{1,0}M_{0,1} - M_{0,0}M_{1,1}} \cr
\textrm{For } B: & p^* &= \Pr[\sigma^*_B(0) = 0] &= \frac {D_1M_{0,1}(D_0M_{1,0} - D_1M_{1,1})} {D_0D_1 M_{0,1}M_{1,0} - D_0 D_1 M_{0,0}M_{1,1}}\cr
&1-p^* &= \Pr[\sigma^*_B(0)=1] &= \frac {D_1M_{1,1}(D_1M_{0,1}- D_0 M_{0,0})} {D_0D_1 M_{0,1}M_{1,0} - D_0 D_1 M_{0,0}M_{1,1}}\cr
&1-q^* &= \Pr[\sigma^*_B(1)=0] &= \frac {D_0M_{0,0}(D_0M_{1,0}- D_1 M_{1,1})} {D_0D_1 M_{0,1}M_{1,0} - D_0 D_1 M_{0,0}M_{1,1}}\cr
&q^* &= \Pr[\sigma^*_B(1)=1] &= \frac {D_0M_{1,0}(D_1M_{0,1} - D_0 M_{0,0})} {D_0D_1 M_{0,1}M_{1,0} - D_0 D_1 M_{0,0}M_{1,1}} 
\end{align*}
\end{claim}
\begin{proof}
First, observe that under the given assumptions in the claim it holds that $x_0^*,y_1^* \in [0,1]$,  and due to Equation~\eqref{eq:strawman_game_assumptions} it holds that $p^*,q^*,1-p^*,1-q^*$ are all strictly positive (so $p^*,q^* \in (0,1)$). 

Now observe that when $B$ follows $\sigma^*_B$ then $A$ has no incentive to deviate since $\frac {p^*} {1-q^*} = \frac {M_{0,1}D_1}{M_{0,0}D_0}$ and $\frac {1-p^*}{q^*}=\frac {M_{1,1}D_1}{M_{1,0}D_0}$. When $A$ follows $\sigma^*_A$ then $B$ has no incentive to deviate since
\begin{align*}
M_{0,0} x_0^* + M_{1,0}y_1^* & = \frac{M_{1,0}M_{0,1}\rho_0 - M_{0,0}M_{1,1}\rho_0 } {M_{1,0}M_{0,1}-M_{0,0}M_{1,1}} &=\rho_0 \cr
M_{0,1} x_0^* + M_{1,1}y_1^* & = \frac{M_{0,1}M_{1,0}\rho_1 - M_{1,1}M_{0,0}\rho_1 } {M_{1,0}M_{0,1}-M_{0,0}M_{1,1}} &=\rho_1
\end{align*}
\end{proof}
Observe that when $D_0^2M_{0,0}M_{1,0} = D_1^2M_{0,1}M_{1,1}$ then $p^*=q^*$. Furthermore, in this case we have that $p^*> \tfrac 1 2$ because 
\begin{align*}
&2D_0D_1M_{0,1}M_{1,0} - 2D_1^2M_{0,1}M_{1,1} > D_0D_1M_{0,1}M_{1,0}-D_0D_1M_{0,0}M_{1,1}\cr
& \Leftrightarrow
D_0D_1M_{0,1}M_{1,0} - D_1^2M_{0,1}M_{1,1} > D_0^2M_{0,0}M_{1,0} -D_0D_1M_{0,0}M_{1,1} \cr
&\Leftrightarrow D_1M_{0,1}(D_0M_{1,0}-D_1M_{1,1}) > D_0M_{0,0}(D_0M_{1,0}-D_1M_{1,1})\cr
&\Leftrightarrow 0 > -D_1M_{0,1} + D_0M_{0,0}
\end{align*}
where the last derivation and the last inequality are both true because of Equation~\eqref{eq:strawman_game_assumptions}. This concludes the existence part of Theorem~\ref{thm:NE_which_is_RR}. The more complicated part is to show that $B$'s BNE strategy is \emph{unique}. We make the following argument.
\begin{theorem}
\label{thm:mixed_BNE_under_condition}
Assume that $0 < \rho_1M_{1,0} - \rho_0M_{1,1} < M_{0,1}M_{1,0}-M_{0,0}M_{1,1}$ and $0 < \rho_0M_{0,1} - \rho_1M_{0,0} <  M_{0,1}M_{1,0}-M_{0,0}M_{1,1}$. Then in all BNEs of the game both types of $B$ agent play a mixed strategy.
\end{theorem}

Assuming Theorem~\ref{thm:mixed_BNE_under_condition} holds and using our above best-responsse analysis the uniqueness of the BNE of Theorem~\ref{thm:NE_which_is_RR} is immediate. If both $p$ and $q$ are non-integral then it must hold that $x_0^*$ and $y_1^*$, since this is the unique solution to the system of two linear equations in two variables that set both types of $B$ agent indifferent. Under the assumption of Theorem~\ref{thm:mixed_BNE_under_condition} we have that both $x_0^*$ and $y_1^*$ are non-integral as well. This means that $B$ plays $(p,q)$ s.t. $A$ is indifferent to the value of $x_0^*,y_1^*$; i.e. $p = (1-q) \frac{M_{0,1}D_1}{M_{0,0}D_0}$ and $1-p = q \frac {M_{1,1}D_1}{M_{1,0}D_0}$. Again, since this is a linear system in two variables, there exists a unique $(p,q)$ pair that satisfy this condition, which is given by $(p^*,q^*)$. 

In the rest of this section (following t, our goal is to prove the Theorem~\ref{thm:mixed_BNE_under_condition}. In fact, we give a full analysis of all the points $(p,q)$ that \emph{may} be $B$'s BNE strategy, and for each such possible $(p,q)$ we analyze the conditions over the parameters of the game underwhich it is a BNE strategy for $B$. The analysis is fairly long and tedious, as it involves checking feasibility constraints over the $6$ parameters of the game: $\rho_0$, $\rho_1$, $M_{0,0}$, $M_{0,1}$, $M_{1,0}$ and $M_{1,1}$. Furthermore, after deriving the suitable feasibility constraints, we show that they cover all settings of the parameters of the game and are mutually exclusive (when inequalities are strict).

\cut{
\subsection{Utilities and Best-Response Analysis}
\label{subsec:NE_analysis}

Recall, we assume $\Pr[t=0]=D_0$ and $\Pr[t=1]=D_1$ where wlog $D_0\geq D_1$. As we did before, we denote $B$'s strategy $\sigma_B$ using $p = \Pr[\sigma_B(0)=0]$ and $q = \Pr[\sigma_B(1)=1]$. In contrast to the previous analysis, now $A$ has to decide between $3$ alternatives per $\hat\type$ signal, so $A$ has $6$ options. However, seeing as $A$'s choice to opt-out always give $A$ a utility of $0$, we just denote $4$ alternatives:
\begin{eqnarray*}
x_0 = \Pr[\sigma_A(0)=0] , && x_1=\Pr[\sigma_A(0)=1] \cr
y_0 = \Pr[\sigma_A(1)=0] , && y_1=\Pr[\sigma_A(1)=1] 
\end{eqnarray*}
and we constrain $x_0+x_1 \leq 1$ and $y_0 + y_1 \leq 1$.\footnote{Whereas in the previous section we constrained $x_0+x_1=1$ and $y_0+y_1=1$.} 

The utility functions of $B$ are much like they were before:
\begin{align*}
\textrm{For type } \type=0:~~~ & U_{B,0} = p(\rho_0 - x_0 M_{0,0} + x_1 M_{1,0}) + (1-p) (-y_0 M_{0,0} + y_1 M_{1,0}) \cr
\textrm{For type } \type=1:~~~ & U_{B,1} = q(\rho_1 - y_1 M_{1,1} + y_0 M_{0,1}) + (1-q) (-x_1 M_{1,1} + x_0 M_{0,1}) 
\end{align*}
The utility function of $A$ too is very similar to before:
\begin{align*}
u_A &  = D_0 \big( p(x_0M_{0,0} - x_1 M_{1,0}) + (1-p)(y_0 M_{0,0} - y_1 M_{1,0})  \big)\cr
& ~~~+ D_1\big(  q(y_1M_{1,1}-y_0M_{0,1}) + (1-q)( x_1M_{1,1}-x_0M_{0,1} )  \big) \cr
& = x_0 \left( D_0 p M_{0,0} - D_1(1-q)M_{0,1}\right) + x_1 \left( -D_0 p M_{1,0} + D_1(1-q)M_{1,1}\right) \cr
& ~~~+ y_0 \left( D_0 (1-p) M_{0,0} - D_1qM_{0,1}\right) + y_1 \left(- D_0(1- p) M_{1,0} + D_1q M_{1,1}\right)
\end{align*}

In order to determine what it $A$'s best response to a strategy $(p,q)$ of $B$, we compare the utility from $x_0$, $x_1$ and opting out, and similarly for $y_0$, $y_1$ and opting out.
\begin{itemize}
\item When $\E[ u_A \textrm{ from setting } \sigma_A(0)=0 ~|~ \hat\type=0] > \E[ u_A \textrm{ from setting } \sigma_A(0)=1 ~|~ \hat\type=0]$ (i.e. when $D_0pM_{0,0} -D_1(1-q)M_{0,1} > D_1(1-q)M_{1,1}-D_0pM_{1,0}$), then $A$ prefers playing $\sigma_A(0)=0$ to $\sigma_A(0)=1$, (i.e., $x_1=0$).\\
When the inequality holds in the opposite direction, i.e. $D_0pM_{0,0} -D_1(1-q)M_{0,1} < D_1(1-q)M_{1,1}-D_0pM_{1,0}$, then $x_0=0$.
\item When $\E[ u_A \textrm{ from setting } \sigma_A(0)=0 ~|~ \hat\type=0] < 0$ (i.e. when $D_0pM_{0,0} -D_1(1-q)M_{0,1} < 0$), then $A$ prefers opting-out to playing $\sigma_A(0)=0$ (i.e. $x_0=0$).
\item When $\E[ u_A \textrm{ from setting } \sigma_A(0)=1 ~|~ \hat\type=0] < 0$ (i.e. when $-D_0pM_{1,0} +D_1(1-q)M_{1,1} < 0$), then $A$  prefers opting-out to playing $\sigma_A(0)=1$ (i.e. $x_1=0$).
\end{itemize}
\begin{itemize}
\item When $\E[ u_A \textrm{ from setting } \sigma_A(1)=1 ~|~ \hat\type=1] > \E[ u_A \textrm{ from setting } \sigma_A(1)=0 ~|~ \hat\type=1]$ (i.e. when $D_0 (1-p) M_{0,0} - D_1qM_{0,1} < - D_0(1- p) M_{1,0} + D_1q M_{1,1}$), then $A$ prefers playing $\sigma_A(1)=1$ to $\sigma_A(1)=0$, (i.e., $y_0=0$).\\
When the inequality holds in the opposite direction, i.e. $D_0 (1-p) M_{0,0} - D_1qM_{0,1} > - D_0(1- p) M_{1,0} + D_1q M_{1,1}$, then $y_1=0$.
\item When $\E[ u_A \textrm{ from setting } \sigma_A(1)=1 ~|~ \hat\type=1] < 0$ (i.e. when $- D_0(1- p) M_{1,0} + D_1q M_{1,1} < 0$), then $A$ prefers opting-out to playing $\sigma_A(1)=1$ (i.e. $y_1=0$).
\item When $\E[ u_A \textrm{ from setting } \sigma_A(1)=0 ~|~ \hat\type=1] < 0$ (i.e. when $D_0 (1-p) M_{0,0} - D_1qM_{0,1} < 0$), then $A$  prefers opting-out to playing $\sigma_A(0)=1$ (i.e. $y_0=0$).
\end{itemize}

\begin{align*}
& \textrm{When } D_0pM_{0,0} -D_1(1-q)M_{0,1} > D_1(1-q)M_{1,1}-D_0pM_{1,0} & \textrm{ then } x_1 = 0, \cr
& \textrm{ and when } D_0pM_{0,0} -D_1(1-q)M_{0,1} < D_1(1-q)M_{1,1}-D_0pM_{1,0} & \textrm{ then } x_0 = 0; \cr
& \textrm{When } D_0pM_{0,0} -D_1(1-q)M_{0,1} < 0 & \textrm{ then } x_0 = 0; \cr
& \textrm{When } -D_0pM_{1,0} +D_1(1-q)M_{1,1} < 0 & \textrm{ then } x_1 = 0. \cr
&&\cr
& \textrm{When } D_0 (1-p) M_{0,0} - D_1qM_{0,1} > - D_0(1- p) M_{1,0} + D_1q M_{1,1} & \textrm{ then } y_1 = 0, \cr
& \textrm{ and when } D_0 (1-p) M_{0,0} - D_1qM_{0,1} < - D_0(1- p) M_{1,0} + D_1q M_{1,1} & \textrm{ then } y_0 = 0; \cr
& \textrm{When } D_0 (1-p) M_{0,0} - D_1qM_{0,1} < 0 & \textrm{ then } y_0 = 0; \cr
& \textrm{When } - D_0(1- p) M_{1,0} + D_1q M_{1,1} < 0 & \textrm{ then } y_1 = 0. \cr
\end{align*}
We turn to examining the ``lines of indifference'' -- the set of strategies $(p,q)$ that cause $A$ to be indifferent between at least two strategies.
\begin{align*}
l_1 :~~~ & D_0pM_{0,0} -D_1(1-q)M_{0,1}=0 \cr
l_2 :~~~ & D_0p(M_{0,0}+M_{1,0}) -D_1(1-q)(M_{0,1}+M_{1,1}) = 0 \cr
l_3 :~~~ & -D_0pM_{1,0} +D_1(1-q)M_{1,1} = 0\cr
\cr
l_4 :~~~ & - D_0(1- p) M_{1,0} + D_1q M_{1,1} = 0\cr
l_5 :~~~ & D_0(1- p) (M_{0,0}+M_{1,0}) - D_1q (M_{0,1}+M_{1,1}) = 0  \cr
l_6 :~~~ & D_0 (1-p) M_{0,0} - D_1qM_{0,1} = 0 \cr
\end{align*}
Clearly, the point $(0,1)$ lie on all $3$ lines $l_1,l_2,l_3$. Some arithmetic show that the following points reside on each of these lines: $(\tfrac {D_1}{D_0}\tfrac{M_{0,1}}{M_{0,0}} ,0) \in l_1$, $(\tfrac {D_1}{D_0}\tfrac{M_{0,1}+M_{1,1}}{M_{0,0}+M_{1,0}},0)\in l_2$, $(\tfrac {D_1}{D_0}\tfrac{M_{1,1}}{M_{1,0}},0)\in l_3$. Using the assumptions given in~\eqref{eq:strawman_game_assumptions}, we have that $\tfrac {D_1}{D_0}\tfrac{M_{0,1}}{M_{0,0}} > 1 > \tfrac {D_1}{D_0}\tfrac{M_{1,1}}{M_{1,0}}$, so the line $l_1$ lies above the line $l_3$. Next we use the fact that \[\max\{\tfrac{M_{0,1}}{M_{0,0}},\tfrac{M_{1,1}}{M_{1,0}} \} \geq \tfrac{M_{0,1}+M_{1,1}}{M_{0,0}+M_{1,0}} \geq \min\{\tfrac{M_{0,1}}{M_{0,0}},\tfrac{M_{1,1}}{M_{1,0}} \}\footnote{Proof: for any positive $a,b,c,d$ we have $a+b = c\tfrac a c + d\tfrac b d \leq \max\{\tfrac a c,\tfrac b d \}(c+d)$.} \] to deduce that the line $l_2$ is ``sandwiched'' in between $l_1$ and $l_3$. A schematic description of the three lines is given in Figure~\ref{fig:indifference_lines}.
\begin{figure}[t]
\centering
\includegraphics[scale=0.65]{indifference_lines.jpg}
\caption{\label{fig:indifference_lines}\small The lines on which $A$ is indifferent between at least two options. Lines $l_1,l_2,l_3$ (in blue) affect $A$'s behavior when seeing the $\hat\type=0$ signal, and lines $l_4,l_5,l_6$ affect $A$'s behavior when seeing the $\hat\type=1$ signal.}
\end{figure}

Following $A$'s best-response strategies we deduce that whenever $B$ plays a strategy $(p,q)$ such that:
\begin{itemize}
\item $(p,q)$ lies below the $l_3$ line, then $A$ plays $x_1=1$, $x_0=0$. I.e., in response to the signal $\hat\type=0$, $A$ prefers to accuse $B$ of being of type $\tilde\type=1$.
\item $(p,q)$ lies between the $l_1$ and $l_3$ lines, then $A$ plays $x_0=x_1=0$. I.e., in response to the signal $\hat\type=0$, $A$ prefers to opt out.
\item $(p,q)$ lies above the $l_1$ line, then $A$ plays $x_0=1$ and $x_1=0$. I.e., in response to the signal $\hat\type=0$, $A$ prefer to accuse $B$ of being of types $\tilde\type=0$.
\end{itemize} 

Lines $l_4,l_5,l_6$ are the symmetric analogue of lines $l_1,l_2,l_3$. Following the same line of reasoning, $l_4$ lies above the other two lines. As best response to any $(p,q)$ above line $l_4$ we have that $A$ sets $y_1=1$; and for any $(p,q)$ below line $l_4$, $A$ sets $y_1=0$. 
For the sake of completeness we give additional details below, but the reader can skip to Section~\ref{subsec:NE_locations}.

Analogously, the point $(1,0)$ lies at the intersection of lines $l_4, l_5, l_6$. And again, the following points lie on the following lines: $(0, \tfrac {D_0}{D_1}\tfrac{M_{1,0}}{M_{1,1}} )\in l_4$, $(0,\tfrac {D_0}{D_1}\tfrac{M_{0,0}+M_{1,0}}{M_{0,1}+M_{1,1}})\in l_5$, $(0, \tfrac {D_0}{D_1}\tfrac{M_{0,0}}{M_{0,1}} ) \in l_6$. We use the assumptions given in~\eqref{eq:strawman_game_assumptions} to deduce that  $\tfrac {D_0}{D_1}\tfrac{M_{1,0}}{M_{1,1}}>1>\tfrac {D_0}{D_1}\tfrac{M_{0,0}}{M_{0,1}}$, and the fact that
\[ \max\{\tfrac{M_{0,0}}{M_{0,1}},\tfrac{M_{1,0}}{M_{1,1}} \} \geq \tfrac{M_{0,0}+M_{1,0}}{M_{0,1}+M_{1,1}} \geq \min\{\tfrac{M_{0,0}}{M_{0,1}},\tfrac{M_{1,0}}{M_{1,1}}\} \] to deduce that $l_4$ is above $l_6$ and $l_5$ is ``sandwiched'' between them. This too is shown in Figure~\ref{fig:indifference_lines}. Similar best-response strategies hold for a strategy $(p,q)$ of $B$ such that
\begin{itemize}
\item $(p,q)$ lies below the $l_6$ line, then $A$ plays $y_0=1$, $y_1=0$. I.e., in response to the signal $\hat\type=1$, $A$ prefers to accuse $B$ of being of type $\tilde\type=0$.
\item $(p,q)$ lies between the $l_4$ and $l_6$ lines, then $A$ plays $y_0=y_1=0$. I.e., in response to the signal $\hat\type=1$, $A$ prefers to opt out.
\item $(p,q)$ lies above the $l_4$ line, then $A$ plays $y_1=1$ and $y_0=0$. I.e., in response to the signal $\hat\type=1$, $A$ prefer to accuse $B$ of being of types $\tilde\type=1$.
\end{itemize} 
}

\subsection{Proof of Theorem~\ref{thm:NE_which_is_RR}: Characterizing All Potential BNEs of the Game}
\label{subsec:NE_locations}

Consider the space $[0,1]\times[0,1]$ of all possible strategies $(p,q)$ for the two types of  $B$ agents. We denote the two ``lines of indifference'' for $A$ on this square:
\begin{eqnarray*}
l_1 : &p = \frac {M_{0,1}D_1}{M_{0,0}D_0} (1-q) \cr
l_2 : &1-p = \frac {M_{1,1}D_1}{M_{0,1}D_1} q
\end{eqnarray*}
where $(p,q)=(0,1) \in l_1$ and $(p,q)=(1,0) \in l_2$. These lines partition the $[0,1]\times[0,1]$ square into multiple different regions, as shown in Figure~\ref{fig:non_NE_areas}).

\begin{figure}[t]
\centering
\includegraphics[scale=0.6]{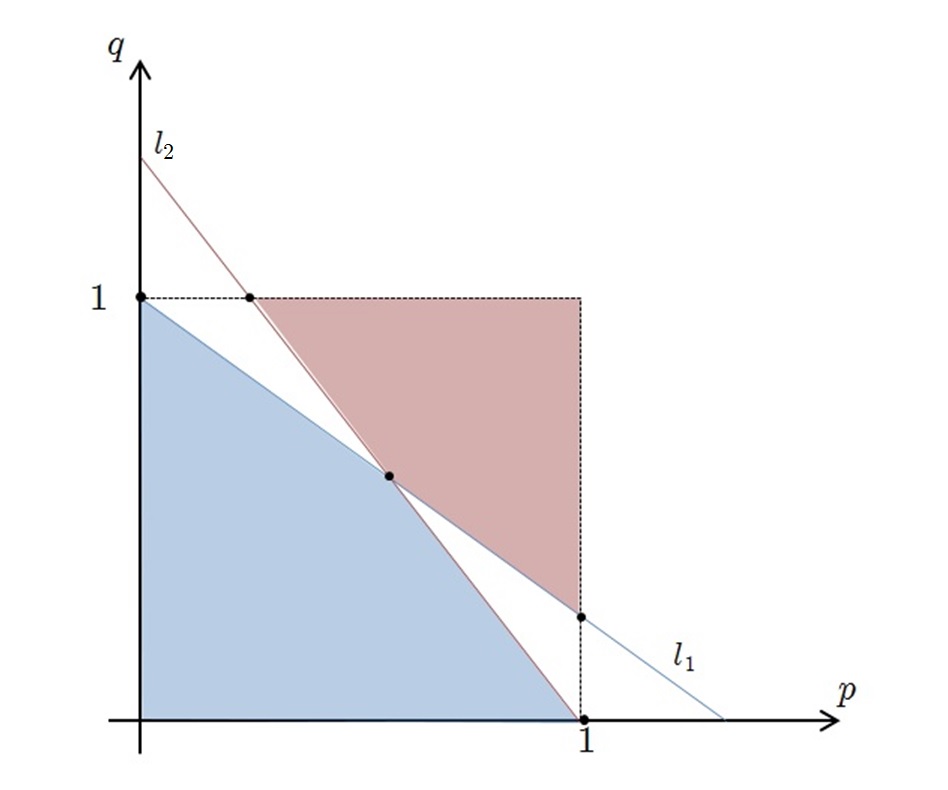}
\caption{\small The 4 regions created by the lines $l_1$ and $l_2$, and the $5$ intersection points of the two lines and the borders of the squares.\label{fig:non_NE_areas}}
\end{figure}

First, we argue that any $(p,q)$ in the lower region, below $l_1$ and $l_2$ (the shaded blue region in Figure~\ref{fig:non_NE_areas}) cannot be $B$'s strategy in a BNE. We in fact already shown this: for any $(p,q)$ in this blue region we have that $\frac p {1-q} < \frac {M_{0,1}D_1}{M_{0,0}D_0}$ and $\frac {1-p} q < \frac {M_{1,1}D_1}{M_{0,1}D_1}$, contradicting our earlier claim.


Secondly, we argue that a point $(p,q)$ above both lines is $B$'s strategy in a BNE is if the valuation of $B$ for the coupons is high. Observe, for any $(p,q)$ above the line, $A$'s best response is to set $x_0=y_1=1$, which means $B$'s utility is $p(\rho_0-M_{0,0})+(1-p)M_{1,0}$ for agents of type $\type=0$, and $q(\rho_1-M_{1,1})+(1-q)M_{0,1}$ for type $\type=1$. Therefore, $B$ has no incentive to deviate from $(p,q)$ only if $\rho_0 \geq M_{0,0}+M_{1,0}$ and $\rho_1 \geq M_{0,1}+M_{1,1}$. In particular, when both inequalities are strict, we have that $(1,1)$ is $B$'s BNE; when both are equalities, any point above both $l_1$ and $l_2$ is a BNE; and when one is an equality and the other is a strict inequality, we have that the BNE strategy is on the border of the $[0,1]\times[0,1]$ square.

We now turn our attention to the points strictly between the lines $l_1$ and $l_2$ (excluding all points on these lines). To any $(p,q)$ in these regions, $A$'s best response is to play $\tilde\type=\hat\type$ when seeing one signal and to opt-out when seeing the other signal. E.g., for any $(p,q)$ above $l_1$ but below $l_2$ (the top-left region), $A$ opts out when seeing the $\hat\type=1$ signal, but play $\tilde\type=0$ when seeing the $\hat\type=0$ signal. In that case, $B$'s utility function is $p(\rho_0 - M_{0,0})$ for type $\type=0$, and $q\rho_1 + (1-q)M_{0,1}$for type $\type=1$. Therefore, unless $\rho_0=M_{0,0}$, agents of type $\type=0$ have incentive to deviate (either to playing $p=0$ or $p=1$). In addition, it must also hold that $\rho_1\geq M_{0,1}$. (If this inequality is strict, then the BNE strategy lies on the border of the square.) Analogously, should the BNE strategy lie above the $l_2$ line but below the $l_1$ line (lower-right area), then it must be the case that $\rho_1=M_{1,1}$ and $\rho_0 \geq M_{1,0}$.

We now consider points on the $l_1$ and $l_2$, excluding the $5$ intersection points we have with either the two lines, or any of the lines intersection with the $p=1$ line or the $q=1$ line.
\begin{itemize}
\item For any $(p,q)$ on the $l_1$ line below the $l_2$ line (top-left side of $l_1$ bordering the blue region): $A$'s best response to any such $(p,q)$ is to set $y_0=y_1=0$ and $x_1=0$. It follows that $B$'s utility is $p(\rho_0 - x_0M_{0,0})$ and $q\rho_1 + (1-q) x_0 M_{0,1}$ for types $\type=0$ and $\type=1$ resp. Hence, if $0 \leq \tfrac {\rho_0}{M_{0,0}} = \tfrac {\rho_1}{M_{0,1}} \leq 1$ then such strategies can be BNE
\item For any $(p,q)$ on the $l_1$ line above the $l_2$ line (bottom-right side of $l_1$ bordering the red region): $A$'s best response to such $(p,q)$ is to set $x_1=y_0=0$, $y_1=1$; $B$'s utility functions are $p(\rho_0-x_0M_{0,0}) + (1-p) M_{1,0}$ and $q(\rho_1 - M_{1,1})+ (1-q)x_0M_{0,1}$ for types $\type=0$ and $\type=1$ resp. It follows that if $0 \leq \tfrac{ \rho_0-M_{1,0} } {M_{0,0}} = \tfrac {\rho_1-M_{1,1}}{M_{0,1}}\leq 1$, then such point give a BNE
\item For any $(p,q)$ on the $l_2$ line below the $l_1$ line (bottom-right side of $l_2$ bordering the blue region): As a response to any $(p,q)$ here, $A$ sets $x_0=x_1=0$ and $y_0=0$. $B$'s utility function is therefore $p\rho_0+(1-p) (y_1M_{1,0})$ and $q(\rho_1-y_1M_{1,1})$. Hence, in case $0\leq \tfrac {\rho_0}{M_{1,0}} = \tfrac {\rho_1}{M_{1,1}} \leq 1$, we have a BNE with such $(p,q)$. 
\item For any $(p,q)$ on the $l_2$ line above the $l_1$ line (top-left side of $l_2$ bordering the red region): As best response to such $(p,q)$, $A$ sets $x_1=y_0=1$ and $x_0=1$. And so $B$'s utility functions are $p(\rho_0 - M_{0,0})+(1-p)y_1M_{1,0}$ and $q(\rho_1 - y_1 M_{1,1} +(1-q)M_{0,1}$. Therefore, if we have $0\leq \tfrac {\rho_0 - M_{0,0}}{M_{1,0}} = \tfrac {\rho_1-M_{0,1}}{M_{1,1}} \leq 1$ then we have a BNE with such $(p,q)$.
\end{itemize}
Thus far (with the exception of the potential BNE at the point $(p,q)=(1,1)$) we have considered only BNE that may arise only when the parameters of the game ($\rho_0,\rho_1$ and the entries of $M$) satisfy some equality constraints. Assuming we perturb the values of $\rho_0$ and $\rho_1$ a little such that none of the above-mentioned equalities hold, we are left with $5$ points on which the BNE can occur:
\begin{align*}
& P_1= (0,1) ~,~~~ P_2 = (1 - \frac {D_1M_{1,1}}{D_0 M_{1,0}} , 1) ~,~~~ P_3 = (1,0) ~,~~~ P_4 = (1, 1-\frac {D_0 M_{0,0}}{D_1M_{0,1}}) \cr
& P_5 = \left( \frac {D_0D_1M_{0,1}M_{1,0} - D_1^2 M_{0,1}M_{1,1}} {D_0D_1 M_{0,1}M_{1,0} - D_0 D_1 M_{0,0}M_{1,1}} , \frac {D_0D_1M_{0,1}M_{1,0} - D_0^2 M_{0,0}M_{1,0}} {D_0D_1 M_{0,1}M_{1,0} - D_0 D_1 M_{0,0}M_{1,1}} \right) 
\end{align*}
We traverse them one by one. We remind the reader that since $\tfrac {M_{0,1}}{M_{0,0}} > \tfrac {D_0}{D_1} > \tfrac {M_{1,1}}{M_{1,0}}$ then $M_{0,0}M_{1,1} < M_{0,1}M_{1,0}$. We repeatedly use this inequality in the analysis below.

\ifx \fullversion\undefined
\paragraph{Conditions under which $B$'s BNE strategy is $P_1$.} 
\else
\subsubsection{Conditions under which $B$'s BNE strategy is $P_1$.}
\fi
Observe that in this case $B$ never sends the $\hat\type=0$ signal. $A$'s best response in naturally to opt-out, but $A$ still commits to a certain values of $x_0$ and $x_1$ (to prevent $B$ from deviating from the $(0,1)$ strategy). $B$'s utility functions are
\begin{eqnarray*}
\textrm{For } \type=0 : && p(\rho_0 - x_0M_{0,0}+x_1 M_{1,0}) \cr
\textrm{For } \type=1 : && q\rho_1+(1-q)(x_0M_{0,1}-x_1 M_{1,1}) 
\end{eqnarray*}
So $A$ should set $x_0$ and $x_1$ s.t. $\rho_0 \leq x_0M_{0,0}-x_1M_{1,0}$ and $\rho_1 \geq x_0 M_{0,1} - x_1 M_{1,1}$.

\begin{proposition}
\label{clm:conditions_for_P1}
There exist $x_0, x_1\in[0,1]$ satisfying both $\rho_0 \leq x_0M_{0,0}-x_1M_{1,0}$ and $\rho_1 \geq x_0 M_{0,1} - x_1 M_{1,1}$ iff $\rho_0 \leq M_{0,0}$ and $\rho_0 M_{0,1} \leq \rho_1 M_{0,0}$. 
\end{proposition}
\begin{proof}
To see that these conditions are sufficient, assume that $\rho_0 \leq M_{0,0}$ and $\rho_0 M_{0,1} \leq \rho_1 M_{0,0}$. Then we can set $x_0 = \tfrac {\rho_0}{M_{0,0}}$ and $x_1=0$. Clearly, both lie on the $[0,1]$-interval. We can check and see that indeed $\rho_0 \leq \tfrac {\rho_0} {M_{0,0}} M_{0,0} - 0 \cdot M_{1,0}$ and $\rho_1 \geq \tfrac {\rho_0}{M_{0,0}} M_{0,1} - 0\cdot M_{1,1}$.

We now show these conditions are necessary. Suppose that $\rho_0 > M_{0,0}$, then observe that any $x_0, x_1$ satisfying the two constraints must satisfy $0\leq x_1M_{1,0} \leq x_0 M_{0,0} - \rho_0$, so $x_0 \geq \tfrac{\rho_0}{M_{0,0}}$. As a result of our assumption, we have that $x_0 > 1$. Contradiction.

So assume now that $\rho_0 \leq M_{0,0}$ yet $\rho_0 M_{0,1} > \rho_1M_{0,0}$. Any $x_0, x_1$ satisfying the two constraints must also satisfy
\[ x_0 \tfrac {M_{0,1}}{M_{1,1}} -\tfrac {\rho_1}{M_{1,1}} \leq x_1 \leq x_0 \tfrac {M_{0,0}}{M_{1,0}} - \tfrac {\rho_0}{M_{1,0}}\] which, using our assumption, yields
\[ x_0 \tfrac {M_{0,1}M_{1,0} - M_{0,0}M_{1,1}}{M_{1,1}M_{1,0}}  \leq \tfrac {\rho_1}{M_{1,1}} - \tfrac {\rho_0}{M_{1,0}} < \rho_0 \left( \tfrac {M_{0,1}} {M_{0,0}M_{1,1}} - \tfrac 1 {M_{1,0}} \right) = \rho_0 \tfrac {M_{1,0}M_{0,1}-M_{0,0}M_{1,1}} {M_{0,0}M_{1,0}M_{1,1}} \]  
so we have $x_0 < \tfrac {\rho_0}{M_{0,0}}$. Contradiction.
\end{proof}

\ifx\fullversion\undefined
\paragraph{Conditions under which $B$'s BNE strategy is $P_2$.} 
\else
\subsubsection{Conditions under which $B$'s BNE strategy is $P_2$.}
\fi
As a response to this strategy, $A$'s best response is to set $x_1=y_0=0$ and $x_0=1$, while $y_1$ is to be determined. So $B$'s utility functions are
\begin{eqnarray*}
\textrm{For } \type=0 : && p(\rho_0 - M_{0,0})+(1-p)y_1M_{1,0} \cr
\textrm{For } \type=1 : && q(\rho_1-y_1M_{1,1})+(1-q) M_{0,1} 
\end{eqnarray*}
Therefore, in order for $B$ to not have any incentive to deviate, $A$ should set $y_1$ s.t $\rho_0-M_{0,0}= y_1 M_{1,0}$ and $\rho_1-y_1M_{1,1} \geq M_{0,1}$. 

\begin{proposition}
\label{clm:conditions_for_P2}
There exists a $y_1\in [0,1]$ satisfying both $\rho_0-M_{0,0}= y_1 M_{1,0}$ and $\rho_1-y_1M_{1,1} \geq M_{0,1}$ iff $M_{0,0}\leq \rho_0 \leq M_{0,0}+M_{1,0}$ and $(\rho_0 -M_{0,0}){M_{1,1}} \leq \left(\rho_1-M_{0,1}\right){M_{1,0}}$.
\end{proposition}
\begin{proof}
Clearly, the only $y_1$ that can satisfy both constraints is $y_1 = \tfrac {\rho_0-M_{0,0}}{M_{1,0}}$, and we therefore must have that $M_{0,0} \leq \rho_0 \leq M_{0,0}+M_{1,0}$. We also need to verify that indeed the inequality holds in the right direction. I.e., to have $(\rho_0 -M_{0,0})\tfrac {M_{1,1}}{M_{1,0}} \leq \rho_1-M_{0,1}$. Clearly, if those two conditions hold then $y_1$ defined as above satisfy the required.
\end{proof}
\newtheorem{observation}[theorem]{Observation}
\begin{observation}
If we have that $(\rho_0 -M_{0,0}){M_{1,1}} \leq \left(\rho_1-M_{0,1}\right){M_{1,0}}$ then also $\tfrac {\rho_0}{\rho_1} < \tfrac {M_{1,0}} {M_{1,1}}$. 
\end{observation}
\begin{proof}
$\rho_0M_{1,1} - M_{0,0}M_{1,1} \leq \rho_1M_{1,0}-M_{0,1}M_{1,0} < \rho_1M_{1,0} - M_{0,0}M_{1,1} \Rightarrow \rho_0 M_{1,1} < \rho_1M_{1,0}$.
\end{proof}

\ifx\fullversion\undefined
\paragraph{Conditions under which $B$'s BNE strategy is $P_3$.} 
\else
\subsubsection{Conditions under which $B$'s BNE strategy is $P_3$.}
\fi
Should $B$ play $(1,0)$, then we have that $A$ only sees the $\hat\type=0$ signal and always opts out (i.e. $x_0=x_1=0$). However, in order to prevent $B$ from deviating, $A$ needs to commit to a $y_0, y_1$ that leave $B$ preferring not to deviate from $(1,0)$. $B$'s utility functions are
\begin{eqnarray*}
\textrm{For } \type=0 : && p\rho_0 + (1-p)(-y_0M_{0,0}+y_1M_{1,0}) \cr
\textrm{For } \type=1 : && q(\rho_1 +y_0 M_{0,1}- y_1M_{1,1}) 
\end{eqnarray*}
Therefore, in order for $B$ to not have any incentive to deviate, $A$ should set $y_0,y_1$ s.t $\rho_0\geq -y_0M_{0,0} + y_1 M_{1,0}$ and $\rho_1+y_0M_{0,1}-y_1M_{1,1}\leq 0$. 
\begin{proposition}
\label{clm:conditions_for_P3}
There exist $y_0, y_1\in[0,1]$ satisfying both $\rho_0\geq -y_0M_{0,0} + y_1 M_{1,0}$ and $\rho_1\leq -y_0M_{0,1}+y_1M_{1,1}$ iff $\rho_1 \leq M_{1,1}$ and $\rho_0 M_{1,1} \geq \rho_1 M_{1,0}$.
\end{proposition}
The proof is completely analogous to the proof of Proposition~\ref{clm:conditions_for_P1}. \cut{
\begin{proof}
To see that these conditions are sufficient, assume that $\rho_1 \leq M_{1,1}$ and $\rho_0 M_{1,1} \geq \rho_1 M_{1,0}$. Then we can set $y_1 = \tfrac {\rho_1}{M_{1,1}}$ and $y_0=0$. Clearly, both $y_0,y_1 \in [0,1]$. We can check and see that indeed $\rho_1 \leq \tfrac {\rho_1} {M_{1,1}} M_{1,1} - 0 \cdot M_{0,1}$ and $\rho_0 \geq \tfrac {\rho_1}{M_{1,1}} M_{1,0} - 0\cdot M_{0,0}$.

We now show these conditions are necessary. Suppose that $\rho_1 > M_{1,1}$, then observe that any $y_0, y_1$ satisfying the two constraints must satisfy $0\leq y_0M_{0,1} \leq y_1 M_{1,1} - \rho_1$, so $y_1 \geq \tfrac{\rho_1}{M_{1,1}}$. As a result of our assumption, we have that $y_1 > 1$. Contradiction.

So assume now that $\rho_1 \leq M_{1,1}$ yet $\rho_0 M_{1,1} < \rho_1M_{1,0}$. Any $y_0, y_1$ satisfying the two constraints must also satisfy
\[ y_1 \tfrac {M_{1,0}}{M_{0,0}} -\tfrac {\rho_0}{M_{0,0}} \leq y_0 \leq y_1 \tfrac {M_{1,1}}{M_{0,1}} - \tfrac {\rho_1}{M_{0,1}}\] which, using our assumption, yields
\[ y_1 \tfrac {M_{0,1}M_{1,0} - M_{0,0}M_{1,1}}{M_{0,0}M_{0,1}}  \leq \tfrac {\rho_0}{M_{0,0}} - \tfrac {\rho_1}{M_{0,1}} < \rho_1 \left( \tfrac {M_{1,0}} {M_{0,0}M_{1,1}} - \tfrac 1 {M_{0,1}} \right) = \rho_1 \tfrac {M_{1,0}M_{0,1}-M_{0,0}M_{1,1}} {M_{0,0}M_{0,1}M_{1,1}} \]  
so we have $y_1 < \tfrac {\rho_1}{M_{1,1}}$. Contradiction.
\end{proof}
}
\ifx\fullversion\undefined
\paragraph{Conditions under which $B$'s BNE strategy is $P_4$.}
\else
\subsubsection{Conditions under which $B$'s BNE strategy is $P_4$.}
\fi
As a response to this strategy, $A$'s best response is to set $x_1=y_0=0$ and $y_1=1$, while $x_0$ is to be determined. So $B$'s utility functions are
\begin{eqnarray*}
\textrm{For } \type=0 : && p(\rho_0 - x_0M_{0,0})+(1-p)M_{1,0} \cr
\textrm{For } \type=1 : && q(\rho_1-M_{1,1})+(1-q)x_0 M_{0,1} 
\end{eqnarray*}
Therefore, in order for $B$ to not have any incentive to deviate, $A$ should set $x_0$ s.t $\rho_0-x_0M_{0,0}\geq M_{1,0}$ and $\rho_1-M_{1,1} = x_0M_{0,1}$. 

\begin{proposition}
\label{clm:conditions_for_P4}
There exists a $x_0\in [0,1]$ satisfying both $\rho_0-x_0M_{0,0}\geq M_{1,0}$ and $\rho_1-M_{1,1} = x_0M_{0,1}$ iff $M_{1,1}\leq \rho_1 \leq M_{1,1}+M_{0,1}$ and $\left(\rho_1-M_{1,1}\right)M_{0,0} \leq \left(\rho_0-M_{1,0}\right)M_{0,1}$.
\end{proposition}
The proof is analogous to the proof of Proposition~\ref{clm:conditions_for_P2}.
\cut{
\begin{proof}
Clearly, the only $x_0$ that can satisfy both constraints is $x_0 = \tfrac {\rho_1-M_{1,1}}{M_{0,1}}$, and we therefore must have that $M_{1,1} \leq \rho_1 \leq M_{1,1}+M_{0,1}$. We also need to verify that indeed the inequality holds in the right direction. I.e., to have $(\rho_1 -M_{1,1})\tfrac {M_{0,0}}{M_{0,1}} \leq \rho_0-M_{1,0}$. This two conditions are equivalent to having the above-defined $x_0$ lie in the $[0,1]$-interval and satisfy both constraints.
\end{proof}
\begin{observation}
If we have $\left(\rho_1-M_{1,1}\right)M_{0,0} \leq \left(\rho_0-M_{1,0}\right)M_{0,1}$  then also $\rho_1 M_{0,0} < \rho_0 M_{0,1}$.
\end{observation}
\begin{proof}
$\rho_1 M_{0,0} - M_{0,0}M_{1,1} \leq \rho_0 M_{0,1} - M_{0,1}M_{1,0} < \rho_0 M_{0,1}-M_{0,0}M_{1,1} \Rightarrow \rho_1 M_{0,0} < \rho_0 M_{0,1}$.
\end{proof}
}
\ifx\fullversion\undefined
\paragraph{Conditions under which $B$'s BNE strategy is $P_5$.} 
\else
\subsubsection{Conditions under which $B$'s BNE strategy is $P_5$.}
\fi
As this point lies on the intersection of $l_1$ and $l_2$, then $A$'s best response to this strategy is to set $x_1=y_0=0$. Thus,$B$'s utility functions are
\begin{eqnarray*}
\textrm{For } \type=0 : && p(\rho_0 - x_0M_{0,0})+(1-p)y_1M_{1,0} \cr
\textrm{For } \type=1 : && q(\rho_1-y_1M_{1,1})+(1-q)x_0 M_{0,1} 
\end{eqnarray*}
It is therefore up to $A$ to pick $x_0$ and $y_1$ that satisfy both equalities $\begin{pmatrix} \rho_0 \cr \rho_1\end{pmatrix} = \begin{pmatrix}
M_{0,0} & M_{1,0} \cr M_{0,1} & M_{1,1} 
\end{pmatrix}  \begin{pmatrix}x_0 \cr y_1 \end{pmatrix}$
Cramer's formula give that the solution to this system is 
\begin{equation}
\begin{pmatrix}
x_0 \cr y_1
\end{pmatrix} = {\displaystyle \frac 1 {M_{1,0}M_{0,1} - M_{0,0}M_{1,1}}} \begin{pmatrix}
-M_{1,1} &M_{1,0}  \cr
M_{0,1} &-M_{0,0}
\end{pmatrix}\begin{pmatrix}
\rho_0 \cr \rho_1
\end{pmatrix}\label{eq:A_NE_strategy_P5}\end{equation}
In order for $x_0,y_1$ to be in the range $[0,1]$, we therefore must have that: (i) $\rho_0 M_{1,1} \leq \rho_1 M_{1,0}$, (ii) $(\rho_1-M_{0,1})M_{1,0} \leq (\rho_0-M_{0,0})M_{1,1} $, (iii) $\rho_1M_{0,0} \leq \rho_0M_{0,1}$, and (iv) $(\rho_0-M_{1,0})M_{0,1} \leq (\rho_1-M_{1,1})M_{0,0}$.
In other words:
\begin{align}
\label{eq:conditions_for_P5}
& 0 \leq \rho_1M_{1,0} - \rho_0M_{1,1} \leq M_{0,1}M_{1,0}-M_{0,0}M_{1,1} \cr & 0 \leq \rho_0M_{0,1} - \rho_1M_{0,0} \leq   M_{0,1}M_{1,0}-M_{0,0}M_{1,1}
\end{align}

\ifx\fullversion\undefined
\paragraph{Summarizing.} Below we give the various conditions under which each of the points may be a BNE:
\begin{table}[hbt]
\centering
\begin{tabular}{ | c | c | c | c |}
\hline
Case & Condition & $A$'s Strategy & $B$'s strategy \cr
No.& & (always: $x_1=y_0=0$) &\cr
\hline
$1$& $\rho_0 \geq M_{0,0}+M_{1,0}$ and~~~ $\rho_1\geq M_{0,1}+M_{1,1}$ & $(x_0,y_1)=(1,1)$ & $(1,1)$\cr\hline
$2$& $\rho_0 \leq M_{0,0}$ and~~~ $\tfrac{\rho_0}{\rho_1}\leq \tfrac{M_{0,0}}{M_{0,1}}$ & $(x_0,y_1) = (\tfrac {\rho_0}{M_{0,0}},0)$ & $(0,1)$ \cr\hline
$3$& $\rho_1 \leq M_{1,1}$ and~~~ $\tfrac{\rho_0}{\rho_1} \geq \tfrac{M_{1,0}}{M_{1,1}}$ & $(x_0,y_1)=(0,\tfrac{\rho_1}{M_{1,1}})$ & $(1,0)$ \cr\hline
$4$& $0 \leq \rho_0-M_{0,0}\leq M_{1,0}$ & $(x_0,y_1) = (1,\tfrac {\rho_0-M_{0,0}}{M_{1,0}})$ & $P_2$ \cr
& $\rho_1M_{1,0}-\rho_0M_{1,1} \geq M_{0,1}M_{1,0} -M_{0,0}{M_{1,1}}$ & &\cr\hline
$5$& $0 \leq \rho_1 - M_{1,1} \leq M_{0,1}$ & $(x_0,y_1)=(\tfrac {\rho_1-M_{1,1}}{M_{0,1}},1)$ & $P_4$ \cr
& $\rho_0M_{0,1}-\rho_1M_{0,0} \geq M_{0,1}M_{1,0}-M_{0,0}M_{1,1}$ & &\cr\hline
$6$& $0 \leq \rho_1M_{1,0} - \rho_0M_{1,1} \leq M_{0,1}M_{1,0}-M_{0,0}M_{1,1}$ & See Eq.~\eqref{eq:A_NE_strategy_P5}& $P_5$ \cr
& $0 \leq \rho_0M_{0,1} - \rho_1M_{0,0} \leq   M_{0,1}M_{1,0}-M_{0,0}M_{1,1}$ & & \cr\hline
\end{tabular}
\caption{\label{apx_tab:Conditions_for_NE}The various conditions under which any of the isolated points are BNE. Recall: $P_2 = (1 - \frac {D_1M_{1,1}}{D_0 M_{1,0}} , 1)$ is the intersection point of the $l_2$-line and the $q=1$ line, $P_4 = (1, 1-\frac {D_0 M_{0,0}}{D_1M_{0,1}})$ is the intersection point of the $l_1$ line and the $p=1$ line, and $P_5 = \left( \frac {D_0D_1M_{0,1}M_{1,0} - D_1^2 M_{0,1}M_{1,1}} {D_0D_1 M_{0,1}M_{1,0} - D_0 D_1 M_{0,0}M_{1,1}} , \frac {D_0D_1M_{0,1}M_{1,0} - D_0^2 M_{0,0}M_{1,0}} {D_0D_1 M_{0,1}M_{1,0} - D_0 D_1 M_{0,0}M_{1,1}} \right)$ is at the intersection of the $l_1$- and $l_2$-lines. Recall also that $M_{0,1}M_{1,0}>M_{0,0}M_{1,1}$.}
\end{table}
\else
\subsubsection{Summarizing.} 
Table~\ref{tab:Conditions_for_NE} summarized the $6$ conditions under which each point is a BNE.
\fi
It is easy to verify that the condition of Theorem~\ref{thm:mixed_BNE_under_condition} is precisely condition $6$, underwhich $B$'s BNE strategy unique (the point $P_5$) and it is mixed.

\subsection{Proof of Theorem~\ref{thm:NE_which_is_RR}: The Uniqueness of $B$'s BNE Strategy}
\label{subsec:characterizing_NE}

So far we have introduced conditions for the existence of various BNEs. In this section, our goal is to show that the above analysis gives a complete description of the game. That is, to show that the cases detailed in Table~\ref{tab:Conditions_for_NE} span all potential values the parameters of the game may take, and furthermore (modulo cases of equality between parameters) they are also mutually exclusive.

\begin{lemma}
\label{lem:conditions_are_exclusive}
Assume that the parameters of the game (i.e. $\rho_0,\rho_1$ and the entries of $M$) satisfy one of the $6$ conditions detailed in Table~\ref{tab:Conditions_for_NE}  with strict inequalities. Then no other condition in Table~\ref{tab:Conditions_for_NE} holds simultaneously. In other words, the conditions in Table~\ref{tab:Conditions_for_NE} are mutually exclusive (excluding equalities). As the conditions are mutually exclusive it means that under the condition specified in Theorem~\ref{thm:mixed_BNE_under_condition}, the game has a unique BNE  -- as specified by case $6$ in Table~\ref{tab:Conditions_for_NE}.
\end{lemma}
\begin{proof}
We traverse the $6$ cases, showing that if case $i$  holds with strict inequalties then some other case $j>i$ cannot hold.
\begin{description}
\item [Case $1$.] Clearly, if the conditions of case $1$ hold, then the conditions of cases $2,3,4$ and $5$ cannot hold. To see that the conditions of case $6$ cannot hold, we argue that the condition $\max\{\rho_1M_{1,0} - \rho_0M_{1,1}, \rho_0M_{0,1} - \rho_1M_{0,0}\} \leq M_{0,1}M_{1,0}-M_{0,0}M_{1,1}$ implies that both $\rho_0 \leq M_{0,0}+M_{1,0}$ and $\rho_1\leq M_{0,1}+M_{1,1}$. This claim follows from the inequalities
\begin{eqnarray*}
\rho_0 (M_{0,1}M_{1,0}-M_{0,0}M_{1,1} ) &&= M_{0,0} (\rho_1M_{1,0} - \rho_0M_{1,1}) + M_{1,0}(\rho_0M_{0,1} - \rho_1M_{0,0}) \cr &&\leq (M_{0,0}+M_{1,0})(M_{0,1}M_{1,0}-M_{0,0}M_{1,1}) \cr
\rho_1 (M_{0,1}M_{1,0}-M_{0,0}M_{1,1} ) &&= M_{0,1} (\rho_1M_{1,0} - \rho_0M_{1,1}) + M_{1,1}(\rho_0M_{0,1} - \rho_1M_{0,0}) \cr &&\leq (M_{0,1}+M_{1,1})(M_{0,1}M_{1,0}-M_{0,0}M_{1,1})
\end{eqnarray*}
\item [Case $2$.] Clearly, the conditions of case $2$ cannot hold simultaneously with the conditions of cases $4$ and $6$. To exclude the other cases, observe that using our favorite inequality $\tfrac {M_{0,0}}{M_{0,1}} < \tfrac {M_{1,0}}{M_{1,1}}$, we have that the condition $\tfrac{\rho_0}{\rho_1} \leq \tfrac {M_{0,0}}{M_{0,1}}$ implies that $\tfrac {\rho_0}{\rho_1} < \tfrac {M_{1,0}}{M_{1,1}}$. Hence case $3$ cannot hold, and neither does case $5$ (using again the fact that $M_{0,1}M_{1,0}>M_{0,0}M_{1,1}$).
\item [Case $3$.] This case is symmetric to case $2$ --- since $M_{0,1}M_{1,0} - M_{0,0}M_{1,1}>0$ then case $3$ rules out case $4$ (and the fact it cannot hold simultaneously with cases $5$ and $6$ is obvious).
\item [Case $4$.] Clearly, case $6$ cannot hold together with case $4$. To show that case $5$ cannot hold too, we claim that if both $\rho_1M_{1,0}-\rho_0M_{1,1} \geq M_{0,1}M_{1,0} -M_{0,0}{M_{1,1}}$ and $\rho_0M_{0,1}-\rho_1M_{0,0} \geq M_{0,1}M_{1,0}-M_{0,0}M_{1,1}$ hold, then $\rho_0 \geq M_{0,0}+M_{1,0}$. This holds because the two inequalities imply
\begin{eqnarray*}
&& \rho_1 \leq ( \rho_0-M_{1,0} )\tfrac {M_{0,1}}{M_{0,0}} + M_{1,1} \textrm{ and }  \rho_1 \geq (\rho_0-M_{0,0})\tfrac{M_{1,1}}{M_{1,0}} + M_{0,1} \cr
& \Rightarrow &\rho_0 \left( \tfrac{M_{0,1}}{M_{0,0}} - \tfrac {M_{1,1}}{M_{1,0}}  \right) \geq M_{0,1}-M_{1,1} + \tfrac {M_{1,0}M_{0,1}}{M_{0,0}} - \tfrac { M_{0,0}M_{1,1}  }{M_{1,0}} \cr
& \Rightarrow & \rho_0 \geq \frac {  M_{0,0}M_{0,1}M_{0,1} - M_{0,0}M_{1,0}M_{1,1} + M_{0,1}M_{1,0}^2 - M_{0,0}^2M_{1,1}  }{ M_{0,1}M_{1,0} - M_{0,0}M_{1,1}} = M_{0,0}+M_{1,0}
\end{eqnarray*}
\item [Case $5$.] Clearly, cases $5$ and $6$ cannot hold simultaneously.
\end{description}
\end{proof}

\begin{lemma}
\label{lem:conditions_are covering}
Any choice of parameters for $\rho_0,\rho_1$ and the entries of $M$ satisfies at least one of the $6$ cases detailed in Table~\ref{tab:Conditions_for_NE}.
\end{lemma}
\begin{proof}
First, suppose $\rho_0 \geq M_{0,0}+M_{1,0}$. We claim that in this case, the value of $\rho_1$ determines which case holds.
\begin{itemize}
\item If $\rho_1 \leq M_{1,1}$ then case $3$ holds, since obviously $ M_{1,1}\tfrac{\rho_0}{M_{1,0}} > M_{1,1} \geq \rho_1$.
\item If $M_{1,1} < \rho_1 \leq M_{0,1}+M_{1,1}$ then case $5$ holds since 
\[\rho_0 M_{0,1} - \rho_1 M_{0,0} \geq (M_{0,0}+M_{1,0})M_{0,1}-\rho_1 M_{0,0} = M_{0,1}M_{1,0} + M_{0,0}(M_{0,1}-\rho_1) \geq M_{0,1}M_{1,0}-M_{0,0}M_{1,1}\]
\item If $\rho_1 > M_{0,1}+M_{1,1}$ then clearly case $1$ holds.
\end{itemize}
Similarly, if we have that $\rho_1 \geq M_{0,1}+M_{1,1}$, then the value of $\rho_0$ determines whether case $2,4$ or $1$ hold.

We therefore assume from now on that $\rho_0 < M_{0,0}+M_{1,0}$ and $\rho_1 < M_{0,1}+M_{1,1}$. 

Suppose that $\tfrac {\rho_0}{\rho_1} \leq \tfrac{M_{0,0}}{M_{0,1}}$. 
\begin{itemize}
\item If $\rho_0 \leq M_{0,0}$ then clearly case $2$ holds.
\item If $\rho_0 \geq M_{0,0}$ then we show case $4$ holds. Observe $\tfrac {\rho_1}{\rho_0} - \tfrac{M_{1,1}}{M_{1,0}} \geq \tfrac {M_{0,1}}{M_{0,0}} - \tfrac{M_{1,1}}{M_{1,0}}$, so $\tfrac{ \rho_1M_{1,0} - \rho_0M_{1,1}}{\rho_0 M_{1,0}} \geq \tfrac{M_{0,1}M_{1,0} - M_{0,0}M_{1,1}}{M_{0,0}M_{1,0}}$. We conclude that $\rho_1M_{1,0} - \rho_0M_{1,1} \geq \tfrac {\rho_0}{M_{0,0}} (M_{0,1}M_{1,0}-M_{0,0}M_{1,1})$. So the fact that $\rho_0 \geq M_{0,0}$ implies that the conditions of case $4$ hold.
\end{itemize}
Analogously, if we assume the $\frac {\rho_0}{\rho_1} \geq \tfrac {M_{1,0}}{M_{1,1}}$, then the same line of argument shows that either case $3$ or case $5$ hold.

So now, we assume both that $\rho_0 < M_{0,0}+M_{1,0}$, $\rho_1 < M_{0,1}+M_{1,1}$ and that $\tfrac {M_{0,0}}{M_{0,1}} < \tfrac {\rho_0}{\rho_1} < \tfrac {M_{1,0}}{M_{1,1}}$. 
\begin{itemize}
\item If $\rho_1 M_{1,0} -\rho_0M_{1,1} \geq M_{0,1}M_{1,0} - M_{0,0}M_{1,1}$, we argue that case $4$ holds. This is because we have both that $\rho_1 < \rho_0 \tfrac{M_{0,1}}{M_{0,0}}$ and that $\rho_1 \geq \rho_0 \tfrac{M_{1,1}}{M_{1,0}} +M_{0,1} - \tfrac {M_{0,0}M_{1,1}}{M_{1,0}}$. Combining the two we get
\[ \rho_0\left( \tfrac {M_{0,1}}{M_{0,0}} - \tfrac {M_{1,1}}{M_{1,0}}\right) > \tfrac {M_{0,1}M_{1,0}-M_{0,0}M_{1,1}}{M_{1,0}} ~~\Rightarrow ~~ \rho_0 > M_{0,0}\]
\item If $\rho_0 M_{0,1} - \rho_1M_{0,0} \geq M_{0,1}M_{1,0} - M_{0,0}M_{1,1}$ then we are in the analogous case, and  we can show, using the inequality $\tfrac {\rho_0}{\rho_1} < \tfrac {M_{1,0}}{M_{1,1}}$, that $\rho_1 > M_{1,1}$.
\end{itemize}
This leaves us with the case that $\rho_0 < M_{0,0}+M_{1,0}$, $\rho_1 < M_{0,1}+M_{1,1}$, $\tfrac {M_{0,0}}{M_{0,1}} < \tfrac {\rho_0}{\rho_1} < \tfrac {M_{1,0}}{M_{1,1}}$ and also $\rho_1 M_{1,0} -\rho_0M_{1,1} < M_{0,1}M_{1,0} - M_{0,0}M_{1,1}$ and $\rho_0 M_{0,1} - \rho_1M_{0,0} < M_{0,1}M_{1,0} - M_{0,0}M_{1,1}$. This is precisely case $6$.
\end{proof}
}
\ifx \cameraready \undefined
Recall, in addition to the conditions specifically stated in Case $6$ in Table~\ref{tab:Conditions_for_NE}, we also require that $D_0^2M_{0,0}M_{1,0} = D_1^2M_{0,1}M_{1,1}$ in order for the two types of agent $B$ to play Randomized Response. In other words, this condition implies that $B$'s BNE strategy, represented by the point \[P_5=\big( \frac {D_0D_1M_{0,1}M_{1,0} - D_1^2 M_{0,1}M_{1,1}} {D_0D_1 M_{0,1}M_{1,0} - D_0 D_1 M_{0,0}M_{1,1}},\frac {D_0D_1M_{0,1}M_{1,0} - D_0^2 M_{0,0}M_{1,0}} {D_0D_1 M_{0,1}M_{1,0} - D_0 D_1 M_{0,0}M_{1,1}} \big)\] lies on the $p=q$ line.
\cut{
\begin{proposition}
\label{clm:conditions_for_RR}
In a BNE of Case $6$, where $B$ plays a strictly randomized strategy (i.e. $p,q \in (0,1)$ ), we have that $p=q$ iff $\tfrac {D_0 M_{1,0}}{D_1M_{1,1}} = \tfrac {D_1 M_{0,1}}{D_0 M_{0,0}}$.
\end{proposition}
\begin{proof}
The coordinates of $P_5$ are $\left( \frac {D_0D_1M_{0,1}M_{1,0} - D_1^2 M_{0,1}M_{1,1}} {D_0D_1 M_{0,1}M_{1,0} - D_0 D_1 M_{0,0}M_{1,1}} , \frac {D_0D_1M_{0,1}M_{1,0} - D_0^2 M_{0,0}M_{1,0}} {D_0D_1 M_{0,1}M_{1,0} - D_0 D_1 M_{0,0}M_{1,1}} \right)$, so the proof follows immediately.
\end{proof}
}
And so, in this case the $B$ agent plays a Randomized Response strategy that preserves $\epsilon$-differential privacy for 
$\epsilon = \ln(\tfrac p {1-q}) =\ln\left( \tfrac {D_1M_{0,1}} {D_0M_{0,0}}\right)$.
Observe that this value of $\epsilon$ is \emph{independent} from the value of the coupon (i.e., from $\rho_0$ and $\rho_1$). This is due to the nature of BNE in which an agent plays her Nash-strategy in order to make her opponent indifferent between various strategies rather than maximizing her own utility. Therefore, the coordinates of $P_5$ are such that they make agent $A$ indifferent between opting out and playing $x_0=1$ (or opting out and $y_1=1$). Since the utility function of $A$ is independent of $\rho_0,\rho_1$, we have that perturbing the values of $\rho_0,\rho_1$ does not affect the coordinates of $P_5$. (Yet, perturbing the values of $\rho_0,\rho_1$ does affect the various relations between the parameters of the game, and so it may determine which of the $6$ cases in Table~\ref{tab:Conditions_for_NE} holds.)

\else

Recall, in addition to the conditions specifically stated in Equation~\eqref{eq:condition_of_randomized_BNE}, we also require that $D_0^2M_{0,0}M_{1,0} = D_1^2M_{0,1}M_{1,1}$ in order for the two types of agent $B$ to play Randomized Response. In other words, the feasibility condition in Equation~\eqref{eq:condition_of_randomized_BNE} implies that $B$'s BNE strategy, denoted by $p^* = \Pr[\sigma_B^*(0)=0]$ and $q^* = \Pr[\sigma_B^*(1)=1]$, is given by 
\[(p^*,q^*)=\big( \frac {D_0D_1M_{0,1}M_{1,0} - D_1^2 M_{0,1}M_{1,1}} {D_0D_1 M_{0,1}M_{1,0} - D_0 D_1 M_{0,0}M_{1,1}},\frac {D_0D_1M_{0,1}M_{1,0} - D_0^2 M_{0,0}M_{1,0}} {D_0D_1 M_{0,1}M_{1,0} - D_0 D_1 M_{0,0}M_{1,1}} \big)\] The additional condition of $D_0^2M_{0,0}M_{1,0} = D_1^2M_{0,1}M_{1,1}$ implies therefore that $p^*=q^*$.
And so, in this case the $B$ agent plays a Randomized Response strategy that preserves $\epsilon$-differential privacy for 
$\epsilon = \ln(\tfrac {p^*} {1-{q^*}}) =\ln\left( \tfrac {D_1M_{0,1}} {D_0M_{0,0}}\right)$.
Observe that this value of $\epsilon$ is \emph{independent} from the value of the coupon (i.e., from $\rho_0$ and $\rho_1$). This is due to the nature of BNE in which an agent plays her Nash-strategy in order to make her opponent indifferent between various strategies rather than maximizing her own utility. Therefore, the coordinates $(p^*,q^*)$ are such that they make agent $A$ indifferent between several pure strategies. And since the utility function of $A$ is independent of $\rho_0,\rho_1$, we have that perturbing the values of $\rho_0,\rho_1$ does not affect the coordinates $(p^*,q^*)$. (Yet, perturbing the values of $\rho_0,\rho_1$ does affect the various relations between the parameters of the game, and so it may determine which of the $6$ feasibility conditions does in fact hold.)
\fi

\subsection{Proof of Theorem~\ref{thm:NE_which_is_RR}: Finding a BNE Strategy for $B$}
\apxOptOut

\ifx \cameraready \undefined

\newpage
\section{Conclusions and Future Directions}
\label{sec:conclusions}
Our work is a first attempt at exposing and reconciling the competing conclusions of two different approaches to the same challenge: the theory of privacy-aware agents (where privacy loss is modeled using differential privacy), and the behavior of standard utility-maximizing agents once they explicitly assess future losses from having their behavior in the current game publicly exposed.
While the canonical privacy-aware agent randomizes her strategy, we show that different explicit privacy losses cause very different behavior among agents. This is best illustrated with the game studied in Section~\ref{sec:matching_pennies} (Theorem~\ref{thm:matching-pennies-continuous-valuations}). In that game, agents assess their future loss and their behavior is therefore quite simple: if the current gain is greater than the future loss, their behavior is to truthfully  report their type; otherwise they lie and report the opposite type. We believe this simple rule explains real-life phenomena, such as people trying to hide their medical condition from the general public while truthfully answering a doctor's questions.
\footnote{I'm likely gain little and potentially lose a lot from revealing my medical history to a random person, whereas I am likely to gain a lot from truthfully reporting my medical history to a doctor.}


Observe however that in all the games we analyzed, we still have not pinned down a game in which the behavior of a non-privacy aware agent \emph{fully} mimics the behavior of a privacy-aware agent. Privacy-aware agents' behavior is, after a fashion, quite reasonable. They trade-off between the value of the coupon they get and the amount of privacy (or change in belief) they are willing to risk. Naturally, the higher the value of the coupon, the more privacy they are willing to risk.
In contrast, in the game discussed in Section~\ref{sec:opt-out-possible}, even under settings where $B$'s BNE strategy $\sigma_B^*$ is randomized and satisfies $\Pr[\sigma_B^*(0)]=\Pr[\sigma_B^*(1)]$, we don't see a continuous change in $B$'s behavior based on the value of the coupon.  Changing solely the value of the coupon while keeping all other parameters the same, we see that $B$ plays the same BNE strategy, whereas $A$'s BNE strategy continuously changes.


It would be interesting to pursue this line of work further, by studying more complex games. In particular, we propose the following scenario, which resembles the standard narrative in differential-privacy literature and should provide a complementary approach to the ``sensitive surveyor'' problem~\cite{GhoshR11,NissimOS12,RothS12,NissimVX14,GhoshLRS14}. Suppose that the signal that $B$ sends is not for a type of coupon that gives $B$ an immediate and fixed reward, but rather a response of $B$ to a survey question. That is, suppose $B$ interacts with a benevolent data curator that wishes to learn the distribution of type-$0$ and type-$1$ agents in the population and $B$ may benefit from the effect of curator's analysis. (For example, the data curator may ask people with a certain disease about their exposure to some substance.) In such a case, $B$'s utility is a function of the curator's ability to well approximate the true answer. In addition to the potential gain, there is also potential loss, based on $B$'s concerns about her private information being publicly exposed. What formulation of this privacy loss results in $B$ playing according to a Randomized Response strategy? What explicit formulation of privacy loss causes $B$ to truthfully report her type knowing that $A$'s will data be published using an $\epsilon$-differential private mechanism?
\fi

\ifx \fullversion \undefined \vspace{-0.6cm} \fi
\section*{Acknowledgments}
\ifx \fullversion \undefined \vspace{-0.3cm} \fi
We would like to thank Kobbi Nissim for many helpful discussions and helping us in initiating this line of work.
\ifx \fullversion \undefined \vspace{-0.5cm} \fi

\ifx \fullversion \undefined
{\small
\bibliographystyle{splncs03}
\bibliography{paper}
}
\else
\bibliographystyle{alpha}
\bibliography{paper}
\fi

\ifx \cameraready \undefined
\appendix

\ifx \fullversion \undefined
\spnewtheorem*{apxthm}{Theorem}{\bfseries}{\itshape}
\else
\newtheorem*{apxthm}{Theorem}{\bfseries}{\itshape}

\fi

\ifx \fullversion \undefined
\section{Privacy Aware Agents.}
\label{apx_sec:privacy_aware_agent}
\begin{apxthm}[Theorem~\ref{thm:behavior_privacy_aware} restated]
\behaviorPrivacyAware
\end{apxthm}
\PAA
\fi

\section{Missing Proofs -- Coupon Game with Proper Scoring Rules}
\label{apx_sec:proper_scoring_rules}

\backgroundProperScoringRules

\ifx\fullversion\undefined
\subsection{Proof of Theorem~\ref{thm:BNE_scoring_rules}}
\begin{apxthm}[Theorem~\ref{thm:BNE_scoring_rules} restated]
\BNEScoringRules
\end{apxthm}
\proofTheoremProperScoringRule
\else
\specificScoringRules

\ifx \fullversion \undefined
\section{Missing Proofs -- Coupon Game with Identity Matrix Payments}
\label{apx_sec:matching_pennies}

\subsection{Proof of Theorem~\ref{thm:coupon_matching_pennies}}
\begin{apxthm}[Theorem~\ref{thm:coupon_matching_pennies} restated]
\couponMatchingPennies
\end{apxthm}
\proofTheoremIdentityMatrix

\subsection{Coupon Game with Continuous Coupon Valuations}
\label{apx_subsec:continuous_valuations}

We now consider the same coupon game with payments given in the form of the identity matrix, but under a different setting. Whereas before we assumed the valuations that the two types of $B$ agents have for the coupon are fixed (and known in advance), we now assume they are not fixed. In this section we assume the existence of a continuous prior over $\rho$, where each type $\type \in \{0,1\}$ has its own prior, so $\CDF_0(x) \stackrel {\rm def} = \Pr[\rho < x ~|~ t=0]$ with an analogous definition of $\CDF_1(x)$. We use $\CDF_B$ to denote the cumulative distribution function of the prior over $\rho$ (i.e., $\CDF_B(x) = \Pr[\rho < x] = D_0\CDF_0(x)+D_1\CDF_1(x)$). We assume the $\CDF$ is continuous and so $\Pr[\rho=y]=0$ for any $y$. Given any $z \geq 0$ we denote $\CDF_B^{-1}(z)$ the set $\{y :~ \CDF_B(y)=z\}$. We proceed by proving Theorem~\ref{thm:matching-pennies-continuous-valuations}.
\begin{apxthm}[Theorem~\ref{thm:matching-pennies-continuous-valuations} restated]
\matchingPenniesContinuousValuations
\end{apxthm}
\continuousCouponValuations
\fi

\ifx\fullversion\undefined
\section{Missing Proofs -- Coupon Game with an Opt-Out Strategy}
\label{apx_sec:opt-out}
\begin{apxthm}[Theorem~\ref{thm:NE_which_is_RR} restated]
\NEWhichIsRR
\end{apxthm}
\apxOptOut
\fi

\fi
\end{document}